\pdfoutput=1
\documentclass[acmsmall,screen,nonacm,appendix]{acmart}

\usepackage{xspace}
\usepackage[capitalize]{cleveref}
\crefname{line}{Line}{Lines}
\crefrangelabelformat{line}{#3#1#4--#5#2#6}

\usepackage{microtype}
\usepackage{wrapfig}
\usepackage{natbib}
\usepackage[inline]{enumitem}
\usepackage{mathpartir}
\let\ifdraft\iffalse
\usepackage{style/notation}
\usepackage{style/appendix}

\usepackage{multicol}
\usepackage{thmtools,thm-restate}

\usepackage{style/tikzstyle}

\AtEndPreamble{\theoremstyle{acmdefinition}
  \newtheorem{remark}[theorem]{Remark}
}

\ifappendix
  \setcopyright{none}
\else
  \setcopyright{rightsretained}
  \acmPrice{}
  \acmDOI{10.1145/3571219}
  \acmYear{2023}
  \copyrightyear{2023}
  \acmSubmissionID{popl23main-p144-p}
  \acmJournal{PACMPL}
  \acmVolume{7}
  \acmNumber{POPL}
  \acmArticle{26}
  \acmMonth{1}
  \received{2022-07-07}
  \received[accepted]{2022-11-07}
\fi

\additionalmaterial[The extended version of this paper~\cite{fullversion}]{\cite{fullversion}}

\ifappendix
  \setcopyright{none}
\fi

\begin{document}
\title{The Path to Durable Linearizability}
\ifappendix
  \subtitle{(Extended Version)}
\fi

\author{Emanuele D'Osualdo}
\email{dosualdo@mpi-sws.org}
\orcid{0000-0002-9179-5827}
\affiliation{\institution{MPI-SWS}
  \city{Saarland Informatics Campus}
  \country{Germany}
}

\author{Azalea Raad}
\email{azalea.raad@imperial.ac.uk}
\affiliation{\institution{Imperial College London}
  \country{UK}
}

\author{Viktor Vafeiadis}
\email{viktor@mpi-sws.org}
\affiliation{\institution{MPI-SWS}
  \city{Saarland Informatics Campus}
  \country{Germany}
}

\begin{abstract}

There is an increasing body of literature proposing new and efficient
persistent versions of concurrent data structures ensuring that a consistent
state can be recovered after a power failure or a crash.
Their correctness is typically stated in terms of \emph{durable linearizability} (DL),
which requires that individual library operations appear to be executed
atomically in a sequence consistent with the real-time order and, moreover,
that recovering from a crash return a state corresponding to a prefix of that
sequence.
Sadly, however, there are hardly any formal DL proofs, and those that do exist
cover the correctness of rather simple persistent algorithms on specific
(simplified) persistency models.

In response, we propose a general, powerful, modular, and incremental proof technique that
can be used to guide the development and establish DL.
Our technique is
(1) \emph{general}, in that it is not tied to a specific persistency and/or consistency model,
(2) \emph{powerful}, in that it can handle the most advanced persistent algorithms in the literature,
(3) \emph{modular}, in that it allows the reuse of an existing linearizability argument, and
(4) \emph{incremental}, in that the additional requirements for establishing DL
depend on the complexity of the algorithm to be verified.
We illustrate this technique on various versions of a persistent set,
leading to the link-free set of Zuriel et al.
 \end{abstract}

\ifnoappendix


\fi

\maketitle

\section{Introduction}
\label{sec:intro}
\newcommand{\lint}{linearizability\xspace}
\newcommand{\Lint}{Linearizability\xspace}

\emph{Non-volatile memory}~(NVM)~\cite{nvm1,nvm2,nvm3}
is a new kind of memory, which
ensures that its contents survive crashes (\eg due to power failure)
while having performance similar to that of RAM.
As such, it has generated a lot of interest in the systems community,
with an increasing body of work proposing persistent data structures
that can be restored to a consistent state after a crash.

These persistent data structures are typically adaptations of existing
\emph{linearizable} concurrent data structures, and so their
correctness is given in terms of \emph{durable \lint} (DL) \cite{durable-lin},
an extension of \lint \cite{lin} to account for crashes.
Similar to how \lint requires all operations to appear to execute atomically in some legal total order consistent with their real-time execution order,
DL requires the same to also hold for the state recovered after a crash:
it should correspond to some legal execution of a subsequence of the operations before the crash
containing at least all the operations that completed before the crash.

The adaptations required to make a concurrent data structure persistent, however, are far from trivial,
especially when the goal is to achieve optimal performance.
The reason is that hardware implementations do not persist every write to NVM immediately,
which would automatically turn any linearizable data structure to a persistent one,
but rather put them into a buffer to be persisted at some later point.
Writes in such buffers can moreover be persisted out of order, leading to \emph{weak persistency} semantics,
which is another layer of complexity above the \emph{weak memory consistency} semantics of modern CPUs.
To ensure that writes are persisted, programs can issue a special \code{flush(x)}
instruction (\eg \code{CLFLUSH x} on Intel-x86 machines),
which blocks until all pending writes to \p{x} are persisted.
Introducing an appropriate flush instruction after \emph{every} memory access (\ie both reads and writes)
can restore the sane \emph{strict persistency}~\cite{PelleyCW14} model,
where the order in which writes persist (the ``persistency'' or ``non-volatile'' order) agrees
with the order in which they become visible to other processors in the system (the ``volatile'' order).
Doing so, however, incurs a prohibitive cost, and
so programmers of persistent libraries strive to use as few flushes as possible,
which may in turn require adding auxiliary state to the algorithm or redesigning part of it.

A natural question arises.
If persistent data structures are adaptations of concurrent ones,
can we establish their correctness by reusing the correctness proof of their concurrent analogues?
In other words, is it possible to dissect the DL requirements in a way that we can reuse the invariants established by the \lint proof?
The existing literature sadly does not provide an answer to this question.
Most papers (\eg \cite{persistent-queues,SOFT:oopsla})
come with informal proof arguments in English that do not consider the intricacies of NVM semantics,
while in the few papers that come with formal proofs \cite{ptso,parm}
their arguments are highly entangled with a specific memory consistency/persistency model,
and are not able to disentangle the concurrency aspects of the proof from the persistency ones.

In this work, we show that such modularity and reuse are possible.
We present the first formal proof methodology for verifying durable \lint of persistent algorithms.
Our proof technique enjoys the following four properties:
\begin{itemize}
\item It is \emph{modular}, in that it separates out the proof obligations concerning 
\lint from those concerning persistency and from those concerning the recovery code,
thereby allowing the reuse of an existing \lint proof.

\item It is \emph{general} in that it is not tied to a specific model
like epoch persistency \cite{durable-lin}, Px86 \cite{Px86}, or PArm \cite{parm},
but supports arbitrary models with different volatile and non-volatile orders.

\item It is \emph{powerful} in that it can handle the most advanced persistent algorithms in the literature.
We will illustrate this point by applying it to produce the first formal DL proof of
the recent link-free set of \citet{SOFT:oopsla},
an algorithm, whose \lint argument is already challenging,
as it cannot be shown correct using fixed linearization points.

\item It is \emph{incremental} in the sense that the difficulty of the additional requirements for establishing DL
is proportional to the complexity of the algorithm to be verified.
\end{itemize}
Our proof technique is captured by our \masterthm{} in \cref{sec:master-thm},
which presents a detailed methodology for establishing DL following the
``linearize-first'' implementation scheme, where the effects of an operation
are first committed to volatile memory and are later persisted,
\ie the operation first reaches its linearization point and later, what we call, its persistency point.
The key idea is to add to a proof of linearizability a number of general conditions on commutativity of the operations restricting when operations may be linearized and persisted in different orders.

We also have a dual version of our \masterthm{} suitable for the data structure implementations following the ``persist-first'' approach,
\eg SOFT~\cite{SOFT:oopsla},
where persistency points precede linearization.

\paragraph{Outline}
We start, in \cref{sec:overview}, with an informal overview of our approach.
In \cref{sec:devel}, we develop a memory-model-agnostic definition of durable \lint, which is used to specify persistent libraries.
In \cref{sec:master-thm}, we present our proof technique culminating in our \masterthm{}.
Then, in \cref{sec:eval} as an extended case study,
we obtain the first formal proof of the link-free set
as an application of the \masterthm.
We conclude with discussion of related work.
All the omitted definitions and proofs, including the full verification
of the link-free set can be found in
\ifappendix Appendix.
  In \cref{sec:appendix-prelim} we give all the details on the semantic model
  and the Px86 memory model.
  \Cref{sec:full-proof-technique} presents and proves
  the \masterthm\ and \pfmasterthm{} in full generality.
  \Cref{sec:verif-linkfree} contains the full proof of our main case study.
\else the extended version of this paper~\cite{fullversion}.
\fi 

 \ifappendix\pagebreak\fi
\section{Overview}
\label{sec:overview}

Our goal is to produce a proof of durable linearizability,
starting from a proof of standard linearizability.
To make the discussion concrete, we will use
a standard set implementation based on \citet{Harris}
as a running example.

\begin{figure}
  \small
\begin{tabular}{l@{\hspace{3em}}l}
  \begin{sourcecode}[gobble=2,lineskip=-2pt]
  record Node:
    key: $\Nat \dunion \set{+\infty,-\infty}$
    nxt: $\Bool \times \Addr_\nullptr$


  def find(h, k):
    p = h
    <|_,c|> = p.nxt
    while(1):
      if c.nxt == <|0,_|>: @\label{lp:harris:find}@
        if c.key >= k:
          if p.nxt == <|0,_|>:
            return <|p,c|>
          c = h
        p = c
      else:
        trim(p, c)
      <|_,c|> = c.nxt


  def trim(p, c):
    #flush(c)#
    <|_,s|> = c.nxt
    CAS(p.nxt, <|0,c|>, <|0,s|>)
    #flush(p)#
  \end{sourcecode}
&
  \begin{sourcecode}[gobble=2,firstnumber=last,lineskip=-2pt]
  def insert(h, k):
    while(1):
      <|p,c|> = find(h, k) @\label{lp:harris:insno}@
      if c.key == k:
        #flush(c); flush(c.orig)#
        return false
      n = alloc(Node)
      n.key = k
      n.nxt = <|0,c|>
      #n.orig = p; flush(n); flush(p.orig)#
      if CAS(p.nxt, <|0,c|>, <|0,n|>): @\label{lp:harris:insok}@
        #flush(p)#
        return true

  def delete(h, k):
    while(1):
      <|p,c|> = find(h, k) @\label{lp:harris:delno}@
      if c.key != k:
        return false
      <|b,n|> = c.nxt
      if b == 0:
        #flush(c.orig)# @\label{line:harris:del-flush-orig}@
        if CAS(c.nxt, <|0,n|>, <|1,n|>): @\label{lp:harris:delok}@
          trim(p, c)
          return true
  \end{sourcecode}
\end{tabular}
   \caption{
    A list-based set implementation (in black).
    The addition of the flushes (in \textcolor{extracode}{red}) makes it durable.
  }
  \label{fig:harris-set-code}
\end{figure}

The basic algorithm, shown in black in \cref{fig:harris-set-code},
is designed to implement a finite set of numeric keys in volatile memory.
A set~$S \subs \Key$ is represented in memory as a singly linked list of nodes.
Each node has a \p{key} field storing an element of
$ \Key \dunion \set{+\infty, -\infty} $.
Two sentinel head and tail nodes (with keys $-\infty$ and $+\infty$ resp.)
are always present.
The linked list is ordered in strict increasing order.
Each node can be marked or unmarked;
only unmarked nodes reachable from the head represent elements of~$S$,
and the marked nodes are considered deleted
(they can be lazily unlinked from the list).
The marking is stored as the least-significant bit of the \p{nxt} field,
the pointer to the next node.
In our pseudo-code we represent the \p{nxt} field explicitly
as a pair~$\tup{b, p}$ where $b \in \set{0,1}$ is the marking bit (0 for unmarked, 1 for marked as deleted) and~$p$ is the address of the next node.

The set operations are $\p{insert}(h, k)$ and $\p{delete}(h, k)$
where~$h$ is the address of the head node
(fixed throughout the lifetime of the set)
and $k \in \Key$ is the key to be inserted/deleted.
A successful insert returns \p{true};
if the~$k$ was found to be already in the set, the operation returns \p{false}.
Similarly, the return value of a delete indicates whether
the operation was successful.

\paragraph{Linearizability}
Linearizability~\cite{lin} ensures that, from the perspective of a client of the library,
each call to a library function appears as a single abstract event;
furthermore, these library events are ordered by a total order~$\lin$
which satisfies the abstract semantics of the library (\eg pops and pushes match), and respects the execution (or ``real-time'') order of calls.

A common way to prove linearizability, under sequential consistency (SC)~\cite{sc},
is through identifying the
\emph{linearization point} of each function call, namely the concrete event (\ie memory access)
in the function implementation that represents the moment
when the function's effects become observable to other operations.
Such a proof starts by identifying:
\begin{itemize}
\item A set of \emph{abstract states}, $\AbsState$,
  representing the data-type abstractly presented to the client.
Associated with the states is a set of allowed transitions
  for each operation.
\item A \emph{(volatile) representation function},
  $ \volatile \from \AbsState \to \powerset(\Mem) $,
  formalising how the specific implementation represents an abstract state in memory.
\item For each execution with events~$E$,
  an injective partial function, $\lp \from \CallId \pto E$,
  that identifies the linearization point of each call.
  The function is partial because there might be pending calls that have not reached their linearization points yet.
\end{itemize}

In our example, the abstract states are finite sets of numeric keys $ S \in \finpow(\Key) $.
The transition system on abstract states
asserts \eg that a successful insert of~$k$ is only allowed on a set not containing~$k$ and leads to a state where~$k$ is added to the set.
The representation function 
\eg constrains the memory representing a set~$S$ to be such that
all and only the keys in~$S \dunion \set{+\infty, -\infty}$ are stored in unmarked nodes in the ordered linked list reachable from the head.

The linearization points for the Harris list are as follows:
\begin{itemize}
  \item The successful \p{CAS} at Line~\ref{lp:harris:insok} linearizes
        a successful insert.
  \item A failed insert linearizes at Line~\ref{lp:harris:find}
        during the call to \p{find} of Line~\ref{lp:harris:insno}.
  \item The successful \p{CAS} at Line~\ref{lp:harris:delok} linearizes
        a successful delete.
  \item A failed delete linearizes at Line~\ref{lp:harris:find}
        during the call to \p{find} of Line~\ref{lp:harris:delno}.
\end{itemize}

Given $\volatile$ and $\lp$, linearizability can be then reduced
to an induction over the interleaving sequence of events $e_0 \dots e_n$
of an arbitrary execution.
Let $M_i$ be the memory contents before $e_i$ is executed.
Assuming $M_i \in \volatile(q)$ for some abstract state~$q$,
one must prove that:
\begin{itemize}
  \item if $e_i$ is the linearization point of an operation $\var{op}$,
        then $ M_{i+1} \in \volatile(q') $ for some $q'$
        such that $(q,q')$ is a valid transition for $\var{op}$;
  \item otherwise, the concrete step~$e_i$ preserves the abstract state,
        \ie $ M_{i+1} \in \volatile(q) $.
\end{itemize}

We refer to the proof scheme above as ``induction over execution sequences''.

\paragraph{Linearizability in weak memory models}
In declarative presentations of weak-memory models,
an execution is represented as a graph of events,
related through a number of relations (\eg $\po$, the ``program order'' in each thread) witnessing the execution's consistency.
In such models the notion of ``execution order'' is in fact weaker than
in SC.
In particular, the program order $\po$ on memory accesses
does not necessarily agree with the order in which the accesses take global effect.

To recover a global ``execution order'' in a proof of linearizability,
we can ask the prover to provide, in addition to~$\volatile$ and~$\lp$,
a strict ``volatile order''~$\vo$.
We can then use~$\vo$ to order the execution sequences in the proof of linearizability.
More precisely, we would consider
execution sequences~$e_0 \dots e_n$ that respect~$\vo$.
Any such sequence induces a memory~$\mem{e_0 \dots e_n} \in \Mem \is \Loc \pto \V$
which assigns to each location~$x$ the last value written to it in $e_0 \dots e_n$.
The linearizability proof would then be performed by induction
over \vo-respecting execution sequences~$e_0 \dots e_n$,
defining ${M_i = \mem{e_0 \dots e_{i-1}}}$.

An essential component of the traditional linearizability definition
is that linearization agrees with $\po$ between calls.
To obtain this we can require that~$\po$
between linearization points be preserved by $\vo$,
\ie $\restr{\po}{\LinPt} \subs \vo$ with
where $ \LinPt \is \set{ \lp(c) | c \in \dom(\lp) } $.

\paragraph{Non-volatile memory}

Non-volatile memory~(NVM) introduces another level of complication.
Writes to NVM persist (survive crashes) but not necessarily in the order they were executed.
For example, on Px86 (describing the Intel-x86 persistency model)~\cite{Px86},
if locations~$x$ and~$y$ are not stored in the same cache line,
two sequential writes to $x$ and~$y$ may persist in any order.
To ensure the write on~$x$ persists before that on~$y$,
one must issue a \code{flush} on~$x$ before writing to~$y$.

Px86 thus introduces the strict, total ``non-volatile order'', $\nvo$, on durable
events (\ie writes, updates or flushes)
describing in which order these events persist, 
and a set~$\Persisted$ of those events that have persisted before the crash.
The persisted events~$\Persisted$ have to be consistent with~$\nvo$,
\ie if $(e,e') \in \nvo$ and $e'\in \Persisted$ then $e \in \Persisted$.

In this setting, linearizability is not an adequate correctness criterion as
it does not account for crashes.
A linearizable data-structure without any modifications
would not be correct under Px86:
a crash might leave the persisted memory in an inconsistent state.
Specifically, if no flush is issued, there is no guarantee that any change at all
is persisted even for operations that already returned to the client.
Moreover, even if flushes are issued before returning, pending calls
might have already executed their linearization points,
making the update observable to other threads,
but their updates might not reach the NVM before the crash.
In the Harris list, for example, a key~$k$ might be inserted in the set, and observed by other concurrent inserts, but after a crash one might find that the node carrying~$k$ is not reachable from the head,
or that it is, but has an uninitialized \p{key} field.
That is, it is possible for a crash to invalidate the invariants encoded in~$\volatile$.

\begin{figure}
  \begin{tikzpicture}[
  exec chain,
  font=\small,
event size=6pt,
  event sep=.3pt,
  era sep=1.5cm,
]

\def\CHAINLEN{3}
\def\LOCNUM{4}
\foreach \LINES [count=\ERA from 1] in {
  {4/4/3, 4/4/0, 4/4/1, 4/4/4},
  {4/4/3, 4/3/3, 4/2/4, 4/3/2},
  {3/2/0, 3/4/0, 3/3/0, 3/4/0}} {
  \begin{scope}[start chain=era going below,xshift=(13*\EVSIZE+\ERASEP)*\ERA]

    \foreach \R/\P/\N [count=\LOC from 1] in \LINES {
      \node[event,init](init-\ERA-\LOC) {};
      \begin{scope}[start branch=loc going right]
      \ifnum\R>0
      \foreach \E in {1,...,\R} {
        \node[event,rec](rec-\ERA-\E-\LOC){};
      }
      \fi
      \ifnum\P>0
      \foreach \E in {1,...,\P} {
        \node[event,pers]{};
      }
      \fi
      \node[event,final](fin-\ERA-\LOC){};
      \ifnum\N>0
        \foreach \E in {1,...,\N} {
          \node[event,lost]{};
        }
      \fi
      \end{scope}
    }
    \ifnum\ERA<\CHAINLEN
\path
    (16.5*\EVSIZE,0)
      node[blindred,circle,inner sep=1pt,above] {$\lightning$}
      edge[densely dotted,draw=blindred,semithick] (16.5*\EVSIZE,-\LOCNUM*\EVSIZE)
    ;
\fi
  \end{scope}
}

\node[around loc=fin-1,ACMRed,red abs state={\durable(q_1)}](final-1){};
\node[around loc=init-2,ACMRed,red abs state={\durable(q_1)}](initial-2){};
\draw[shorten > = 1pt,ACMRed,-latex,clear=5pt]
    (final-1) to[bend left=10] (initial-2);

\node[around loc=init-1,abs state={\durable(q_0)}]{};
\node[around loc=rec-1-4,low abs state={\durable(q_0)\inters\volatile(q_0)}]{};

\node[around loc=rec-2-4,low abs state={\durable(q_1)\inters\volatile(q_1)}]{};

\foreach \ERA in {1,...,\CHAINLEN} {
  \node[anchor=south west,font=\footnotesize,inner xsep=0pt] at (init-\ERA-1.north west) {Era~\ERA};
}

\end{tikzpicture}
   \vspace*{-1ex}
  \caption{
    An execution chain with eras separated by crashes ($\,\color{blindred}\lightning\,$).
    For each era we draw the initial memory~ (\protect\evbox{init}),
    the recovery events~(\protect\evbox{rec})
    the persisted writes~(\protect\evbox{pers}, \protect\evbox{final}),
    and the writes which were executed but had not persisted yet when the crash happened (\protect\evbox{lost}).
    In each era, the last persisted writes to each location (\protect\evbox{final})
    provide the initial memory of the next era.
    In the last era every write has persisted.
}
  \label{fig:chain}
\end{figure}

To account for crashes, instead of single executions,
formal persistency models consider \emph{execution chains}:
sequences of executions, where each execution, called an \emph{era},
is abruptly terminated by a crash (with the exception of the last one).
\Cref{fig:chain} shows an example chain:
the shaded area denotes the set of events that have persisted before the crash.
The frontier of the persisted events represents the \nvo-latest
persisted writes to each location: this defines the initial memory of
the next era.
At the beginning of each era, a data-structure-specific recovery routine is
run sequentially before resuming normal execution.

To ensure correctness in the presence of crashes, 
\emph{durable linearizability} (DL)~\cite{durable-lin} requires that:
\begin{enumerate*}
  \item each execution era be linearizable;
  \item the effects of every completed (returned) call be persisted before the next crash;\footnote{
    The requirement that completed calls must have persisted is
    only achievable if the \p{flush} primitive is synchronous,
    \ie blocks until it takes effect.
    When flushes are asynchronous, the correctness criterion
    can be weakened to \emph{buffered} durable linearizability
    that removes the constraint on completed calls.
    In this paper we will only consider synchronous flushes and
    the unbuffered version of durable linearizability.}
  \label{dlitem:completed}
  \item concatenating the linearizations of all eras forms a valid linearization.
\end{enumerate*}
To achieve DL, the programmer has two main tools:
flushes, and the recovery procedure run after each crash.
Let us first focus on schemes that do not require recovery.

A brute-force way to ensure DL is by issuing a flush after each memory access. 
Specifically,
flushing after a write~$w$ ensures that~$w$ is persisted before continuing;
flushing after a read~$r$ ensures that the write observed by~$r$
is persisted before continuing.

When proving DL, this scheme ensures that $\vo$ includes $\nvo$.
As such, since persisted memory is simply $\mem{\vec{e}}$
where $ \vec{e} $ is the \nvo-respecting enumeration of the persisted events $P$
(written~$\mem{\restr{\nvo}{P}}$),
proving linearizability using $\lp$ and $\volatile$
proves that the
original volatile invariants
are now maintained in persistent memory.
When a crash occurs, the persisted memory belongs to~$\volatile(q)$
for some legal~$q$ and the post-crash execution can continue without recovery.

While this aggressive flushing strategy allows for a straightforward adaptation of linearizability to DL,
it yields poor performance.
For this reason, libraries such as FliT~\cite{flit} and Mirror~\cite{mirror} employ alternatives
that more efficiently implement a strict persistency abstraction on top of weaker models.
Conceptually, these libraries improve performance by avoiding redundant flushes on the same write.
Is it possible to do better?
Indeed, \citet{CohenGZ18} proved that any DL library can be implemented using one flush per operation,
which is far lower than what using Mirror or FliT can produce.
To approach this optimum, it is necessary to optimise flushes manually, and genuinely relax the order of persistency on writes.

\paragraph{A first optimisation}
An analysis of the Harris list example reveals that flushing during \code{find}
is not strictly necessary:
when traversing keys $k_1,\dots,k_n$ on the way to finding $k$,
the presence (or absence) of $k_i$
in the abstract set does not influence whether inserting/deleting~$k$ is legal.
Therefore, observing a key~$k_i \ne k$ in the set during traversal
does not require the insertion/marking of~$k_i$ to be persisted:
the result of inserting~$k$ does not reveal information about the presence or absence of~$k_i$.
This is the key insight of NVTraverse~\cite{nvtraverse}, proposing a flushing scheme for tree-based data structures with traverse-and-update operations, with no flushes during traversal.

The program in \cref{fig:harris-set-code}
with the inclusion of the commands in \textcolor{extracode}{red},
instantiates the scheme for the Harris list as follows.
Consider insertions:
a successful insert must first persist (flush) the new node~$n$.
  The second obvious flush needed is the one of \p{p} after the successful
  \p{CAS} which inserted~$n$.
  These flushes alone, however, are insufficient:
  when the \p{CAS} on \p{p} swings the pointer,
  we cannot be sure \p{p} is reachable from the head \emph{in persistent memory}.
  There could be a long list of pending inserts of keys~$k_1\dots k_n$ which all executed their linearization points but have not reached the final flush;
  when this is the case,
  a concurrent insert of~$k_n$ can traverse the $k_1\dots k_n$ nodes
  without flushing them,
  persist $k_n$ and $k_{n {-} 1}$, and
  report to the client that $k_n$ is already in the set.
  If a crash happens then, the node storing~$k_n$ would not be reachable from the head.
The solution is to ensure \p{p} is reachable from the head in persistent memory by flushing the node that initially made \p{p} reachable, \ie its \emph{origin}. We thus record the origin of each node in its \p{orig} field.
The other flushes are issued with analogous motivations.

This more sophisticated scheme cannot be proven by simply adapting the linearizability argument.
In fact, the order of persistency is relaxed,
and as a consequence $\nvo \not\subs \vo$,
contrary to the brute-force approach.
A synthetic example that shows this basic pattern is reproduced in \cref{fig:vo-vs-nvo}(a),
comprising two concurrent operations $\var{op_1}$ (on the left) and $\var{op_2}$ (on the right).
The two \p{CAS} instructions represent linearization points of $\var{op_1}$ and $\var{op_2}$;
\p{x} and \p{y} are distinct locations storing~$0$ initially.
\Cref{fig:vo-vs-nvo}(b) shows a possible execution graph generated by the program under Px86: every memory access leads to a node, labelled with the access' effect.
The linearization of $\var{op_2}$ is observed via a read
by $\var{op_1}$ as signified by the ``read-from'' $\rf$ edge,
but the linearization point on \p{y} does not
depend on the value read from \p{x}.
This observation of the \p{x} value, however, implies
that the \p{CAS} on \p{x} comes \vo-before the \p{CAS} on \p{y}.
The brute-force flushing strategy would mandate the issuing of the red flush (in $\var{op_2}$)
after the read of \p{x}, thus implying that $\vo$ on the \p{CAS}es
is the order in which they will be persisted (thus $\vo$ and $\nvo$ agree on them).
The read of \p{x} in $\var{op_2}$ represents a read during a traversal, and the corresponding flush would be optimized away by the NVTraverse-style optimization.
Without the red flush, the \p{CAS}es can be persisted in either order,
and thus it is possible for~$\nvo$ and $\vo$ to order them differently.
This is shown in the execution graph of \cref{fig:vo-vs-nvo}(c) generated
by the program with the red flush removed.

The overall difference induced by the optimization can be observed
if a crash happens after both \p{CAS}es have executed, and the \p{CAS} on~\p{y} has persisted.
In the unoptimized version, the \p{CAS} on \p{x} would be persisted too,
yielding $\p{x}=\p{y}=1$.
In the optimized version, we might see
$\p{x}=0$ and $\p{y}=1$ after the crash.
The question is: when is such optimisation sound?
We propose to look at the question by identifying the discrepancies
between the linearizations constructed by the (volatile) linearizability argument
using~$\volatile$, $\lp$ and $\vo$
and a DL argument built from~$\volatile$, $\lp$ and~$\nvo$.
These discrepancies can be grouped in two categories:
\begin{enumerate*}
\item two linearization points might be ordered one way by $\vo$ and the other way by $\nvo$; and
\item operations whose linearization points are reads are not meaningfully ordered by $\nvo$.
\end{enumerate*}
The proof strategy we propose is to first prove (volatile) linearizability,
and then prove some properties that entail that the legality of the $\nvo$-induced order on linearization points follows from the legality of the $\vo$-induced one.
\begin{figure}
  \raisebox{-2.5em}{(a)}
  $
    \begin{parall}\begin{threadcode}[gobble=6]
      CAS(x,0,1)
      flush(x)
      $$
      $$
    \end{threadcode}
    \PARALLEL
    \begin{threadcode}[gobble=6]
      a = *x
      #flush(x)#
      CAS(y,0,1)
      flush(y)
    \end{threadcode}
    \end{parall}
  $
  \qquad
  \raisebox{-2.5em}{(b)}
  \begin{tikzpicture}[exec graph, font=\small, baseline=(FLx),
]
  { [thread=1]
  \node [linpt] (Ux) {$\U*{x}{0}{1}$};
  \node [event] {$\FL{x}$};
  }
  { [thread=2]
  \node [right=1cm of Ux, event] (Rx) {$\R*{x}{1}$};
  \node [event,extracode] (FLx) {$\FL{x}$};
  \node [linpt] (Uy) {$\U*{y}{0}{1}$};
}
  \draw
    (Ux)
      edge[rfe] node[above,font=\tiny,inner sep=1pt]{$\rf$} (Rx)
      edge[nvo] node[above,font=\tiny,inner sep=1pt,sloped]{$\nvo$} (FLx)
      edge[vo,bend left=5pt] node[below,dep lbl]{$\vo$} (Uy.north west)
    (FLx)
      edge[nvo,transform canvas={xshift=-1ex}] node[below,dep lbl]{$\nvo$} (Uy)
  ;
\end{tikzpicture}
  \raisebox{-2.5em}{(c)}
  \begin{tikzpicture}[exec graph, font=\small, baseline=(FLx)]
  { [thread=1]
  \node [linpt] (Ux) {$\U*{x}{0}{1}$};
  \node [event] {$\FL{x}$};
  }
  { [thread=2]
  \node [right=1cm of Ux, event] (Rx) {$\R*{x}{1}$};
  \node [linpt] (Uy) {$\U*{y}{0}{1}$};
}
  \draw
    (Ux)
      edge[rfe] node[above,dep lbl=1pt]{$\rf$} (Rx)
      edge[vo,bend left=5pt] node[above,dep lbl=1pt]{$\vo$} (Uy)
    (Uy)
      edge[nvo,bend left=5pt] node[below,dep lbl=1pt]{$\nvo$} (Ux)
;
\end{tikzpicture}   \caption{Optimizing flushes might introduce disagreement between linearization and persistency orders.
    Figure~(a) shows two parallel operations with each \p{CAS} acting as a linearization point.
    Removing the flush in \textcolor{extracode}{red} allows~$y$ to be persisted before~$x$ is.
    Figures~(b) and (c) explain why in terms of an execution graphs generated by the program.
    Each memory access is a node labelled with
    \evtag{U} (successful \p{CAS}),
    \evtag{R} (read), or
    \evtag{FL} (flush).
    The unlabelled arrows indicate program order, $\rf$ is the ``reads-from'' order.
    Figure~(b) is generated by the program with the red flush, which forces $\nvo$ and $\vo$ to order the \p{CAS}es in the same way.
    Figure~(c) is a possible execution if the red flush is removed.
  }
  \label{fig:vo-vs-nvo}
\end{figure}
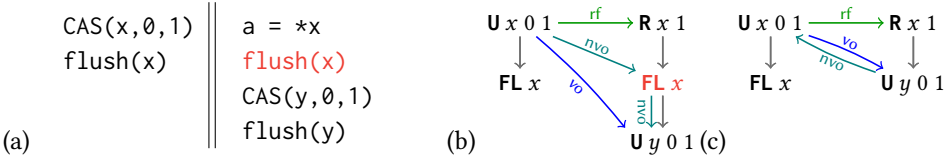
Concretely, for each execution we ask to prove:
\begin{enumerate}
  \item Linearizability of the execution using $\volatile$, $\lp$ and $\vo$.
    \label{cond:informal-lin-vo}
  \item If two linearization points $\lp(c_1)$ and $\lp(c_2)$ are such that
    $ {\lp(c_1) \nvo-> \lp(c_2)} $
    then \\ either
    $ {\lp(c_1) \vo-> \lp(c_2)} $, or
    the calls~$c_1$ and~$c_2$ \emph{abstractly} commute.
    \label{cond:informal-commut}
\end{enumerate}
By ``abstractly commute'' here we mean that,
according to the library specification,
$c_1$ followed by $c_2$ produces the same abstract state
as $c_2$ followed by $c_1$.
What these two conditions imply is that the sequence
of linearization points in $\nvo$ is legal;
therefore any prefix of it is legal.
What is left to prove is that the prefix that persisted (\ie the prefix in~$P$)
produces some abstract state~$q$ which is the one encoded in persisted memory:
\begin{enumerate}[resume]
\item $\mem{\restr{\nvo}{P}} \in \volatile(q)$.
\label{cond:informal-persisted-mem}
\end{enumerate}

For the Harris list, this strategy allows us to reuse the vanilla linearization argument
(flushes and the \p{orig} field have no bearing on that proof).
The proof of condition (\ref{cond:informal-commut}) is done by contraposition:
one shows that it is impossible for two non-commuting operations
to be ordered in opposite ways by $\vo$ and $\nvo$.
This is because an operation on a key~$k$ always flushes
the nodes relevant to $k$ before linearizing.
For example, take an insert and a \vo-subsequent delete of~$k$,
where no other operation on~$k$ took place in between in $\vo$ order.
The delete must have found the inserted node \p{c} in the linked list
(we know this from the vanilla linearization proof)
and therefore flushed the node that first made \p{c} reachable
at Line~\ref{line:harris:del-flush-orig}.
Since the linearization point of the insert is the update of the next pointer
of \p{c.orig}, the flush implies that the linearization point of the insert
is \nvo-before that of the delete.
When proving condition (\ref{cond:informal-persisted-mem}),
we already know that the \nvo-induced sequence is legal;
we can also safely ignore linearization points which are reads
because they have already been proven sound in
condition~(\ref{cond:informal-lin-vo}).

\paragraph{Decoupling recovery}
So far we have only considered schemes with a trivial recovery.
This imposes a greater onus on the implementation which has to ensure
that persisted memory satisfies the invariants at all times.
The possibility of using a non-trivial recovery opens opportunities for
much more radical optimizations.
For instance, in the Harris list,
if we know which nodes belong to the set,
the links between the nodes are indeed redundant
and can be reconstructed upon recovery.

Before examining how these more sophisticated optimisations work,
we structure our proof technique
in a way that decouples the correctness of the recovery and the correctness
of the other library operations.
To do so, we introduce a further parameter of a DL proof:
the \emph{durable} representation function
$\durable \from \AbsState \to \powerset(\Mem)$.
The idea is that $ \volatile(q) $ imposes stronger (or incomparable)
constraints on the memory than $\durable(q)$,
and the constraints of $\durable(q)$ are enough
for the recovery to repair the memory in such a way that it
belongs to $\volatile(q)$ before resuming execution.

The verification of DL then comprises two parts:
\begin{enumerate}
\item Verifying that each step of the recovery takes any
      $ M \in \durable(q) $ to some $ M' \in \durable(q) $,
      and when the recovery terminates the memory is in
      $ \durable(q) \inters \volatile(q) $.
      Note that this verification can be done under SC since the recovery is sequential.
\item Verifying that an execution of the operations, with initial memory in
      $ \durable(q) \inters \volatile(q) $
      is DL, producing a history taking $q$ to $q'$
      and producing a persisted memory in~$\durable(q')$.
\end{enumerate}
These proof obligations allow for an inductive argument
proving that any chain is durably linearizable,
as illustrated in \cref{fig:chain}.
The schemes so far had $ \durable = \volatile $.
Let us next consider a variation of the Harris list with a non-trivial
recovery called a ``link-free set''~\cite{SOFT:oopsla}.

\paragraph{Decoupling linearization and persistency points}

\begin{figure}
  \small
\begin{tabular}{l@{\hspace{3em}}l}
  \begin{sourcecode}[gobble=2,lineskip=-2pt]
  record Node:
    key: $\Nat \dunion \set{+\infty,-\infty}$
    nxt: $\Bool \times \Addr_\nullptr$
    valid: $\Bool$

  def init():
    t = alloc(Node)
    t.key = $+\infty$; t.valid = 1
    t.nxt = <|0,$\nullptr$|>
    h = alloc(Node)
    h.key = $-\infty$; h.valid = 1
    h.nxt = <|0,t|>
    return h


  def find(h, k):
    p = h
    <|_,c|> = p.nxt
    while(1):
      if c.nxt == <|0,_|>: @                      \label{lp:basic-lf:find}@
        if c.key >= k:
          if p.nxt == <|0,_|>:@                   \label{line:find-check-p}@
            return <|p,c|>
          c = h @                                 \label{line:find-restart}@
        p = c
      else:
        trim(p, c) @                              \label{line:basic-lf:find-trim}@
      <|_,c|> = c.nxt

  def trim(p, c):
    flush(c) @                                    \label{line:basic-lf:trim-flush}@
    <|_,s|> = c.nxt
    CAS(p.nxt, <|0,c|>, <|0,s|>)
  \end{sourcecode}
  &
  \begin{sourcecode}[gobble=2,firstnumber=last,lineskip=-2pt]
  def insert(h, k):
    while(1):
      <|p,c|> = find(h, k) @                      \label{lp:basic-lf:insno}@
      if c.key == k:
        c.valid = 1 @                             \label{line:basic-lf:insno-valid}@
        flush(c) @                                \label{line:basic-lf:insno-flush}@
        return false
      n = alloc(Node)
      n.key = k
      n.nxt = <|0,c|>
      if CAS(p.nxt, <|0,c|>, <|0,n|>): @          \label{lp:basic-lf:insok}@
        n.valid = 1 @                             \label{pt:basic-lf:insok}@
        flush(n)
        return true

  def delete(h, k):
    while(1):
      <|p,c|> = find(h, k) @                      \label{pt:basic-lf:delno}\label{lp:basic-lf:delno}@
      if c.key != k:
        return false
      <|_,n|> = c.nxt
      c.valid = 1 @                               \label{line:basic-lf:delok-valid}@
      if CAS(c.nxt, <|0,n|>, <|1,n|>): @          \label{pt:basic-lf:delok}\label{lp:basic-lf:delok}@
        trim(p, c) @                              \label{line:basic-lf:delok-trim}@
        return true

  def recover(nodes):
    h = init()
    for n in nodes:
      if n.valid == 1:
        <|b,_|> = n.nxt
        if b == 0 and n.key${}\in \Key$:
          seqInsert(h, n)
  \end{sourcecode}
\end{tabular}
   \caption{
    A simple link-free set implementation.
  }
  \label{fig:basic-lfset-code}
\end{figure}

The idea of the link-free set,
in \cref{fig:basic-lfset-code},
is to adopt a Harris list data structure as the volatile representation of the set, but let the persistent representation forgo the links between nodes.
More precisely, each node has a boolean \p{valid} field (set to 0 at allocation)
indicating whether the node is a \emph{persistent} member of the set, if not marked as deleted.
Provided the validity field is updated correctly,
the recovery can scan the memory for all the allocated nodes and
newly arrange all the ones that are unmarked (not deleted) and valid (persisted) into a sorted linked list.
This then eliminates the need for persisting the updates on links in any specific order. To ensure consistency between the volatile view and the persisted view,
the link-free set adopts a ``linearize first, persist second'' strategy:
a new node is first inserted in the volatile linked list,
where it becomes visible to other threads,
then its validity field is set and the node is flushed.
If another operation observes the node
(for example a concurrent insert of the same key returning false)
then it will \emph{help} persist the node by first setting its validity and then
flushing it (possibly with a benign race, generating a redundant flush),
before returning. As in the optimized version before, nodes that are traversed but do not
influence the legality of some operation need not be helped to persist.

This optimization has a number of ramifications for verification.
First, we now have~${\volatile \neq \durable}$:
the volatile representation insists that the links must describe a
sorted linked list; the durable representation only asks
that there be a unique, valid and unmarked node for each key in the set.
Second, the original linearization points
are still valid with respect to the volatile structure,
but they do not represent the point where the update they implement is made persistent.
For inserts, for instance, the update to the links is the volatile linearization point, but even if that update is persisted, the inserted node would still be seen as not part of the set after a crash, until the validity field is set and persisted.
The latter update makes the node persistently inserted;
we call this kind of update a \emph{persistency point}.
This duplicity directly reflects the difference between
$\volatile$ and $\durable$:
while linearization points induce an update
from a memory in~$\volatile(q)$ to a memory in~$\volatile(q')$,
persistency points (if persisted)
change the persisted memory from one in~$\durable(q)$ to
one in~$\durable(q')$.

We thus ask the prover to specify persistency points using a partial function
$ \pt \from \CallId \pto P $:
$\pt(c)$ is the persisted event that records the effect of the call~$c$
in persistent memory.
The function is partial because the persistency point of a call might not have been executed/persisted yet.

The persistency points of the link-free set example are as follows:
\begin{itemize}
  \item For successful inserts,
    it is the moment when the inserted node
    is first made valid,
    \ie
    on Line~\ref{pt:basic-lf:insok} of \cref{fig:basic-lfset-code} at the latest.
      \item For successful deletes,
    it coincides with the operation's linearization point,
    Line~\ref{pt:basic-lf:delok} of \cref{fig:basic-lfset-code}.
\end{itemize}

We have thus two overlaid linearization arguments:
the volatile one on $\volatile$, $\lp$, $\vo$;
and the persistent one on
$\durable$, $\pt$, $\nvo$.
On one hand, to prove DL it would suffice to
provide the argument on
$\durable$, $\pt$, $\nvo$.
On the other hand, however, proving linearizability directly on $\nvo$
is challenging.
First, as we noted above, we would need some special treatment of ``read'' operations.
Second, and more important, the reasons for the legality of the sequence
are typically justified by the volatile data structure, and not
just the persisted one.
For instance, the reason why a successful insert of~$k$ is legal is due to
the traversal of the linked list providing evidence
that no other unmarked node holding~$k$ is present.
Since the links in persisted memory might not be consistent,
this argument cannot consider persisted memory only.
As such, just as we sketched above, we propose a proof scheme
that allows most of the correctness argument to be done on $\vo$.
Then, we identify the potentially problematic reorderings of the linearization
produced by the mismatch between $\lp$ and $\pt$, and $\vo$ and $\nvo$.
For those problematic reorderings we require the operations involved to
commute.
Together, these conditions would entail that the linearization
induced by $\pt$ and $\nvo$ is legal.
Then, one needs to verify that the final abstract state reached through
this legal linearization is in fact the one recoverably encoded in the final persisted memory.

All in all, our proof technique requires to prove, roughly:
\begin{enumerate}
  \item Linearizability of the execution using~$\volatile$, $\lp$ and~$\vo$.
    \label{cond:informal-full-lin-vo}
  \item If $ \pt(c_1) \nvo-> \pt(c_2)$, then either
    $ \lp(c_1) \vo-> \lp(c_2) $ or
    $c_1$ and~$c_2$ abstractly commute.
    \label{cond:informal-full-commut}
  \item If $\lp(c_1) \vo-> \lp(c_2)$ but $\pt(c_1)=\bot$,
    then either $\pt(c_2)=\bot$ or~$c_1$ and~$c_2$ abstractly commute.
    \label{cond:informal-full-voided}
  \item Assuming the linearization induced by~$\pt$ and~$\nvo$
    on persisted events abstractly produces a state~$q'$,
    the persisted memory belongs to $\durable(q')$.
    \label{cond:informal-full-lin-nvo}
\end{enumerate}
Condition~(\ref{cond:informal-full-commut}) considers pairs of calls that have
linearized and persisted, but such that volatile and persistent linearizations
would disagree on their ordering.
In that case they are required to abstractly commute.
This allows us to perform reorderings of the volatile linearization
until the sequence respects $\nvo$, while preserving its legality.

Condition~(\ref{cond:informal-full-voided}) corrects for
what we call the ``voided'' calls: those that are linearized
in the middle of the volatile linearization, but have not persisted.
They are required to commute with all the persisted calls in front of them.
This ensures that we can move all of them at the end of the linearization
and then remove them, while preserving legality.

Our \masterthm\ (\cref{sec:master-thm})
formalises a generalisation of this scheme.
As we show in \cref{sec:eval} the scheme makes it possible to prove
the volatile linearization argument first, and exploit it to deduce
the invariants needed to show the commutation lemmas.
The necessity of proving these lemmas
can be seen as the underlying motivation for the placement of the flushes.
Finally, proving that the persisted memory representation of states is
correct can be done while assuming the sequence of persistency points is legal,
effectively making available the relevant volatile invariants
in support of the persistent argument.

Our full proof technique also supports two advanced techniques:
\emph{hindsight linearization} and \emph{persist-first implementations}.
Hindsight~\cite{hindsight} refers to linearizable operations for which
it is not possible to find a fixed event representing their linearization point.
Their correctness is thus proven ``after the fact''.
Our General \thename\ Theorem, presented in~\cite{fullversion},
supports hindsight by allowing a standard linearization
for the other operations to be carried out first,
and then adding hindsight operations,
with a limited impact on the commutation conditions.

Persist-first implementations (\eg~SOFT~\cite{Px86} and Mirror~\cite{mirror})
maintain two versions of their data,
one in persistent memory and one volatile version used for enabling fast access,
and write first to the persistent version and then update the volatile one.
This reduces the needed flushes to the lowest theoretical bound.
The commutation conditions we presented apply to the ``linearize-first''
implementations.
In \cref{sec:persist-first-master} we present a \pfmasterthm\ which,
using dual commutation conditions, applies to persist-first schemes.

 \section{Operational Model and Durable Linearizability}
\label{sec:devel}

\subsection{Preliminaries}
\label{sec:prelim}

\paragraph{Relations}
We write $ \ev{X} $ for the identity relation on~$X$,
$ \tr{\rel} $ for transitive closure,
$ \rtr{\rel} $ for reflexive transitive closure,
$ \maybe{\rel} $ for reflexive closure,
$ \inv{\rel} $ for inverse, and
$ \imm{\rel} \is
    \set{(x,y) \in \rel |
      \nexists z. (x,z) \in \rel \land (z,y) \in \rel  } $.
We say $\rel$ is \emph{acyclic} if $\tr{\rel}$ is irreflexive.

\paragraph{Sequences}
We use the notation~$\vec{e}$ to range over finite sequences,
$\len{\vec{e}}$
  for the length of the sequence,
$ \vec{e}(i) $
  for the item at position~$0 \leq i < \len{\vec{e}}$ in the sequence,
$\upto{\vec{e}}{i}$
  for the sequence~$ \vec{e}(0)\dots\vec{e}(i) $, and
$ \tailfrom{\vec{e}}{i} $ for the sequence
  $ \vec{e}(i)\dots\vec{e}(\len{\vec{e}}-1) $.
The empty sequence is denoted by~$\emptyvec$.
We sometimes implicitly coerce $\vec{e}$ to the set of its items.
Given a set~$A$ we write~$A^*$ for the set
of finite sequences of elements of~$A$.
We write~$\vec{e} \concat \vec{e}'$ for the concatenation of the two sequences.
Given sequences $\vec{e},\vec{e}' \in A^*$, we say
$\vec{e}$ is a \emph{scattered subsequence} of $\vec{e}'$,
if all the items of~$\vec{e}$ appear in $\vec{e}'$ in the same order.
The expression $ \restr{\vec{e}}{B} $ denotes the
longest scattered subsequence of~$\vec{e}$ consisting only of elements of~$B$,
\eg $ \restr{cabcbacb}{\set{a,b}} = abbab $.
For $\vec{e} \in A^*$, we also write $ \vec{e} \setminus B $ for $ \restr{\vec{e}}{A\setminus B} $.

\begin{definition}[Enumeration]
\label{def:enum}
  Given a relation $ {\rel} \subs A \times A $,
  and a finite set $X \subs A$ with~$n$ elements,
  we write $ \enum[X]{\rel} $ for the set of
  enumerations $x_0 \dots x_n$ of $X$ such that
  $\forall i,j \leq n. (x_i, x_j) \in \rel \implies i < j$.
Notice that if ${\rel}$ is an acyclic relation,
  then $\enum[X]{\rel} \ne \emptyset$.
  If $\rel$ is a strict total order on~$X$,
  then $ \enum[X]{\rel} = \set{\vec{e}} $ and we write
  $ \enum[X]{\rel} $ for $\vec{e}$ directly.
  We omit~$X$ when clear from the context.
\end{definition}

\paragraph{Partial functions}
We write $ f \from A \pto B $ if $f$ is a partial function from~$A$ to~$B$,
\ie a function of type~$ f \from A \to B \dunion \set{\bot} $.
The \emph{domain} of~$f$ is $ \dom(f) \is \set{ a \in A | f(a) \ne \bot } $.
The \emph{range} of~$f$ is $ \rng(f) \is \set{ f(a) \in A | a \in A, f(a) \ne \bot } $.
We say $f$ is finite if its domain is finite.

\subsection{Actions and Events}

We make a number of simplifying modelling choices.
First, we only model the scenario where the whole of working memory
is NVM.
Second, we abstract memory management issues and we will assume
memory is managed and garbage collected.
We thus include an atomic allocation primitive but no de-allocation.
Third, the algorithms we are interested in do not use pointer arithmetic,
so we will only model pointers as opaque references
and prove no null-dereference is possible.
Finally, we model the heap as a uniform collection
of structured records, with some fixed finite set of field names~$\Field$.
None of these choices are fundamental.

A \emph{location}~$\loc \in \Loc \is \Addr \times \Field$ is
a pair of an address $x \in \Addr$ and a field name~$\p{f} \in \Field$,
and we will write them as~$ \loc[x.f] $.
The set~$\V$ collects all possible values associated to fields.
The set of locations is partitioned into \emph{cache lines}~$\CacheLine$.
The fields of an address are assumed to fit in a single cache line,
so we postulate that:
$
  \A \CL \in \CacheLine.
\loc[x.f] \in \CL \implies \A \pr{\p{f}} \in \Field.\pr{\loc[x.f]} \in \CL.
$

The set~$\Action$ is the set of \emph{actions}~$\alpha$ which are of the form:
\begin{grammar}
  \alpha \is
    \R{x}{f}{v} |
    \W{x}{f}{v} |
    \U{x}{f}{v}{v'} |
    \MF |
    \FL{x} |
    \Alloc{x} |
    \Ret{v} |
    \Err
\end{grammar}
where~$x\in\Addr$, $\p{f}\in\Field$, $v,v'\in\V$.
We include the standard
read ($\evtag{R}$), write ($\evtag{W}$) and update ($\evtag{U}$)
actions, memory fences ($\MF$), flushes ($\evtag{FL}$),
and three non-standard actions.
Allocation actions~$\Alloc{x}$ initialise all the fields at a fresh~$x$ with zero.
Return actions~$\Ret{v}$ represent the return instruction of a library call;
we will use them as the atomic event representing
the whole invocation in the abstract traces of linearizable libraries.
An error action~$\Err$ is emitted when reading from or writing to
a location with invalid address ($\nullptr$ or not allocated).

Each action (bar~$\Ret{v}$) mentions a single address~$x$ which we can access
with~$\addrOf(\alpha)$.
Similarly,~$\locOf(\alpha)$ is the location mentioned in an action, if any.
As an exception, $\locOf(\Alloc{x})$ is the set~${\set{ \loc[x.f] | \p{f} \in \Field }}$ since an allocation initializes all fields to zero.
The value of a return action is~$\valOf(\Ret{v})\is v$.
We also speak of the \emph{read value} ($\rvalOf$) and \emph{written value} ($\wvalOf$)
of an action:
$
  \rvalOf(\R{x}{f}{v}) \is
  \rvalOf(\U{x}{f}{v}{v'}) \is
    v
$,
$
  \wvalOf(\W{x}{f}{v'}) \is
  \wvalOf(\U{x}{f}{v}{v'}) \is
    v'
$, and
$
  \wvalOf(\Alloc{x}) \is 0.
$
We assume a fixed set of operation names~$\Op$.
For the set library $\Op = \set{\p{insert}, \p{delete}}$.

We also assume an enumerable universe of \emph{events}~$\Event$
equipped with three functions:
  \begin{itemize}
    \item $ \actOf  \from \Event \to \Action $,
      associating an action to every event.
      We lift functions on actions to events in the obvious way,
      e.g.~$\locOf(e) = \locOf(\actOf(e))$
      and write $ (e \of \alpha) $ to indicate that $ \actOf(e) = \alpha $.

    \item $ \cidOf  \from \Event \to \CallId_\bot \dunion \set{\recoveryId} $,
      associating a \emph{call identifier} to every event and $\bot$ to client events.
      Here~$\CallId$ is a fixed enumerable set of call identifiers, and
      $\recoveryId$ is a special identifier reserved for the call
      to the recovery procedure;
      We require~$ \cidOf(e) \neq \bot $ if $ (e \of \Ret{\wtv}) $.

    \item $ \callOf \from \CallId \to \Call $,
      where $\Call \is (\Op \times \V^*)$,
      returns the operation called and its parameters.
  \end{itemize}

The following sets group events by their action type:
\begin{align*}
  \Updates &\is
    \set{ e \in \Event | e\of\U{x}{f}{v}{v'}}
  &
  \MFences &\is
    \set{ e \in \Event | e\of\MF}
  &
  \Flushes & \is \set{ e \in \Event | e \of \FL{x} }
  \\
  \Writes &\is \set{ e \in \Event | (e\of\W{x}{f}{v}) \lor (e\of\Alloc{x})}
  &
  \UWrites & \is \Writes \union \Updates
  &
  \Durable & \is \Writes \union \Updates \union \Flushes
  \\
  \Reads &\is \set{ e \in \Event | e\of\R{x}{f}{v} }
  &
  \UReads & \is \Reads \union \Updates
\end{align*}
We also group events based on their call identifier:
\begin{align*}
  \Rets &\is
    \set{ e \in \Event | e\of\Ret{v}, \cidOf(e) \in \CallId}
  &
  \EvOfCid{i} &\is \set{ e \in \Event | \cidOf(e)=i }
  \\
  \LibEv &\is \set{e\in \Event\setminus \Rets| \cidOf(e) \ne \bot}
  &
  \sameCid &\is \set{ (e_1, e_2) | \cidOf(e_1)=\cidOf(e_1) \ne \bot }
\end{align*}
The set~$\Rets$ collects all return events associated with calls
(excluding the one of the recovery),
the set~$\EvOfCid{i}$ collects all events of the call identified by~$i$,
the set~$\LibEv$ includes all internal library events
(returns are considered to be visible by the client).
The relation~$\sameCid$ relates all events belonging to the same call.
Subscripting a set of events with a location
selects the events for which that location is relevant:
for each of the sets of events~$\mathbb{S}$ defined above,
their location-subscripted variant is
$\mathbb{S}_{\loc} = \mathbb{S} \inters \Event_{\loc}$
and
$\mathbb{S}_{L} = \mathbb{S} \inters \Event_{L}$,
where
$\Event_{\loc} \is \set{ e \in \Event | \loc \in \locOf(e) }$, and
$\Event_{L} \is \set{ e \in \Event | L \inters \locOf(e) \ne \emptyset }$.

\subsection{Executions}

We adopt the declarative approach of weak memory model specifications,
where executions are represented using graphs of events
and dependency relations.
The events of an execution should be understood as the concrete low-level
instructions issued by a closed multi-threaded program.
In the context of the verification of a library,
this closed program would be an arbitrary client issuing both
instructions produced by calls to the library,
and arbitrary instructions on its own locations
(which are assumed to be disjoint from the ones manipulated by the library).

\begin{definition}[Execution]
\label{def:exec}
  An \emph{execution} is a structure
  $
    G = \tup{
      E, \Init, \Persisted, \po, \rf, \mo, \nvo
    }
  $ where
  \begin{itemize}
    \item $E \subs \Event$
      is a finite set of events.
      In the context of the execution~$G$,
      the sets~$\Writes$, $\Reads$, etc should be understood
      as subsets of~$E$.
      Moreover, $ G.\CallId = \cidOf(E) $.
\item $\Init \subs \Writes$
      is the set of \emph{initialisation events},
      with $
        \A e_1,e_2 \in \Init.
          e_1\ne e_2 \implies \locOf(e_1) \inters \locOf(e_2) = \emptyset
      $ and $
        \A e\in\Init.\cidOf(e)=\bot
      $.
      Moreover, there are no double allocations:
      for all $e\in E$ with $(e\of\Alloc{x})$,
      $
        \A {\pr{e} \in \Init}.
          {\addrOf(\pr{e}) \ne x}
      $ and $
        \A {\pr{e} \in E}.
          ((\pr{e}\of \Alloc{y}) \land \pr{e}\ne e) \implies x \ne y.
      $
\item $\Persisted \subs \Durable$
      is the set of \emph{persisted events},
      with $ \Init \union \Flushes \subs P $.
\item $ \po \subs E \times E $
      is the \emph{`program-order' relation},
      required to be a strict partial order
      with ${\Init \times (E \setminus \Init) \subs \po}$.
\item $\rf \subs \UWrites \times \UReads$
      is the \emph{`reads-from' relation}
      between events of the same location with matching values;
      \ie $\A (a, b) \in \rf.
              \locOf(a) {=} \locOf(b) \land \wvalOf(a) {=} \rvalOf(b)$.
      Moreover, $\rf$ is total and functional on its range,
      \ie every read or update is related to exactly one write or update.
\item $\mo \subs E \times E$
      is the \emph{`modification-order'},
      required to be a disjoint union of relations
      $\set{\mo_{\loc}}_{{\loc} \in \Loc}$,
      such that each $\mo_{\loc}$ is a strict total order on $\UWrites_{\loc}$,
      and $
        \Init_{\loc} \times (\UWrites_{\loc} \setminus \Init_{\loc})
          \subs \mo_{\loc}
      $.
\item $\nvo \subs \Durable \times \Durable$
      is the \emph{`non-volatile-order'},
      required to be a strict total order on~$\Durable$,
      such that
        $\Init \times (\Durable \setminus \Init) \subs \nvo$ and
$
          (e_1,e_2) \in \nvo \land e_2 \in \Persisted
          \implies e_1 \in \Persisted $.
  \end{itemize}
  The derived \emph{happens-before} relation
  is defined as $\hb \is \tr{(\po \union \rf)}$.
\end{definition}

A memory model is characterised by the subset of executions that are
feasible, called \emph{consistent} executions.
Different models can be used by adopting a different consistency criterion.

Although our proof technique applies independently
of the choice of consistency criterion,
we will articulate it on the Px86 memory model,
defined in full in \appendixref{sec:appendix-prelim}.
An execution~$G$ is \emph{Px86-consistent} if there exists
a strict \emph{total store order}
$ \tso \subs G.E \times G.E $ representing the global order in which durable
instructions are observed to affect the memory,
which satisfies the usual x86 axioms~\cite{x86-tso}.
In addition,
$\tso$ must satisfy the following three conditions:
\begin{align*}
  \A \CL \in \CacheLine.
\ev{\Durable_{\CL}} \seq \tso \seq \ev{\Durable_{\CL}}
  &\subs \nvo
&
  \ev{\Flushes} \seq \tso \seq \ev{\Durable}
  &\subs \nvo
&
  E \inters \Flushes &\subs P
\end{align*}
The first condition requires that all the durable events on the same cache line
be persisted in the same order in which they affected the volatile memory.
The second condition says that durable events \tso-following a flush
will be persisted after the flush (and thus after all the durable events on the flushed cache line that happened \tso-before the flush).
The third condition characterises the \emph{synchronous} flush semantics:
flushes are persisted as soon as they are included in an execution.

This model deviates slightly from Px86\textsubscript{sim}
of \citet{Px86},
where $\nvo$ only preserves $\tso$ on durable events on the same \emph{location},
not the same \emph{cache line}.
Our stronger semantics is consistent with the actual hardware implementations
\cite[§10.1.1]{snia}.
In fact, the algorithms of \cite{SOFT:oopsla} are only correct
and optimal under the stronger model we adopt in this paper.
\ifappendix
  For more details, see \cref{rm:px86}.
\else
  For instance, for the \p{insert} in \cref{fig:basic-lfset-code} it is
  crucial that the \p{valid} and \p{key} fields get persisted together,
  or a crash might leave a valid node in memory with an uninitialized key.
\fi

\begin{definition}[Memory]
\label{def:mem}
  The \emph{memory} relative to
  some sequence of events $\vec{e}$,
  is the finite partial function
  $ \mem{\vec{e}} \from \Loc \pto \V $
  defined as:
  $
    \mem{\vec{e}}(\loc) \is
      \wvalOf(\vec{e}(i))
  $
  where
  $ i = \max\set{j | \vec{e}(j) \in \UWrites_{\loc}} $.
  The function is undefined on~$\loc$ if $ \vec{e} \inters \UWrites_{\loc} = \emptyset $.
If~$\rel \subs \Event \times \Event $ is
  a strict total order on~$\UWrites$,
  then $ \enum[\UWrites]{\rel} = \set{\vec{e}} $ for some $\vec{e}$,
  and we write $ \mem{\rel} $ for $ \mem{\vec{e}\,} $.
We write~$\Mem$ for the set of finite partial functions from locations to values.
\end{definition}

A set of initial events, by definition,
contains at most one write event per location.
Memories and sets of initial events are therefore in a 1-to-1 correspondence
modulo identity of events.
By virtue of this, we shall implicitly coerce memories into sets of initial events and vice versa.

\begin{definition}[Chain]
  A \emph{chain} is a sequence~$G_0\dots G_n$ of consistent executions such that
  $G_{0}.\Init = \emptyset$,
  $G_{i+1}.\Init = \mem{\restr{G_{i}.\nvo}{G_{i}.\Persisted}}$
  for every~$0\leq i< n$, and
  $G_{n}.\Persisted = G_{n}.\Durable$.
\end{definition}

A library implementation is a pair
${\LibImpl = \tup{\LibImpl[op],\LibImpl[rec]}}$
where $\LibImpl[op]$ describes the implementation of each operation,
and $\LibImpl[rec]$ describes the implementation of the recovery.
An \emph{execution of}~$\LibImpl$ is a consistent execution
where the recovery is run sequentially at the beginning, and then
arbitrary client events and calls of operations run concurrently.
We only consider executions where libraries and clients work on
disjoint locations.
A chain of~$\LibImpl$ is a chain of executions of~$\LibImpl$.
The formal definitions are unsurprising and relegated to \appendixref{sec:appendix-prelim}.

An execution of a library implementation contains both
client events, and internal events
that are conceptually opaque to the client.
From the perspective of the client, all operation calls should be viewable
as instantaneous return events giving back control to the client,
ordered by some (legal) total order.
Abstract executions encode such a client-side view of an execution.
Since some calls might not have returned yet,
the abstract execution can insert the return events for the calls that
have not returned yet but have already conceptually carried out their work.
An important constraint is that the abstract execution sees all
return events (including the inserted ones) as being persisted.
This means that the operations need to ensure their updates have been persisted
before returning, to be consistent with their abstract execution.

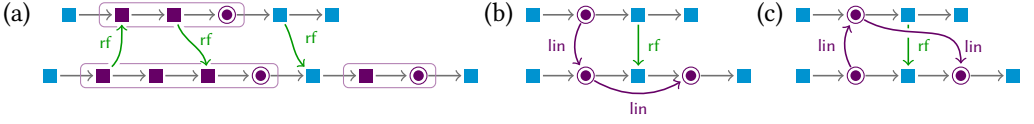
\begin{figure}
  {(a)\ }
  \begin{tikzpicture}[
  abstract execution,
  font=\small,
  baseline=(thread1-begin.south),
]

  \begin{scope}[start chain=thread1,on grid]
  \node[client event]{};
  \node[lib event](t1){};
  \node[lib event](f2){};
  \node[ret event](ret1){};
  \node[client event](f3){};
  \node[client event]{};
  \end{scope}

  \begin{scope}[start chain=thread2,on grid,xshift=-.25cm,yshift=-.8cm]
  \node[client event]{};
  \node[lib event](f1){};
\node[lib event]{};
  \node[lib event](t2){};
  \node[ret event](ret2){};
  \node[client event](t3){};
\node[lib event](c3){};
\node[ret event](ret3){};
  \node[client event]{};
  \end{scope}

  \node[around=(t1)(ret1),](c1){};
  \node[around=(f1)(ret2),](c2){};
  \node[around=(c3)(ret3),](c3){};

  \draw[dep=RF]
    (f1) edge[in=-90,clear] node[left,pos=.6,font=\tiny]{$\rf$} (t1)
    (f2) edge[in=100,out=-70,clear=1pt] node[right=1pt,pos=.35,font=\tiny]{$\rf$} (t2)
    (f3) edge[out=-60] node[right=2pt,pos=.3,font=\tiny]{$\rf$} (t3)
  ;

\end{tikzpicture} {(b)\ }
  \begin{tikzpicture}[
  abstract execution,
  font=\small,
  baseline=(thread1-begin.south),
]

  \begin{scope}[start chain=thread1,on grid]
  \node[client event]{};
\node[ret event](ret1){};
  \node[client event]{};
  \node[client event]{};
  \end{scope}

  \begin{scope}[start chain=thread2,on grid,xshift=0cm,yshift=-.8cm]
  \node[client event]{};
\node[ret event](ret2){};
  \node[client event]{};
\node[ret event](ret3){};
  \node[client event]{};
  \end{scope}

  \draw[dep=RF,font=\tiny]
    (thread1-3) edge node[right]{$\rf$} (thread2-3)
  ;
  \useasboundingbox;
  \draw[dep=LIN,shorten >=2pt,font=\tiny]
    (ret2) edge[bend right] node[below]{$\lin$} (ret3)
    (ret1) edge[bend right] node[left]{$\lin$} (ret2)
  ;

\end{tikzpicture} {(c)\ }
  \begin{tikzpicture}[
  abstract execution,
  font=\small,
  baseline=(thread1-begin.south),
]

  \begin{scope}[start chain=thread1,on grid]
  \node[client event]{};
\node[ret event](ret1){};
  \node[client event]{};
  \node[client event]{};
  \end{scope}

  \begin{scope}[start chain=thread2,on grid,xshift=0cm,yshift=-.8cm]
  \node[client event]{};
\node[ret event](ret2){};
  \node[client event]{};
\node[ret event](ret3){};
  \node[client event]{};
  \end{scope}

  \draw[dep=RF,font=\tiny]
    (thread1-3) edge node[right]{$\rf$} (thread2-3)
  ;
  \draw[dep=LIN,shorten >=2pt,font=\tiny]
    (ret2) edge[bend left] node[left]{$\lin$} (ret1)
    (ret1) edge[out=-30,in=90,looseness=1.2,clear=5pt] node[right,pos=.8]{$\lin$} (ret3)
  ;

\end{tikzpicture}   \caption{The concrete execution graph of Figure~(a) shows
    client events (\protect\evbox{fill=client}), and
    library events (\protect\evbox{fill=lib})
    and return events (\protect\evbox{fill=lib,circle,double distance=1pt,draw=lib,})
    generated by calls (encircled).
    Figures~(a) and (b) represent the only two possible
    abstract executions of Figure~(a):
    the $\rf$ edge between client events constrains
    the possible linearization orders.
  }
  \label{fig:abs-exec}
\end{figure}

\begin{definition}[Abstract execution]
\label{def:abs-exec}
  Fix an execution~$G$.
  A set of \emph{completion events} for~$G$
  is a set~$ C \subs \Rets \setminus G.\Rets $ such that
  $
    \A e\in C. \E e' \in {G.E}.
      \cidOf(e) = \cidOf(e') \land
      G.{\retOf(e')} = \bot
  $.

\begin{samepage}
  Given a set~$C$ of completion events for~$G$,
  we define the execution~$\abs[C]{G}$ as follows:
  \begin{itemize}
    \item $
      \abs[C]{G}.E = \abs{E} \is (G.E \setminus \LibEv) \dunion C
    $.
    \item $
      \abs[C]{G}.\Persisted =
        (G.\Persisted \setminus \LibEv)
    $.
    \item $
      \abs[C]{G}.\po =
        \restr{
          \tr{(
            G.\po \union C_{\po}
)}
        }{\abs{E}}
    $
    where
    $
      C_{\po} =
        (G.E \times C) \inters \sameCid
    $
    \item $
      \abs[C]{G}.\rf = \restr{G.\rf}{\abs{E}}
    $, $
      \abs[C]{G}.\mo = \restr{G.\mo}{\abs{E}}
    $, and $
      \abs[C]{G}.\nvo = \restr{G.\nvo}{\abs{E}}
    $.
  \end{itemize}
\end{samepage}

An \emph{abstract execution} of~$G$
  is a tuple $\tup{\abs[C]{G}, \lin}$
  consisting of
the execution~$\abs[C]{G}$, and
\item
    a strict total order~$\lin$ on $ P_{\Rets} \is G.\Rets \dunion C $
      such that
      $ \restr{\abs[C]{G}.\hb}{P_{\Rets}} \subs \lin $.
Henceforth we use~$\abs{G}$ as a meta-variable ranging over
  the possible abstract executions of~$G$.
\end{definition}

\Cref{fig:abs-exec} shows an example.
As shown in \cref{fig:abs-exec}(a),
an execution~$G$ will in general include client events
(\ie $ G.E \setminus \LibEv \ne \emptyset $).
In particular this means that~$\abs[C]{G}.\hb$ would include edges
between client events, and between client events and return events.
The requirement that $\lin$ should preserve those edges
encodes the idea that the linearization order should never contradict
the client-observable ordering of calls.
In the example, \cref{fig:abs-exec}(b) and \Cref{fig:abs-exec}(c)
are the only abstract executions of~$G$ that respect
$\abs[C]{G}.\hb$: the $\rf$ edge between client events
is preserved in the abstract execution,
requiring that the call of the thread at the top
is linearized before the second call of the thread at the bottom.

\begin{definition}[Histories]
  A \emph{history} is a sequence~$ \h \in \Hist \is (\Call \times \V)^* $.
  Given a sequence of events~$ e_0\dots e_n \in \Rets^* $,
  their history is defined as
  \[
    \hist(e_0\dots e_n) \is
      \tup{\callOf(\cidOf(e_0)), v_0} \dots \tup{\callOf(\cidOf(e_n)), v_n}
  \]
  where $ e_i \of \Ret{v_i} $.
  The legal histories of a library are
  specified as a set~$ \Legal \subs \Hist $.
\end{definition}

\begin{definition}[Durable linearizability]
\label{def:pers-lin}
  A library with legal histories~$\Legal$
  and implementation $ \LibImpl $ is
  \emph{linearizable} if
  for all executions~$G$ of~$\LibImpl$,
  there exist an abstract execution~$\tup{\abs{G}, \lin}$ of~$G$
  such that $ \hist(\enum{\lin}) \in \Legal $.
The library is \emph{durably linearizable}
  if for every chain~$ G_0\dots G_n $ of~$\LibImpl$ there
  are abstract executions
  $ \tup{\abs{G_0},\lin_0},\dots,\tup{\abs{G_n},\lin_n} $ such that
  each $ \tup{\abs{G_i}, \lin_i} $ is an abstract execution of~$G_i$
  and $ \hist(\enum{\lin_0} \concat \dots \concat \enum{\lin_n}) \in \Legal$.
\end{definition}

\subsection{Library Specifications}

Legal histories are all that is needed to specify the desired
abstract behaviour of a linearizable library.
We introduce an abstract-machine-based way of specifying legal histories,
that will allow us to give a more structured proof technique for
proving persistent linearizability.

\begin{definition}[Library Specification]
\label{def:lib-spec}
  A \emph{library specification} is an abstract machine
  that accepts legal histories of library calls.
  Formally, a specification is a tuple
  $ \tup{\AbsState, \Delta, \initState} $
  where
  $ \AbsState $ is a set of \emph{abstract states},
  $ \Delta \from \Call \times \V \to \powerset(\AbsState \times \AbsState) $
  is the transition relation indexed by a call and return value,
  and $\initState \in \AbsState$ is the initial abstract state.
The
  \emph{legal histories of~$ \tup{\AbsState, \Delta} $ from~$q$ to~$q'$} form
  the set~$ \LegalPath{q}{q'} \subs \Hist $,
  defined as the smallest such that
  $ \emptyvec \in \LegalPath{q}{q} $, and
  if $ \h \in \LegalPath{q}{q'} $ and $ (q',q'') \in \Delta(\var{call},v) $
  then $ \h \concat \tup{\var{call},v} \in \LegalPath{q}{q''} $.
The \emph{legal histories of $ \tup{\AbsState, \Delta, \initState} $}
  are defined as
  $
    \LegalFrom{\initState} \is
      \Union_{q\in\AbsState} \LegalPath{\initState}{q}
  $.
  If $c\in\CallId$, we may write~$\Delta(c,v)$ for $\Delta(\callOf(c),v)$.

  A specification is \emph{deterministic} if
  $ \A \tup{\var{call},v} .
    \A q,q_1,q_2.(q,q_1),(q,q_2) \in \Delta(\var{call},v)
        \implies {q_1 = q_2}
  $.
\end{definition}

The legal histories for a set data structure can be formalised as follows.
  Assume a numeric totally ordered type of keys~$\Key$.
    The legal histories of a library implementing a finite set of keys
  are the legal histories of the following (deterministic) library specification.
  The abstract states form the set
  $
    \KS \is \finpow(\Key).
  $
  The transition relation is defined as:
  \begin{align*}
    \Delta(\p{insert}, k, \p{true}) &=
      \set{(S, S \dunion \set{k}) | k \not\in S }
    &
    \Delta(\p{delete}, k, \p{true}) &=
      \set{(S, S\setminus\set{k}) | k \in S }
    \\
    \Delta(\p{insert}, k, \p{false}) &=
      \set{(S, S) | k \in S }
    &
    \Delta(\p{delete}, k, \p{false}) &=
      \set{(S, S) | k \notin S}
\end{align*}

We define a natural notion of equivalence on histories which we will use
to justify the history manipulations in our \masterthm.

\begin{definition}[Equivalent histories]
  Given a library specification $ \tup{\AbsState, \Delta, \initState} $ and
  histories~$ {\h_1,\h_2 \in \Hist} $,
  $ \h_1 \hequiv{\AbsState}{\Delta} \h_2 $
  holds when
  $
    \A q,q' \in \AbsState.
      \h_1 \in \LegalPath{q}{q'}
      \iff
      \h_2 \in \LegalPath{q}{q'}
  $.
\end{definition}

Our proof strategy for durable linearizability
exploits some notions of independence between operations:
commutativity, and the weaker voidability.

\begin{definition}[Commutativity]
\label{def:commut}
  Let $ \tup{\AbsState, \Delta, \initState} $ be a library specification and
  ${\tup{c,v},\tup{c',v'} \in \Call \times \V}$.
  We say
  $\tup{c,v}$ \emph{commutes with}~$\tup{c',v'}$,
  written~${\tup{c,v} \comm{\AbsState}{\Delta} \tup{c',v'}}$,
  if $
    \tup{c,v} \tup{c',v'}
      \hequiv{\AbsState}{\Delta}
    \tup{c',v'} \tup{c,v}.
  $
\end{definition}

For example $ \tup{\p{insert}, k, b} $
commutes with~$ \tup{\p{insert}, k', b'} $ if $k \ne k'$.

\begin{definition}[Voidable call]
\label{def:voidable}
  Let $ \tup{\AbsState, \Delta, \initState} $ be a library specification,
  ${\tup{c,v} \in \Call \times \V}$, and
  ${\h \in \Hist}$.
  We say
  $\tup{c,v}$ is \emph{\pre \h-voidable}
  if
  $
    \A q \in \AbsState.
      {(\tup{c,v} \concat \h) \in \LegalFrom{q}}
        \implies
          \h \in \LegalFrom{q}.
  $
\end{definition}

\begin{lemma}[Voidability in $\KS$]
\label{lm:ks-voidable}
  For the set specification, the following hold
  \begin{itemize}
\item
    $\tup{\p{insert},k,\p{false}}$ and
    $\tup{\p{delete},k,\p{false}}$ are \emph{\pre \h-voidable}
      for every~$\h$.
  \item
    $\tup{\p{insert},k,\p{true}}$ is \emph{\pre \h-voidable}
      if and only if
        $\h$ contains no calls to operations on the key~$k$,
        or $\tup{\p{insert},k,\p{true}} \concat \h$ is not legal.
  \item
    $\tup{\p{delete},k,\p{true}}$ is \emph{\pre \h-voidable}
      if and only if
        $\h$ contains no calls to operations on the key~$k$,
        or $\tup{\p{delete},k,\p{true}} \concat \h$ is not legal.
  \end{itemize}
\end{lemma}

Note that operations which affect the abstract state
may still be voidable.
Moreover,
if the operation $\var{op}$ commutes with all the operations in~$\h$
then~$\var{op}$ is \pre \h-voidable,
but the converse is not necessarily true.
For example, a failed insert of~$k$ is \pre \h-voidable
when~$\h$ consists of a successful delete of~$k$,
but the two calls do not commute.
 \section{A Proof Technique for Persistent Linearizability}
\label{sec:master-thm}

In this section we formalize, through series of lemmas,
our methodology for proving durable linearizability.
The first step is to decouple the verification of recovery and of the operations,
by identifying an interface between the two in the form of the durable and recovered state representation functions.
Then linearizability is reduced to an induction over appropriate sequences of events.
Finally, \cref{sec:the-master-theorem} presents our \masterthm{}.
Throughout the section, we fix some arbitrary
library specification~$\tup{\AbsState, \Delta, \initState}$.

\subsection{Decoupling Recovery}
\label{sec:decoupling-recovery-formal}

As a first step,
we decouple the verification of recovery and of the operations,
so that they can be combined in persistently linearizable chains.
To do so, we specify two invariants, the durable and recovered memories.
Assume the library is specified using the style of~\cref{def:lib-spec}.
The proof technique we propose requires the definition of two
functions $\durable,\recovered \from \AbsState \to\powerset(\Mem)$:
\begin{itemize}
  \item $\durable(q)$ is the set of all
    durable memory representations of~$q$,
  \item $\recovered(q)$ is the set of all
    recovered memory representations of~$q$.
\end{itemize}

We say a memory~$M$ encodes a durable state~$q$ if $M \in \durable(q)$,
or that it encodes a recovered state~$q$ if $M\in\recovered(q)$.
For the recovery, $\durable(q)$ acts both as a precondition,
and as an invariant that must be preserved by each of its steps;
$ \recovered(q) $
is the postcondition of the recovery.
When verifying the operations,
one assumes that the initial memory has been recovered.
At any point in time, the code of operations
maintains the invariant~$\E q.\durable(q)$.
Technically, we start by defining when we consider a recovery
sound with respect to $\durable$ and $\recovered$.

\begin{definition}[Sound recovery]
\label{def:sound-recovery}
Given $\durable,\recovered \from \AbsState \to \powerset(\Mem)$,
  a recovery implementation~$\LibImpl[rec]$ is said
  \emph{\pre\tup{\durable,\recovered}-sound} if,
  for any execution~$G$ of~$\tup{\LibImpl[op], \LibImpl[rec]}$,
  with $G.\Init \in \durable(q)$ for some~$q\in\AbsState$,
  and with~$\LibImpl[op]$ arbitrary,
  the following hold:
  \begin{defenum}
    \item
      \label{cond:recovery-recovers}
if~$
        {G.\Rets} \inters \EvOfCid{\recoveryId} \ne \emptyset
      $ then $ {\mem{\restr{G.\po}{G.\Init\union\EvOfCid{\recoveryId}}} \in \recovered(q)} $;
    \item
      \label{cond:recovery-stutters}
      $
        \A w\in {G.\EvOfCid{\recoveryId}}.
          \A q\in \AbsState, \vec{e},\vec{e}'.
            \mem{G.\Init \concat \vec{e} \concat w \concat \vec{e}'}\in \durable(q)
              \iff
                \mem{G.\Init \concat \vec{e} \concat \vec{e}'}\in \durable(q).
        $
  \end{defenum}
\end{definition}

\Cref{cond:recovery-recovers}
considers the recovery run from a memory encoding a \emph{durable} state~$q$,
and ensures that, when the recovery returns,
the volatile memory encodes the \emph{recovered} state~$q$.
\Cref{cond:recovery-stutters}
requires that any write issued by the recovery,
preserve the durable state encoded by the memory.
More precisely, the writes of the recovery should be irrelevant for~$\durable$.
This requirement implies that any crash occurring during recovery will leave the memory
in a recoverable state, without altering the encoded abstract state.
It is also used in the verification of operations,
to argue that the abstract state of the persisted memory at the time of a crash
is not affected if some recovery events have not been persisted.

\begin{definition}[\pre\tup{\durable,\recovered}-Linearizability]
\label{def:ind-pers-lin}
  Consider
an implementation of operations~$ \LibImpl[op] $.
  We say~$\LibImpl[op]$ is \emph{\pre\tup{\durable,\recovered}-linearizable}
  if, for every~\pre\tup{\durable,\recovered}-sound~$\LibImpl[rec]$,
  all $q\in \AbsState$,
  all~$G$ execution of~$\tup{\LibImpl[op], \LibImpl[rec]}$
  with $G.\Init \in\durable(q)$,
  there is a $q'\in\AbsState$ such that:
  \begin{defenum}
    \item
      \label{cond:ops-legal}
      there is an abstract execution $\tup{\abs{G}, \lin}$ of~$G$
      with $ \hist(\enum{\lin}) \in \LegalPath{q}{q'} $;
    \item
      \label{cond:ops-recoverable}
      $\mem{\restr{G.\nvo}{G.\Persisted}} \in \durable(q')$.
  \end{defenum}
\end{definition}

\Cref{cond:ops-legal} asks to find a linearization that is
a legal history from abstract state~$q$ to~$q'$;
\Cref{cond:ops-recoverable} requires that the
persisted memory encode the same apparent final state~$q'$.

\begin{restatable}{theorem}{recdecoupling}
\label{th:rec-decoupling}
If, for some $\durable,\recovered \from \AbsState \to \Mem$
with~$ {\emptyset \in \durable(\initState)} $,
$\LibImpl[rec]$ is \pre\tup{\durable,\recovered}-sound and
$\LibImpl[op]$ is \pre\tup{\durable,\recovered}-linearizable,
then $\tup{\LibImpl[rec], \LibImpl[op]}$ is
durably linearizable.
\end{restatable}

\subsection{Linearizability Proofs}
\label{sec:lin-stategy-formal}

The first component of a linearizability proof is the volatile order.

\begin{definition}A \emph{volatile order of~$G$} is a transitive relation
  $\vo \subs G.E \times G.E$ such that:
    \begin{defenum}
    \item $G.\Init \times (G.E \setminus G.\Init) \subs \vo$
    \label{cond:init-min}
    \item $G.\EvOfCid{\recoveryId} \times (G.E \setminus (G.\Init \union G.\EvOfCid{\recoveryId})) \subs \vo$
    \label{cond:rec-min}
    \item $
        \ev{\UWrites \inters G.\EvOfCid{\recoveryId}} \seq
        \po \seq
        \ev{\UWrites \inters G.\EvOfCid{\recoveryId}}
          \subs \vo
      $
    \label{cond:rec-tso-resp}
  \end{defenum}
\end{definition}

If an order $\rel$ satisfies condition \ref{cond:init-min},
then for every $ \vec{e} \in G.\enum{\rel} $
we have that for some $i_0 < \len{\vec{e}}$,
$ \upto{\vec{e}}{i_0} = G.\Init $.
Henceforth, we write $\initOf(\vec{e})$ for such $i_0$.
If $\rel$ also satisfies condition \ref{cond:rec-min},
there is some $j$ such that
$
  i\leq j
  \iff
  \vec{e}(i) \in G.\Init \union G.\EvOfCid{\recoveryId}.
$
We write $\recEndOf(\vec{e})$ for such index~$j$.
Finally, if~$\rel$ also satisfies \ref{cond:rec-tso-resp} then
$
  \mem{\upto{\vec{e}}{\recEndOf(\vec{e})}}
  = \mem{\restr{G.\po}{G.\Init\union\EvOfCid{\recoveryId}}}.
$

Next, we introduce linearization strategies,
which package in a tuple the components needed to define a linearization
through the identification of linearization points.

\begin{samepage}
\begin{definition}[Linearization strategy]
\label{def:lin-strategy}
  Given an execution~$G$,
  a \emph{linearization strategy} for~$G$ is a tuple
  $ \tup{\lf,r,\rel, \alpha} $
  where:
  \begin{itemize}
    \item $ \lf \from G.\CallId \pto G.E $
          is an \emph{injective} finite partial function,
          which identifies so-called \emph{linearization event} (if any),
          for the calls of~$G$;
\item $ r \from G.\CallId \pto \Val $
          associates to each call a return value;
\item $ \rel \subs G.E \times G.E $ is an acyclic relation on events;
    \item $ \alpha \from \AbsState \to \powerset(\Mem) $
          is the state representation function.
  \end{itemize}
  We require:
  \begin{defenum}
\item $
      \A c\in\dom(\lf).
        \cidOf(\lf(c)) \ne \bot \land \cidOf(\lf(c)) \ne \recoveryId
    $,
    \item $ \dom(\lf) \subs \dom(r) $, and
    \item $ \A e \in {G.\Rets}. r(\cidOf(e)) = \valOf(e) $.
   \end{defenum}
\end{definition}
\end{samepage}

The function~$\lf$ identifies the event at which each operation is seen to have taken effect;
since not all calls might have reached that event, the function is partial.
The function~$r$ records the return values; it is required to agree with any return event present in the execution, but it might assign return values to calls that have not returned yet, but have already linearized.
The relation~$\rel$ is the one used to order the linearization events to
induce the linearized sequence of calls.
The representation function~$\alpha$ formalises how an abstract state can be
represented in memory.

We overload the `$\hist$' symbol so we can extract histories from sequences
of events, through linearization strategies.

\begin{definition}
\label{def:history-of}
  Let $ \tup{\lf,r,\rel, \alpha} $ be a linearization strategy.
We define:
  \begin{equation*}
    \histOf[\lf]{r}(\vec{e}) \is
      \begin{cases}
        \emptyvec
          \CASE \vec{e} = \emptyvec
        \\
        \tup{\callOf(c), r(c)} \concat \histOf[\lf]{r}(\vec{e}')
          \CASE \vec{e} = e \concat \vec{e}'
          \land \lf(c) = e
        \\
        \histOf[\lf]{r}(\vec{e}')
          \CASE \vec{e} = e \concat \vec{e}'
          \land \A c. \lf(c) \ne e
      \end{cases}
  \end{equation*}
\end{definition}

We defined the notion of strategy
generically because we will instantiate it in two ways
in proofs:
  once with the ``volatile'' parameters
  $\lf=\lp$, $\rel=\vo$, $\alpha=\volatile$,
  and once with the ``persistent'' parameters
  $\lf=\pt$, $\rel=\nvo$, $\alpha=\durable$.
In both cases, we want to use the strategy to validate an execution
by induction on \pre\rel-preserving sequences of events.

\begin{definition}[Validating strategy]
\label{def:valid-strategy}
  The strategy $ \tup{\lf,r,\rel, \alpha} $
  \emph{\pre\initOf-validates}~$G$ if
  for all~$ {\vec{e} \in \enum{\rel}} $, all $q_0 \in \AbsState$
  such that $ \mem{\upto{\vec{e}}{\initOf(\vec{e})}} \in \alpha(q_0) $,
  all~$ \initOf(\vec{e}) < i < \len{\vec{e}} $,
  and all~$ q \in \AbsState $,
  if
  $ \histOf[\lf]{r}(\upto{\vec{e}}{i-1}) \in \LegalPath{q_0}{q} $ and
  $ \mem{\upto{\vec{e}}{i-1}} \in \alpha(q) $,
  then:
  \begin{defenum}
    \item if $ \vec{e}(i) \ne \lf(c) $ for all~$c$, then
      $ \mem{\upto{\vec{e}}{i}} \in \alpha(q) $;
      \label{cond:stutter}\item if $ \vec{e}(i) = \lf(c) $ for some~$c$, then
      $ \mem{\upto{\vec{e}}{i}} \in \alpha(q') $
      for some~$q'$ such that
      $ (q,q') \in \Delta(c, r(c)) $.\label{cond:linpt-trans}\end{defenum}
  The strategy \emph{\pre\recEndOf-validates}~$G$ if it satisfies
  the above conditions using~$\recEndOf$ instead of~$\initOf$.
\end{definition}

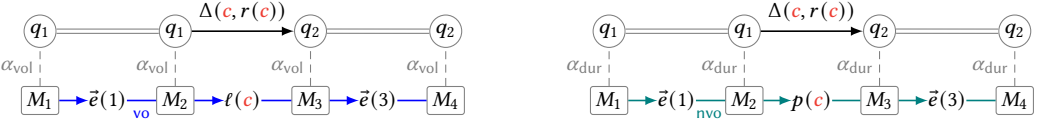
\begin{figure}
{\begin{tikzpicture}[
  linearization,
  font=\footnotesize,
  baseline=(M1.base)
]
  \def\mycall{{\color{ACMRed}c}}

  \begin{scope}[start chain=trace,on grid,join color=VO]
  \node[mem](M1){$M_1$};
  \node[event](e1){$\vec{e}(1)$};
  \node[mem](M2){$M_2$};
  \node[event,lp](e2){$\lp(\mycall)$};
  \node[mem](M3){$M_3$};
  \node[event](e3){$\vec{e}(3)$};
  \node[mem](M4){$M_4$};
  \end{scope}

  \path[VO] (e1) -- node[below,font=\tiny]{$\vo$} (M2);

  \node[above=of M1,state] (s1) {$q_1$};
  \node[above=of M2,state] (s2) {$q_1$};
  \node[above=of M3,state] (s3) {$q_2$};
  \node[above=of M4,state] (s4) {$q_2$};

  \draw
    (M1) edge[abs] node[abs fun]{$\volatile$} (s1)
    (M2) edge[abs] node[abs fun]{$\volatile$} (s2)
    (M3) edge[abs] node[abs fun]{$\volatile$} (s3)
    (M4) edge[abs] node[abs fun]{$\volatile$}
(s4)
    (s1) edge[eq] (s2)
    (s3) edge[eq] (s4)
    (s2) edge[->,semithick,black]
         node[above]{$ \Delta(\mycall, r(\mycall)) $}
    (s3)
  ;

\end{tikzpicture} }
\hfill
{\begin{tikzpicture}[
  linearization,
  font=\footnotesize,
  baseline=(M1.base)
]
  \def\mycall{{\color{ACMRed}c}}

  \begin{scope}[start chain=trace,on grid,join color=NVO]
  \node[mem](M1){$M_1$};
  \node[event](e1){$\vec{e}(1)$};
  \node[mem](M2){$M_2$};
  \node[event,lp](e2){$\pt(\mycall)$};
  \node[mem](M3){$M_3$};
  \node[event](e3){$\vec{e}(3)$};
  \node[mem](M4){$M_4$};
  \end{scope}

  \path[NVO] (e1) -- node[below,pos=.45,font=\tiny]{$\nvo$} (M2);

  \node[above=of M1,state] (s1) {$q_1$};
  \node[above=of M2,state] (s2) {$q_1$};
  \node[above=of M3,state] (s3) {$q_2$};
  \node[above=of M4,state] (s4) {$q_2$};

  \draw
    (M1) edge[abs] node[abs fun]{$\durable$} (s1)
    (M2) edge[abs] node[abs fun]{$\durable$} (s2)
    (M3) edge[abs] node[abs fun]{$\durable$} (s3)
    (M4) edge[abs] node[abs fun]{$\durable$}
(s4)
    (s1) edge[eq] (s2)
    (s3) edge[eq] (s4)
    (s2) edge[->,semithick,black]
         node[above]{$ \Delta(\mycall, r(\mycall)) $}
    (s3)
  ;

\end{tikzpicture} }
\caption{Examples of validating volatile strategy, on the left,
    and validating persistent strategy, on the right.
  }
  \label{fig:valid-strategy}
\end{figure}

The idea of \cref{def:valid-strategy}
is illustrated in \cref{fig:valid-strategy}.
When applied to volatile strategies (\cref{fig:valid-strategy}(a)),
showing that an execution is validated by the strategy
is done by induction on sequences $\vec{e}$ respecting~$\vo$.
For each position $i$ in the sequence,
assuming the \vo-induced memory~$M_i = \mem{\upto{\vec{e}}{i-1}}$
at that point encodes some state~$q$ through~$\volatile$,
one checks if the event~$\vec{e}(i)$
is a linearization point or not according to~$\lp$.
If it is, one checks that the event transforms the memory to one that encodes
some~$q'$ that is legal for the linearizing operation.
If it is not, one checks that the encoded state is preserved.
The end result is to have proven the linearization induced by the strategy
is legal.

The same proof scheme can be applied to persistent strategies
(\cref{fig:valid-strategy}(b)),
with the aim of proving that the persistency points produce a persisted memory
that encodes the expected state.

\begin{lemma}
\label{lm:validates-legal}
  Let\/ $G$ be an execution of~$\tup{\LibImpl[op],\LibImpl[rec]}$
  for some \pre\tup{\durable,\recovered}-sound~$\LibImpl[rec]$,
  with $G.\Init \in \durable(q)$: \begin{defenum}
    \item
      $\tup{\lp, r, \vo, \volatile}$
      \pre\recEndOf-validates~$G {} \implies
        \A{\vec{e} \in G.\enum{\vo}}.
          \E q'.{
            \histOf[\lp]{r}(\vec{e}) \in \LegalPath{q}{q'}
          }.
      $
    \item
      $\tup{\pt, r, \restr{\nvo}{G.\Persisted}, \durable}$
      \pre\initOf-validates~$G {} \implies
        \A{\vec{e} \in \enum[G.\Persisted]{\nvo}}.\E q'.
          \histOf[\pt]{r}(\vec{e}) \in \LegalPath{q}{q'}
          \land
          {\mem{\vec{e}} \in \durable(q')}.
      $
  \end{defenum}
\end{lemma}

\begin{remark}[Completeness]
  As shown by \cite{SchellhornDW14}, identifying linearization points is a complete technique for proving linearizability,
  provided that \begin{enumerate*}
    \item linearization points can be dependent on future events, and
    \item a single event can linearize multiple calls.
  \end{enumerate*}
  Since we define a linearization strategy given a full execution~$G$,
  the first item above is fully supported.
  Even when extended with support for hindsight, however,
  linearization strategies cannot fully accommodate the second item,
  since the maps from calls to linearization events are required to
  be injective.
  This is an obstacle when multiple writer operations linearize
  together (\eg in the elimination stack, where a push and pop cancel out).
  In principle, it is possible to allow a linearization map~$\lp$
  to return, for each call, a pair~$\tup{e,i}$ of an event and a $i\in\Nat$.
  This way, $\lp$ can be injective and still associate multiple calls
  to the same event (their relative order being determined by~$i$).
  For the sake of simplicity, we use here the simpler, incomplete
  definition.
\end{remark}

\subsection{The \thename\ Theorem}
\label{sec:the-master-theorem}

We are now ready for our \masterthm, which we first state and then explain.

\begin{theorem}[\thename\ Theorem]
\label{th:master}
  Consider a library with deterministic
  specification~$\tup{\AbsState, \Delta, \initState}$
  and operations implementation~$ \LibImpl[op] $.
  Let~$ \durable,\recovered,\volatile \from \AbsState \to \powerset(\Mem) $
  with $ \A q.\recovered(q) \subs \volatile(q) $.
  To prove
  $\LibImpl[op]$ is \pre\tup{\durable,\recovered}-linearizable,
  it is sufficient to prove the following.
  Fixing arbitrary
\pre\tup{\durable,\recovered}-sound $\LibImpl[rec]$,
    $q\in \AbsState$, and
    $G$ execution of~$\tup{\LibImpl[op], \LibImpl[rec]}$
    with $G.\Init \in\durable(q)$,
find:{\setlength{\multicolsep}{6.0pt plus 2.0pt minus 1.5pt}\begin{multicols}{2}\begin{itemize}
      \item a volatile order $\vo$
\item $ \lp \from G.\CallId \pto G.E $
\item $ r \from G.\CallId \pto \V $
  \columnbreak
      \item $ \pt \from G.\CallId \pto G.\Persisted $
      \item a persisted readers set\\$ \PRd \subs G.\CallId \setminus \dom(\pt) $
    \end{itemize}
  \end{multicols}}such that:
  \begin{defenum}[itemsep=.4\baselineskip]
    \item $
          \dom(\pt) \union \PRd \subs \dom(\lp).
        $
      \label{cond:main:pers-subs-lp}
    \item $
          \A c \in \PRd.\;
            \Delta(c,r(c)) \subs \idOn{\AbsState}.
        $
\label{cond:main:pr-id}
    \item
      $
        \A c,c'.
          \smash{\lp(c) \hb-> \lp(c')}
          \implies
          \smash{\lp(c) \vo-> \lp(c')}.
      $
\label{cond:main:vo-pres-hb}
    \item $
        \A c \in \dom(\lp).
          \E e_1,e_2 \in {G.\EvOfCid{c}}.
            \smash{e_1 \hb?-> \lp(c) \hb?-> e_2}.
       $
      \label{cond:main:linpt-hb}
    \item $\tup{\lp, r, \vo, \volatile}$
      \pre\recEndOf-validates~$G$.
      \label{cond:main:linpt-valid}
    \item $
        \cidOf(G.\Rets) \subs \dom(\pt) \union \PRd.
      $
      \label{cond:main:rets-persist}
    \item For any~$c, c'\in \dom(\pt)$, either:
      \\
      $
          \tup{\callOf(c), r(c)}
            \comm{\AbsState}{\Delta}
          \tup{\callOf(c'), r(c')}
      $ or $
          \smash{\pt(c) \nvo-> \pt(c')}
          \implies
          \smash{\lp(c) \vo-> \lp(c')}.
      $\label{cond:main:vo-nvo-agree-commute}\item
      \begin{samepage}
        For any~$c \in \dom(\lp) \setminus (\dom(\pt) \union \PRd)$,
        and all $\vec{e} \in {\enum[G.E]{\vo}}$:
        \\
        if $
          \vec{e} = \vec{e}' \concat \lp(c) \concat \vec{e}''
        $ then $
            \tup{\callOf(c), r(c)}
        $ is
        \pre \h-voidable,
        where $
          \h = \restr{\histOf[\lp]{r}(\vec{e}'')}{(\dom(\pt) \union \PRd)}
        $.
\end{samepage}
      \xdef\voidableitem{\arabic{defenumi}}
      \label{cond:main:voided-linpt-voidable}
    \item $
      \histOf[\pt]{r}(\enum[G.\Persisted]{\nvo})
      \in \LegalFrom{q}
        \implies
          \tup{\pt, r, \restr{\nvo}{G.\Persisted}, \durable}
      $ \pre\initOf-validates~$G$.
      \label{cond:main:perst-valid}
  \end{defenum}
\end{theorem}

The theorem follows the high-level description of \cref{sec:overview}.
To prove \pre\tup{\durable,\recovered}-linearizability of~$G$,
we have to provide two linearization strategies:
the volatile one~$\tup{\lp, r, \vo, \volatile}$
and the persistent one~$\tup{\pt, r, \restr{\nvo}{G.\Persisted}, \durable}$.
Since the persistency points must be durable events,
$\dom(\pt)$ only represents the persisted ``writer'' calls,
\ie calls that induce an abstract state change.
The persisted readers set~$\PRd$ indicates which of the ``reader'' calls
should be considered (logically) persisted.

Then the theorem asks us to check a number of conditions.
Condition \ref{cond:main:pers-subs-lp} checks
that every persisted operation has a linearization point.
Condition \ref{cond:main:pr-id}
ensures that the persistent readers are indeed readers.

The goal of conditions
\ref{cond:main:vo-pres-hb} and \ref{cond:main:linpt-hb}
is to make sure the final linearization respects
the~$\hb$ order between the calls, as required by \cref{def:abs-exec}.
In particular, we do not want the linearization to contradict
the~$\po$ ordering between calls.
In fact~$\vo$ might not preserve~$\hb$ in general:
typically $\vo$ reconstructs a global notion of time which
might contradict some $\po$ edges,
for examples the ones between write and read events.
Condition \ref{cond:main:vo-pres-hb}
states that~$\vo$ has to preserve~$\hb$ on linearization points,
which are typically either reads or updates.
Condition \ref{cond:main:linpt-hb}
requires a linearization point to be \hb-between two events of the call it linearizes.
This ensures that an~$\hb$ edge between calls implies an $\hb$ edge between
the respective linearization points.
If linearization points are local to the call they linearize,
the condition is trivially satisfied.

The conditions so far represent well-formedness constraints
that usually hold by construction.
The conditions that follow are the ones where the actual algorithmic insight is involved.

Condition \ref{cond:main:linpt-valid}
requires to prove regular volatile linearizability
using the volatile linearization strategy.
This is typically a straightforward adaptation
of the proof of a non-durable version of the same data structure
(\eg Harris list for the link-free set).
The adaptation simply needs to ensure that the events added
to make the data structure durable preserve the encoded state when executed.

When making a linearizable data structure durable,
after having decided how the persisted memory represents an abstract state (through~$\durable$),
the main design decision is which flushes to issue and when.
Conditions~\ref{cond:main:rets-persist},
\ref{cond:main:vo-nvo-agree-commute}, and
\ref{cond:main:voided-linpt-voidable}
are the ones that check that all the flushes needed for correctness are issued
by the operations.
This can help guiding optimizations,
as redundant flushes would be the ones not contributing to the proofs of these conditions.

Condition~\ref{cond:main:rets-persist} reflects the
\emph{unbuffered} nature of durable linearizability:
it requires every returned call to be considered as persisted.
The only way to ensure a writer operation is persisted before returning
is to issue a (synchronous) flush on the address of the persistency point of
the call, which we dub the ``flush before return'' policy.

Condition \ref{cond:main:vo-nvo-agree-commute}
deals with the discrepancies between the volatile and the persistent linearization orders.
Specifically, it consider two calls that have both persisted.
If they abstractly commute, their relative order in the linearization
does not matter.
Otherwise, it must be the case that they are ordered by the volatile
linearization strategy in the same order as they are by
the persistent strategy.
In implementations, this is typically achieved by making
every call execute their linearization point after having made sure that
the persistency point of earlier non-commuting calls are flushed.

Condition \ref{cond:main:voided-linpt-voidable}
deals with the second source of discrepancies between volatile and persistent
linearizations: voided calls, \ie calls that the volatile argument
sees as having linearized, but have not persisted.
A simpler (but less general) version of the condition would mirror the previous condition using commutativity:
\begin{enumerate}\em
\item[(\voidableitem\/$'$\!)]
  For any
    $ c \in \dom(\lp) \setminus (\dom(\pt) \union \PRd)$ and
    $c' \in \dom(\pt) \union \PRd$, either:
    \\
    $
      \tup{\callOf(c), r(c)}
        \comm{\AbsState}{\Delta}
      \tup{\callOf(c'), r(c')}
    $ or $
      {\lp(c') \vo-> \lp(c)}.
    $
\end{enumerate}
Condition ($\voidableitem'$) considers
the situation where the volatile linearization
places a voided call~$c$
before a persisted call~$c'$.
If such calls commuted, it would be possible to rearrange the volatile
linearization until all the voided calls appear at the end,
while preserving legality of the sequence.
Then the persistent linearization would be a prefix of the volatile one,
which is also a legal sequence.
Condition ($\voidableitem'$) therefore asks that the calls either commute,
or that the voided call is \vo-after the persisted one.

The case study of \cref{sec:eval}, however, does not satisfy ($\voidableitem'$),
but satisfies the more permissive \ref{cond:main:voided-linpt-voidable},
which uses voidability instead of commutativity:
the condition asks to prove that voided calls are voidable
with respect to the rest of the linearization ahead of them.
More precisely, one considers the linearization induced by
$\vo$ and $\lp$, and finds the linearization point of a voided call~$c$.
Then one needs to check that the call is voidable with respect to
the volatile history of non-voided calls ahead ($\h$).
This means it can be removed from the linearization without
affecting its legality.

The condition is typically enforced by making sure that possibly conflicting
calls in $\h$ issue a flush on the address of the persistency point of~$c$
before persisting themselves: if that is the case,
then either $\h$ does not contain conflicting calls, or
$c$ has been persisted, which means it is not voided.

When the conditions presented so far hold,
the legality of the volatile linearization
implies the legality of the persistent linearization.
What remains to prove is that the persistency points modify the persisted memory
so that it encodes the output state expected given the legal linearization.
This is checked by condition \ref{cond:main:perst-valid},
which allows us to \emph{assume} the persisted linearization is legal
(a fact that follows from the other conditions)
and asks us to prove by induction on $\nvo$ that each persistency point
modifies the memory so that the encoded state reflects the expected change.
Proving this does not involve flushes,
but merely checks that the persistency points enact changes that are compatible
with the durable state interpretation~$\durable$.

\subsection{The Persist-First \thename\ Theorem}
\label{sec:persist-first-master}

The \masterthm\ applies to schemes that first linearize operations and then
persist them, as codified by condition \ref{cond:main:pers-subs-lp}.
These schemes are natural as in hardware
writes are first propagated in working memory and then persisted asynchronously.
Persist-first schemes store the data structure in two redundant versions,
one used for durability and a purely volatile one.
They reverse the usual scheme by first committing the effects of an operation in the persistent representation and then linearizing it in the volatile representation.

To the best of our knowledge, every durable data structure in the literature
adopts either a linearize-first or a persist-first scheme, never mixing the two.
We therefore opted for providing a separate theorem for persist-first, although in principle it is possible to provide a joined version.

\begin{wrapfigure}[8]{R}{30ex}\vspace*{-.5em}
  \centering
  \begin{tikzpicture}[
  state/.style={
    outer sep=0pt,
    inner sep=1pt,
  },
node distance=1.5em
]
\small
\node[state](q0){$q_0$};
\node[state,above right=of q0](q1){$q_1$};
\node[state,below right=of q0](q2){$q_2$};
\node[state,below right=of q1](q3){$q_3$};
\draw[->,semithick]
  (q0)
    edge
    node[inner sep=1pt,anchor=south east,pos=.3] (c) {$c$}
    (q1)
  (q0)
    edge[densely dashed]
    node[pos=.4,inner sep=1pt,anchor=north east] (h1) {$h$}
    (q2)
  (q1)
    edge
    node[pos=.4,inner sep=1pt,anchor=south west] (h2) {$h$}
    (q3)
  ;
\path (h1) -- node[font=\large,sloped] {$\impliedby$} (h2);
\node[below=0 of q2]{Voidable};

\node[state,right=of q3](q0){$q_0$};
\node[state,above right=of q0](q1){$q_1$};
\node[state,below right=of q0](q2){$q_2$};
\node[state,below right=of q1](q3){$q_3$};
\draw[->,semithick]
  (q0)
    edge
    node[inner sep=1pt,anchor=south east,pos=.3] (h) {$h$}
    (q1)
    edge
    node[inner sep=1pt,anchor=north east,pos=.4] (c1) {$c$}
    (q2)
  (q1)
    edge[densely dashed]
    node[inner sep=1pt,anchor=south west,pos=.4] (c2) {$c$}
    (q3)
  ;
\path (c1) -- node[font=\large,sloped] {$\implies$} (c2);
\node[below=0 of q2]{Appendable};
\end{tikzpicture} \vspace*{-1em}
  \caption{Duality between voidability and appendability.}
  \label{fig:void-app-duality}
\end{wrapfigure}
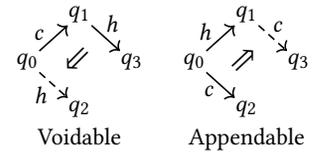

The \pfmasterthm, shown in \appendixref{sec:full-pers-first-master-theorem} coincides with the \masterthm, except for
conditions~\ref{cond:main:pers-subs-lp}, \ref{cond:main:vo-nvo-agree-commute}, and \ref{cond:main:voided-linpt-voidable}.
The obvious substitute of \ref{cond:main:pers-subs-lp} is
$ \dom(\lp) \subs \dom(\pt) $.
This excludes the presence of voided calls, but allows for what we call
``prematurely persisted'' calls:
calls that get persisted
at some point in the persistent linearization, but have not linearized yet.
To account for this,
condition~\ref{cond:main:vo-nvo-agree-commute} is modified to
additionally require any persisted and linearized call~$c$
to be persisted \nvo-before any prematurely persisted~$c'$,
unless~$c$ and~$c'$ commute.

Finally, just as condition~\ref{cond:main:voided-linpt-voidable} of
the \masterthm\ requires voided calls to be ``voidable'',
the \pfmasterthm\ requires prematurely persisted calls to be ``appendable''.

\begin{definition}[Appendable call]
\label{def:voidable}
  Let ${\tup{c,v} \in \Call \times \V}$, and
  ${\h \in \Hist}$.
  We say
  $\tup{c,v}$ is \emph{\pre \h-appendable}
  if,
  for all $q \in \AbsState$
  we have $
      \bigl(
        c \in \LegalFrom{q}
        \land
        \h \in \LegalFrom{q}
      \bigr)
        \implies
          {(\h \concat \tup{c,v}) \in \LegalFrom{q}}.
  $
\end{definition}

\Cref{fig:void-app-duality} shows in which sense voidability and appendability
can be considered duals.
For persist-first, the voidability condition~\ref{cond:main:voided-linpt-voidable} is replaced with one that requires
each prematurely persisted call~$c$ to be \pre \h-appendable
where~$\h$ is the history induced by $\vo$ and $\lp$,
that comes \vo-after the persistency point of~$c$.

 \section{Case Study: the Link-Free Set}
\label{sec:eval}

In this section we sketch how our \masterthm\ can be used
to provide a formal proof of the link-free set of \citet{SOFT:oopsla}
with respect to the Px86 memory model.
The full argument is presented in \appendixref{sec:verif-linkfree}.

\subsection{State Representation}
\label{sec:basic-lf-state-repr}

We begin by formalizing the intended memory representation of a set.
We use a distinguished address $\var{head}$ as the entry point of the data structure for the current era.
We have two overlaid representations: the recoverable and the volatile ones.
\begin{definition}[Durable state representation]
\label{def:basic-lf:recoverable}
  We define $ \durable \from \KS \to \powerset(\Mem) $
  as the function such that
  $ M \in \durable(S) $
  if and only if
  $
    M =
      M_{\lbl{s}} \dunion
      M_{\lbl{d}} \dunion
      M_{\lbl{g}}
  $
  where
\begin{align*}
    M_{\lbl{s}} &=
      \textstyle\Dunion_{k\in S} \left[
        \loc[x_k.key] \mapsto k,
        \loc[x_k.nxt] \mapsto \tup{0,\wtv},
        \loc[x_k.valid] \mapsto 1
      \right]
    \\
    M_{\lbl{d}} &=
      \textstyle\Dunion_{y\in X_{\lbl{d}}} [
        \loc[y.key] \mapsto \wtv,
        \loc[y.nxt] \mapsto \tup{1,\wtv},
        \loc[y.valid] \mapsto 1
      ]
    \\
    M_{\lbl{g}} &=
      \textstyle\Dunion_{y\in X_{\lbl{g}}} [
        \loc[y.key] \mapsto \wtv,
        \loc[y.nxt] \mapsto \tup{0,\wtv},
        \loc[y.valid] \mapsto 0
      ]
  \end{align*}
  for some sets of addresses
  $X_{\lbl{s}} = \set{x_k | k\in S}$,
  $X_{\lbl{d}}$, and
  $X_{\lbl{g}} $.
  Intuitively,
  $ M_{\lbl{s}} $ collects the nodes representing members of $S$,
  $ M_{\lbl{d}} $ collects deleted nodes,
  $ M_{\lbl{g}} $ collects garbage nodes.\footnote{For simplicity we assume head and tail nodes (\ie nodes with key $\pm\infty$) do not survive a crash; they would otherwise also accumulate as garbage nodes since recovery re-allocates them.}\end{definition}

\begin{definition}[Volatile state representation]
\label{def:basic-lf:volatile}
  We define $ \volatile \from \KS \to \powerset(\Mem) $
  as the function such that
  $ M \in \volatile(S) $
  iff
  $
    M =
      M_{\lbl{s}} \dunion
      M_{\lbl{d}} \dunion
      M_{\lbl{u}}
  $
  where
  $ x_{-\infty} = \var{head} $,
$
    S \union \set{+\infty,-\infty} = \set{k_1,\dots,k_n}
  $, and
  \begin{align*}
    M_{\lbl{s}} &=
      \textstyle\Dunion_{1 \leq i \leq n} \left[
        \loc[x_{k_i}.key] \mapsto {k_i},
        \loc[x_{k_i}.nxt] \mapsto \tup{0,\wtv},
        \loc[x_{k_i}.valid] \mapsto \wtv
      \right]
    \\
    M_{\lbl{d}} &=
      \textstyle\Dunion_{y\in X_{\lbl{d}}} [
        \loc[y.key] \mapsto \wtv,
        \loc[y.nxt] \mapsto \tup{1,\wtv},
        \loc[y.valid] \mapsto 1
      ]
    \\
    M_{\lbl{u}} &=
      \textstyle\Dunion_{y\in X_{\lbl{u}}} [
        \loc[y.key] \mapsto \wtv,
        \loc[y.nxt] \mapsto \tup{0,\wtv},
        \loc[y.valid] \mapsto 0
      ]
  \end{align*}
  for some sets of addresses
  $X_{\lbl{d}}$,
  $X_{\lbl{u}} $, and
  $ X_{\lbl{s}} = \set{x_{k_1},\dots,x_{k_n}} $
  and such that
  \begin{gather}
    \A x,y \in \Addr.
      \bigl(
        M(\loc[x.key]) < +\infty \land
        M(\loc[x.nxt]) = \tup{\wtv, y}
      \bigr)
        \implies
          M(\loc[x.key]) < M(\loc[y.key])
    \label{cond:basic-lf:sorted-links}
    \\
    \begin{multlined}
    \A i<n.\E m\geq 0.
    \E y_1,\dots,y_m\in X_{\lbl{d}}.
    \E y_{m+1} = x_{k_{i+1}}.
    \hspace{10em}\\
    \textstyle
      M(\loc[x_{k_i}.nxt]) = \tup{0,y_1}
      \land
      \bigl(
      \LAnd_{1\leq j \leq m}
        M(\loc[y_j.nxt]) = \tup{1,y_{j+1}}
      \bigr)
\end{multlined}
    \label{cond:basic-lf:members-reachable}
    \\
    \A y \in X_{\lbl{d}}.
    \E y' \in X_{\lbl{s}} \union X_{\lbl{d}}.
      M(\loc[y.nxt]) = \tup{\wtv,y'}
    \label{cond:basic-lf:marked-nxt-init}
  \end{gather}
  That is,
  $ M_{\lbl{s}} $ is a sorted linked list of
  valid unmarked nodes representing members of the store,
  possibly interleaved with deleted notes;
  $ M_{\lbl{d}} $ represents deleted nodes, and
  $ M_{\lbl{u}} $ represents uninitialised nodes.
  Note that even links in
  $ M_{\lbl{d}} $ and
  $ M_{\lbl{u}} $ are sorted,
  although they are not required to form a list.
  Moreover, the sortedness constraint~\eqref{cond:basic-lf:sorted-links}
  implies~$k_i < k_{i+1}$ for all~$i < n$.

\end{definition}

The recovery function's goal is to repair some memory in $\durable(S)$
so that it belongs to~$\volatile(S)$.
The recovered state representation is therefore
$ \recovered(S) \is \durable(S) \inters \volatile(S) $.

Showing that the recovery function is sound is easy.
The first soundness condition requires to prove that from some memory in $\durable(q)$, upon termination, the recovery produced a memory in $\recovered(q)$. This is easy to establish with standard sequential reasoning.
The second condition requires the recovery to never introduce intermediate
memories which are not in $\durable(S)$. Since the recovery only modifies the \p{nxt} fields, which are not constrained by $\durable$,
the condition follows immediately.

\subsection{Volatile Linearizability}
\label{sec:lf-lineariz}

We sketch how to prove linearizability by providing a suitable
pre-linearization strategy.
To start, we need to pick an appropriate volatile order.
For Px86 a natural choice is the \emph{global happens before} order
\cite{Alglave12}, which reconstructs the possible order of events from
a global point of view:
\[
\ghb \is \tr{\left(
\restr{\bigl(\po \setminus ((\Writes \union \Flushes) \times \Reads)\bigr)}{
        \Event\setminus \Rets
      }
 \union \mo \union (\inv{\rf}\seq\mo) \union (\rf \setminus \po)
 \right)}
\]

The \masterthm\ asks to provide a set of events~$X$
which includes all linearization points and on which $\vo$ preserves $\hb$.
For the link-free set, we let~$X=\UReads$, since $ \ev{\UReads} \seq \hb \seq \ev{\UReads} \subs \ghb $.

\begin{restatable}[$\ghb$ includes $\hb$ on $\UReads$]{lemma}{lmhbimpliesghb}
\label{lm:hb-implies-ghb}
 $ \ev{\UReads} \seq \hb \seq \ev{\UReads} \subs \ghb $.
\end{restatable}

Next, we have to pick linearization points by defining~$\lp$ and return values~$r$.

\begin{definition}[Linearization points]
\label{def:basic-lf:linpt}
  Given an execution~$G$ of the link-free set
  of \cref{fig:basic-lfset-code},
  we define the finite partial functions
  $
    \lp \from G.\CallId \pto G.\UReads
  $
  and
  $
    r \from G.\CallId \pto \V
  $
  as the smallest such that:
  \begin{itemize}
    \item if $\callOf(c) = \tup{\p{insert}, \wtv}$ then
      \begin{itemize}
        \item If there is an event $(e \of \U{p}{nxt}{\tup{0,\wtv}}{\tup{0,\wtv}})\in G.\EvOfCid{c}$ generated
at Line~\ref{lp:basic-lf:insok}
              then $ \lp(c) = e $ and $r(c) = \p{true}$.
        \item If there is an event $(r\of\Ret{\p{false}}) \in G.\EvOfCid{c}$
              then there is at least one read event of~$ G.\EvOfCid{c} $ generated by
              Line~\ref{lp:basic-lf:find} in the call to \p{find}
              at Line~\ref{lp:basic-lf:insno}; let $ (e \of \R{c}{nxt}{\tup{0,\wtv}}) $ be the \po-last such event.
              Then we set $ \lp(c) = e $ and $r(c) = \p{false}$.
      \end{itemize}
    \item if $\callOf(c) = \tup{\p{delete}, \wtv}$ then
      \begin{itemize}
        \item If there is an event $(e \of \U{c}{nxt}{\tup{0,\wtv}}{\tup{1,\wtv}})\in G.\EvOfCid{c}$ generated
              at Line~\ref{lp:basic-lf:delok},
              then $ \lp(c) = e $ and $r(c) = \p{true}$.
        \item If there is an event $(r\of\Ret{\p{false}}) \in G.\EvOfCid{c}$,
              then
              there is at least one event in~$ G.\EvOfCid{c} $ generated by
              Line~\ref{lp:basic-lf:find} in the call to \p{find}
              at Line~\ref{lp:basic-lf:delno};
              let $ (e \of \R{c}{nxt}{\tup{0,\wtv}}) $ be the \po-last such event.
              Then we set $ \lp(c) = e $ and $r(c) = \p{false}$.
      \end{itemize}
  \end{itemize}
  By construction we have
  $ \A c\in\dom(\lp).\cidOf(\lp(c)) = c $,
  $\lp$ is injective,
  $\cidOf(G.\Rets) \subs \dom(\lp)$,
  $\dom(r) = \dom(\lp)$, and
  $ (e \of \Ret{v}) \in G.\Rets \implies r(\cidOf(e)) = v $.
\end{definition}

It is easy to check that $\tup{\lp, r, \ghb, \volatile}$
is a linearization strategy.
With these parameters we could already prove standard linearizability
by showing that $\tup{\lp, r, \ghb, \volatile}$ \pre\recEndOf-validates
every execution.

\begin{theorem}
\label{th:basic-lf:lp-validates}
  Assume an arbitrary \pre\tup{\durable,\recovered}-sound~$\LibImpl[rec]$.
  If~$G$ is an execution of $\tup{\LinkFreeImpl[op],\LibImpl[rec]}$,
  then $\tup{\lp, r, \ghb, \volatile}$ \pre\recEndOf-validates~$G$.
\end{theorem}

We omit the proof as it is subsumed by the proof on the optimized algorithm
provided in the extended version of this paper~\appendixref{sec:verif-linkfree}.

\subsection{Persistency Points}

We now define the final parameters needed to instantiate the \masterthm:
the persistency points through~$\pt$ and the persisted readers~$\PRd$.

\begin{definition}[Persistency points]
\label{def:basic-lf:perspt}
  Given an execution~$G$ of the link-free set in \cref{fig:basic-lfset-code},
  we define the finite partial function
  $
    \pt \from G.\CallId \pto G.\Persisted
  $
  as the smallest such that:
  \begin{itemize}
    \item if $\callOf(c) = \tup{\p{insert}, \wtv}$
      and $\lp(c) = (e \of \U{p}{nxt}{\wtv}{\tup{0, n}} )$ then\\
      $
        {\pt(c) = \min_{\nvo}\set{e' | (e'\of\W{n}{valid}{1}) \in G.\Persisted}}
      $
      (undefined if~$G.\Persisted$ contains no write to \loc[n.valid]);
    \item if $\callOf(c) = \tup{\p{delete}, \wtv}$
          and $ \lp(c) = (e \of \U{\wtv}{nxt}{\tup{0,\wtv}}{\tup{1,\wtv}}) $
          then $ \pt(c) = e $.
  \end{itemize}
  We also define $\PRd = \set{ c | (r\of\Ret{\p{false}}) \in G.\EvOfCid{c} }$.
  By construction we have that
  $\pt$ is injective,
  $\dom(\pt) \inters \PRd = \emptyset$ and
  $\dom(\pt) \union \PRd \subs \dom(\lp)$.
\end{definition}

\subsection{Applying the \thename\ Theorem}

We have now picked all the parameters we need to apply the \masterthm:
we set
  $\vo = \ghb$,
  $X = G.\UReads$,
  $\lp$ and $r$ as defined in \cref{def:basic-lf:linpt},
$\pt$ and $\PRd$ as defined in \cref{def:basic-lf:perspt}.
We can now check each condition of the \masterthm.

Conditions
  \ref{cond:main:pers-subs-lp} and
  \ref{cond:main:pr-id}
are satisfied by construction.
Condition
  \ref{cond:main:vo-pres-hb}
follows from \cref{lm:hb-implies-ghb}.
Condition
  \ref{cond:main:linpt-hb}
is straightforward from $\cidOf(\lp(c))=c$
(which allows us to pick $e_1=e_2=\lp(c)$).
Condition
  \ref{cond:main:linpt-valid}
is a direct consequence of \cref{th:basic-lf:lp-validates}.
The rest of the conditions represent the crucial steps in the correctness proof.

\paragraph{Flush before returning}
Condition~\ref{cond:main:rets-persist} reflects the
\emph{unbuffered} nature of durable linearizability:
it requires every returned call to be considered as persisted.
The only way to enforce something is persisted before returning
is to issue a (synchronous) flush on the address of the persistency point of
the call.
In our case, the ``read-only'' calls that have returned are included in~$\PRd$
by construction.
For successful inserts/deletes we can prove that
if they returned:
\begin{enumerate*}
\item
  they have issued a flush on the address of their persistency point, and
\item
  the persistency point happens \nvo-before the flush.
\end{enumerate*}
Since $ G.\Flushes \subs G.\Persisted $, we conclude that
if those calls returned, their persistency point was persisted.

\paragraph{Commuting calls}
To prove \ref{cond:main:vo-nvo-agree-commute} we focus on persisted
calls $c,c'\in \dom(\pt)$ that do not commute,
that is persisted successful inserts or deletes of the same key~$k$.
As a preliminary step, we can show that any such $c,c'$
are related by~$\ghb$ in some order.
Thus, since $\nvo$ is total on $G.\Persisted$,
condition \ref{cond:main:vo-nvo-agree-commute} can be reduced to proving that
if $ {\lp(c) \ghb-> \lp(c')} $ then $ {\pt(c) \nvo-> \pt(c')}. $
By transitivity, we can focus on calls whose linearization points
are adjacent, \ie there are no linearization points of calls on key~$k$ between $c$ and $c'$ in $\ghb$ order.
Since we already established legality of the \ghb-induced sequence of
linearization points, we can focus on legal pairs of calls.
This leaves us with two cases:
\begin{itemize}
\item $c$ is a successful insert of~$k$,
  and $c'$ a adjacent successful delete of~$k$.
  From the volatile invariants on the \ghb-induced memory,
  we know that $c'$ will be deleting the node~$n$ inserted by~$c$.
  This means that $\loc[n.valid]$ will be written before $\lp(c')$
  (either by $c$ or~$c'$).
  The earliest such write is $\pt(c)$ which would then be \nvo-before
  $\lp(c') = \pt(c')$ as desired.
\item $c$ is a successful delete of~$k$,
  and $c'$ a adjacent successful insert of~$k$.
  Wme can prove that before $\lp(c)$ is executed,
  at least one call to \p{trim} on the deleted node has been run.
  The flush at line~\ref{line:basic-lf:trim-flush} would then ensure the desired $\nvo$ order.
\end{itemize}

\paragraph{Voided calls}
We now check condition \ref{cond:main:voided-linpt-voidable},
which asks that every call~$c$ which reached its linearization point but
has not persisted must be voidable.
\Cref{lm:ks-voidable} tells us that the only possible problematic calls are
successful inserts or deletes, \ghb-followed by operations on the same key.

\begin{lemma}
  Take any~$c \in \dom(\lp) \setminus (\dom(\pt) \union \PRd)$,
  and all $\vec{e} \in {\enum[G.E]{\vo}}$
  such that $
    \vec{e} = \vec{e}' \concat \lp(c) \concat \vec{e}''
  $.
  Then $
    \tup{\callOf(c), r(c)}
  $ is \pre \h-voidable,
  where $
    \h = \restr{\histOf[\lp]{r}(\vec{e}'')}{(\dom(\pt) \union \PRd)}.
  $
\end{lemma}

The proof proceeds by contradiction:
assume $ \tup{\callOf(c), r(c)} $ is not
\pre \h-voidable.
We have two cases: either $c$ is a successful insert or a successful delete,
of some key~$k$.
Let us focus on the insert case, as the delete case is analogous.
$\vec{e}''$ must contain the linearization point of a
\emph{persisted} call involving the same key.
Let~$c'$ be the id of such call.
We can show that if $c'$ persisted,
the persistency point of~$c$ must have been executed and persisted.
We therefore reach a contradiction with the assumption that
$c \in \dom(\lp) \setminus (\dom(\pt) \union \PRd)$.
Since we already established legality of histories induced by~$\lp$ and~$\ghb$,
we are left with the following cases:
\begin{itemize}
    \item $\callOf{c'} = \tup{\p{insert}, k}$, $r(c) = \p{false}$:
      then lines~\ref{line:basic-lf:insno-valid} and
      \ref{line:basic-lf:insno-flush} have been both executed
      so at least one event setting \p{valid} to 1 has been persisted,
      implying $c\in\dom(\pt)$.
    \item $\callOf{c'} = \tup{\p{delete}, k}$, $r(c) = \p{true}$
      then line~\ref{line:basic-lf:delok-valid} has been executed;
      since $c'$ has persisted, we know that
      the successful CAS at line~\ref{lp:basic-lf:delok} also persisted;
      since \p{c.nxt} and \p{c.valid} fit in the same cache line,
      both writes have been then persisted,
      implying $c\in\dom(\pt)$.
  \end{itemize}

\paragraph{Persisted memory correctness}
We have finally arrived at condition \ref{cond:main:perst-valid}.
Until now, we mostly considered \ghb-induced sequences of events
and invariants on memory.
When $\nvo$ was involved, we only needed to prove that some writes/flushes were
inserted in some crucial points.
This is intentional: the fundamental correctness argument typically rests on
the volatile invariants.
The \masterthm\ allows us to focus on those for as long as possible.
What the proof so far entails is the legality of the (persisted)
volatile history,
and, the fact that the \nvo-induced history is legal and equivalent to the volatile one.
What is missing is to prove that the contents of the \emph{persisted} memory
encode the output state of the (legal) history, with respect to~$\durable$.

Concretely, consider proving
$ \tup{\pt, r, \restr{\nvo}{G.\Persisted}, \durable} $
\pre\initOf-validates~$G$, for the link-free set.
The first condition that \Cref{def:valid-strategy} asks us to prove is
that events that are not persistency points preserve the encoded abstract state.
This is easy to prove:
the only non-trivial writes affect the links, which are ignored by
$\durable$;
the other fields are only updated on uninitialized nodes.
The second condition we need to prove is
that persistency points induce the desired legal transitions.
Proving this directly is challenging.
When considering the persistency point of a successful insert of~$k$,
for example,
we need to prove that when the \p{valid} field of the new node~$n$
is first set to 1, the persisted memory is in $\durable(S)$ for some~$S$
with~$k\notin S$.
The reason why this is true is exclusively due to the volatile invariants
which apply to the \ghb-induced memory, memory that may be encoding
(via $\volatile$) some different~$S'' \ne S$.
The proof strategy embedded in our \masterthm\ resolves the mismatch
by letting us \emph{assume} (by virtue of the other conditions)
that the \nvo-induced history is legal.
In particular, in the insert case,
we are allowed to assume $k \notin S$, which by~$\durable(S)$
would imply that after the persistency point, exactly one valid unmarked node
would be holding the key~$k$.
Finally, another notable simplification introduced by this proof strategy,
is that we can safely ignore the ``read-only'' operations like failed
inserts and deletes, since their place in the linearization has already
been provided through the volatile argument.

\subsection{Optimizations and Extensions}
The full link-free algorithm includes a number of further optimizations
and a wait-free \p{contains} operation.
\Appendixref{sec:verif-linkfree} presents a formal proof of the full algorithm
by using a generalization of the \masterthm.

In particular, the full algorithm optimizes the \p{find} function by removing
lines~\ref{line:find-check-p} and~\ref{line:find-restart}
from \cref{fig:basic-lfset-code}:
it is not necessary to check that the predecessor is unmarked
before returning.
This optimization introduces the need for hindsight linearization
of failed deletes.
The issue is that,
we can no longer identify a point in the program where
the~$p$ and~$c$ returned by~\p{find} are both unmarked,
adjacent and reachable.
This makes it impossible to find any particular event that can serve
as the linearization of a failed delete of some~$k$.
In fact, when we check~$p$ is unmarked it might be too early:
a node~$n$ holding~$k$ might be ahead in the list and unmarked.
Then~$n$ might have been marked before we reached it.
By the time we reach~$c$, a new~$n'$ holding~$k$ might have been added behind in the list, and so at this point we are too late to linearize the delete.
The operation is still correct:
between reading~$p$ unmarked and reading~$c$ unmarked,
there must have been a point when no unmarked node in between them
held~$k$.
This kind of ``after the fact'' argument is called a
``hindsight lemma'' in~\cite{hindsight}.

The general version of our \masterthm\ supports hindsight by means of
a partial map $\hres{\lp}(\vec{e}) \from {\CallId \pto \Nat}$,
which associates to each call $c$ that needs hindsight,
an index~$ \hres{\lp}(\vec{e})(c) $ indicating where in the
(\vo-ordered) sequence of events~$\vec{e}$
the call is logically linearized.
Crucially, $\hres{\lp}$ can be specified after $\lp$ has been defined
and has been used to prove
$\tup{\lp, r, \ghb, \volatile}$ \pre\recEndOf-validates
the execution.
The definition of $\hres{\lp}$ can therefore assume the history induced by~$\lp$
and $\vec{e}$ is legal, and find a position where the hindsight calls linearize,
in the same way we informally argued above.

  \section{Related Work}
\label{sec:related}
The literature includes attempts at strengthening and simplifying the original notion of linearizability \cite{lin}
such as strict linearizability \cite{strict-lin}, as well as sophisticated proof strategies for establishing linearizability (\eg \cite{hindsight}).
As the original definition of linearizability was based on the strong sequential consistency (SC) \cite{sc} memory model,
\citet{lin-tso} later adapted linearizability to the weaker TSO model \cite{x86-tso}, while \citet{lin-c11} adapted it to a fragment of the even weaker C11 model \cite{C11}.
\citet{yacovet} developed a general framework for specifying various correctness conditions for concurrent libraries,
including linearizability.

In order to account for the durability guarantees of implementations in the context of persistent memory, \citet{durable-lin} extended linearizability to durable linearizability (DL).
As with the original notion of linearizability, this original DL definition was tied to the strong SC model.
\citet{ptso} later developed a weak persistency model known as PTSO and adapted the notion of DL to PTSO.
\citet{parm,Px86} subsequently developed the PARMv8 and Px86 models, respectively formalising the (weak) persistency semantics of the ARMv8 and Intel-x86 architectures, and accordingly adapted DL to account for PARMv8 and Px86.
Unlike our memory-model-agnostic approach here, these DL definitions are tied to specific persistency models.

Some memory models admit the so-called data-race freedom property (DRF),
which guarantees that, in absence of data-races,
only SC behaviour is observable.When DRF applies, proving linearizability would be simpler,
but the order in which writes are persisted may still be different
from the SC order, and our technique provides a viable proof strategy for proving DL.

The existing literature includes several examples of durable libraries and data structures. 
The most notable example is PMDK \cite{pmdk}, a collection of libraries for persistent programming.
However, as of yet the PMDK libraries lack formal specifications and have not been formally verified.
\citet{persistent-queues} developed several durable queue libraries over the Px86 model; however, they provide an informal argument (in English) that their implementations are correct (satisfy DL) and do not provide a formal correctness proof. 
Similarly, \citet{SOFT:oopsla} developed two durable set implementations (over Px86), including the link-free set we verify here. 
Once again, they do not provide a formal correctness proof of their implementations, and instead present an informal argument without accounting for the intricacies of the underlying Px86 model. 

\citet{parm,Px86} develop durable variants of the Michael-Scott queue \cite{msq} over the PARMv8 and Px86 models, and formally prove that their implementations are correct.
These implementations are much simpler than those we verify here (\eg they do not involve hindsight reasoning). 
Moreover, unlike our approach here, their proofs are non-modular in that they do not separate the linearizability, persistency and recovery proof obligations.
As such, they do not provide any insights that can be adapted to reason about other durable implementations.

\citet{DerrickDDSW21} proposed a sound and complete
refinement-based proof technique for DL in the context of SC,
which they use to prove a queue from~\cite{persistent-queues}.
Their thread-local simulation technique could in principle be combined
with our \thename\ Theorem, yielding a powerful technique for DL under SC.
We leave this exploration to future work.

\begin{acks}
This work was supported by a
\grantsponsor{ERC}{European Research Council}{https://erc.europa.eu/} (ERC)
Consolidator Grant for the project ``PERSIST'' under the European Union's
Horizon 2020 research and innovation programme
(grant agreement No.~\grantnum{ERC}{101003349}).
\end{acks}

\appendix
\renewcommand{\floatpagefraction}{0.7}  \renewcommand{\dblfloatpagefraction}{0.7}

\section{Preliminary definitions}
\label{sec:appendix-prelim}

We define the set of locations relevant to an action:
\begin{align*}
  \locOf(\W{x}{f}{v}) & \is
    \locOf(\R{x}{f}{v}) \is
    \locOf(\U{x}{f}{v}{v'}) \is
    \set{\loc[x.f]}
  \\
  \locOf(\Alloc{x}) & \is \set{ \loc[x.f] | \p{f} \in \Field }
  \\
  \locOf(\Ret{v}) & \is \emptyset
\end{align*}
All actions act on a single location, except for allocation which
initialises in one go all fields.
For the actions with a single location,
we simply write~$\locOf(\alpha) = \loc$.
The value of a return action is~$\valOf(\Ret{v})\is v$.
We also speak of the \emph{read value} ($\rvalOf$) and \emph{written value} ($\wvalOf$)
of an action:
\begin{align*}
  \rvalOf(\R{x}{f}{v}) &\is
  \rvalOf(\U{x}{f}{v}{v'}) \is
    v
  &
  \wvalOf(\W{x}{f}{v'}) &\is
  \wvalOf(\U{x}{f}{v}{v'}) \is
    v'
  \\
  \rvalOf(\Alloc{x}) &\is
  \bot
  &
  \wvalOf(\Alloc{x}) &\is 0
\end{align*}
and undefined otherwise.

We will assume a fixed set of operation names~$\Op$.
For the set library
we fix~$\Op = \set{\p{insert}, \p{delete}}$.

\begin{definition}[Events]
We will assume an enumerable universe of \emph{events}~$\Event$
equipped with three functions:
  \begin{itemize}
    \item $ \actOf  \from \Event \to \Action $,
      associating an action to every event.
      All the functions on actions are lifted to events in the obvious way,
      e.g.~$\locOf(e) = \locOf(\actOf(e))$.
      We also write $ (e \of \alpha) $ to indicate that $ \actOf(e) = \alpha $.

    \item $ \cidOf  \from \Event \to \CallId_\bot \dunion \set{\recoveryId} $,
      associating a \emph{call identifier} to every event.
      Here~$\CallId$ is a fixed enumerable set of call identifiers, and
      $\recoveryId$ is a special identifier reserved for the call
      to the recovery procedure of a library;
      $\cidOf$~can be undefined ${\cidOf(e) = \bot}$
      which means~$e$ is a client event.
      We require~$ \cidOf(e) \notin \bot $ if $ (e \of \Ret{\wtv}) $.

    \item $ \callOf \from \CallId \to \Call $,
      where $\Call \is (\Op \times \V^*)$,
      associating to each call identifier a call to an operation:
      $ \callOf(i) = \tup{\p{op}, \vec{v}}$
      denotes the name~$\p{op}$ of the operation called
      and the parameters of the call~$\vec{v}$.
\end{itemize}
\end{definition}

The following sets group events by their actions:
\begin{align*}
  \Updates &\is
    \set{ e \in \Event | e\of\U{x}{f}{v}{v'}}
  &
  \MFences &\is
    \set{ e \in \Event | e\of\MF}
  \\
  \Writes &\is \set{ e \in \Event | (e\of\W{x}{f}{v}) \lor (e\of\Alloc{x})}
  &
  \UWrites & \is \Writes \union \Updates
  \\
  \Reads &\is \set{ e \in \Event | e\of\R{x}{f}{v} }
  &
  \UReads & \is \Reads \union \Updates
  \\
  \Flushes & \is \set{ e \in \Event | e \of \FL{x} }
  &
  \Durable & \is \Writes \union \Updates \union \Flushes
\end{align*}
We also group events based on their call identifier:
\begin{align*}
  \Rets &\is
    \set{ e \in \Event | e\of\Ret{v}, \cidOf(e) \in \CallId}
  &
  \EvOfCid{i} &\is \set{ e \in \Event | \cidOf(e)=i }
  \\
  \LibEv &\is \set{e\in \Event\setminus \Rets| \cidOf(e) \ne \bot}
  &
  \sameCid &\is \set{ (e_1, e_2) | \cidOf(e_1)=\cidOf(e_1) \ne \bot }
\end{align*}
The set~$\Rets$ collects all return events associated with calls
(excluding the one of the recovery),
the set~$\EvOfCid{i}$ collects all events of the call identified by~$i$,
the set~$\LibEv$ includes all internal library events
(returns are considered to be visible by the client).
The relation~$\sameCid$ relates all events belonging to the same call.

We also filter events by relevant location,
$\Event_{\loc} \is \set{ e \in \Event | \loc \in \locOf(e) }$,
or sets of locations
$\Event_{L} \is \set{ e \in \Event | L \inters \locOf(e) \ne \emptyset }$;
for each of the sets of events~$\mathbb{S}$ defined above,
their location-subscripted variant is
$\mathbb{S}_{\loc} = \mathbb{S} \inters \Event_{\loc}$
and
$\mathbb{S}_{L} = \mathbb{S} \inters \Event_{L}$.

\begin{definition}[Px86-consistent Execution]
\label{def:px86-consistency}
  An execution~$G =\tup{
    E, \Init, \Persisted, \po, \rf, \mo, \nvo
  }$ is \emph{Px86-consistent} (or \emph{consistent}, for short)
  if there exists a strict order, $\tso \subs E \times E$,
  satisfying:
  \begin{align}
    \Init \times (E \setminus \Init) &\subs \tso
    \tag{\textsc{tso-init}}
    \label{axiom:tso-init}
    \\
    \mo &\subs \tso
    \tag{\textsc{tso-mo}}
    \label{axiom:tso-mo}
    \\
    &\mathclap{\tso \text{ is total on } \Event \setminus \Reads}
    \tag{\textsc{tso-total}}
    \label{axiom:tso-total}
    \\
    \rf &\subs \tso \union \po
    \tag{\textsc{tso-rf1}}
    \label{axiom:tso-rf1}
    \\
    \A x \in \Loc.\ev{\UWrites_x} \seq (\tso \union \po) \seq \inv{\rf} \seq \ev{\UWrites_x}
    &\subs \tso
\tag{\textsc{tso-rf2}}
    \label{axiom:tso-rf2}
    \\
    (
      \ev{\UWrites \union \Reads}
      \seq \po \seq
      \ev{\UWrites \union \Reads}
    ) \setminus (\Writes \times \Reads) &\subs \tso
    \tag{\textsc{tso-po}}
    \label{axiom:tso-po}
    \\
    (
      \ev{\Event}
      \seq \po \seq
      \ev{\MFences}
    ) \union (
      \ev{\MFences}
      \seq \po \seq
      \ev{\Event}
    ) &\subs \tso
    \tag{\textsc{tso-mf}}
    \label{axiom:tso-mf}
\\
    (
      \ev{\UWrites \union \Flushes}
      \seq \po \seq
      \ev{\Flushes}
    ) \union (
      \ev{\Flushes}
      \seq \po \seq
      \ev{\UWrites \union \Flushes}
    ) &\subs \tso
    \tag{\textsc{tso-fl-wufl}}
    \label{axiom:tso-fl-wufl}
\\
    \ev{\Reads} \seq \po \seq \ev{\Flushes}
    &\subs \tso
    \tag{\textsc{tso-r-fl}}
    \label{axiom:tso-r-fl}
    \\
    \A \CL \in \CacheLine.
\ev{\Durable_{\CL}} \seq \tso \seq \ev{\Durable_{\CL}}
    &\subs \nvo
    \tag{\textsc{nvo-cl}}
    \label{axiom:nvo-cl}
\\
    \ev{\Flushes} \seq \tso \seq \ev{\Durable}
    &\subs \nvo
    \tag{\textsc{nvo-fl-d}}
    \label{axiom:nvo-fl-d}
    \\
    E \inters \Flushes &\subs P
    \tag{\textsc{unbuf-fl}}
    \label{axiom:unbuf-fl}
  \end{align}
\end{definition}

\begin{remark}
\label{rm:px86}
\Cref{def:px86-consistency} deviates slightly from Px86\textsubscript{sim}
as formalized in~\cite{Px86}.
In the original Px86\textsubscript{sim} model, condition \eqref{axiom:nvo-cl}
is replaced by two, weaker, conditions:
\begin{align}
  \A x \in \Loc.
    \ev{\Durable_{x}} \seq \tso \seq \ev{\Durable_{x}}
    &\subs \nvo
  \tag{\textsc{nvo-loc}}
  \label{axiom:nvo-loc}
  \\
  \A \CL \in \CacheLine.
    \ev{\UWrites_{\CL}} \seq \tso \seq \ev{\Flushes_{\CL}}
    &\subs \nvo
  \tag{\textsc{nvo-wu-fl}}
  \label{axiom:nvo-wu-fl}
\end{align}

Although these weaker guarantees are consistent with
the official Intel specification,
in practice hardware implementations follow the stronger model we use in this paper~\cite[§10.1.1]{snia}.

In fact, the stronger model is \emph{necessary} for the correctness
and optimality of the algorithms of~\cite{SOFT:oopsla}.
The way \eqref{axiom:nvo-cl} is typically used is by ensuring
crucial data structures (\eg the structures storing a node)
fit a single cache line, and by using alignment annotations to ensure
all the fields will be placed in the same cache line.
This way, a single flush to (any field of) a structure would
be guaranteed to persist all the \po-preceding writes to its fields,
in one go.
Take, for instance, a successful insert of the set implementation of
\cref{fig:basic-lfset-code} (the same applies to \cref{fig:lf-set-code}).
The write to \p{valid} at \cref{pt:basic-lf:insok} needs to be persisted
\emph{after} the writes to \p{key} and \p{val},
or a crash could leave in memory a valid node with an uninitialized key.
In the weaker model, an additional flush would be needed before setting the validity bit, spoiling the near-optimality claims of the algorithm.
\end{remark}

\begin{definition}[Library Implementation Executions]
\label{def:lib-impl-exec}
  A \emph{library implementation} is a pair
  ${\LibImpl = \tup{\LibImpl[op],\LibImpl[rec]}}$ with
  $
    \LibImpl[op] \from
      \mathsf{Call}
        \to \powerset(\Action^*)
  $ and
  $
    \LibImpl[rec]
      \in \powerset(\Action^*)
  $
  where
  $\LibImpl[op]$ maps a call of an operation of the library
  to a \emph{prefix-closed} set of sequences of actions,
  and $\LibImpl[rec]$ is the \emph{prefix-closed} set of sequences of actions
  generated by the recovery procedure.
  We assume no event comes after a return, and that for all sequences
  in $\LibImpl[rec]$ every return event is preceded by a memory fence.

  Recall $ \EvOfCid{i} = \set{ e \in \Event | \cidOf(e)=i } $.
  An execution~$G=\tup{E, \Init, \Persisted, \po, \rf, \mo, \nvo}$
  is an \emph{execution of}~$\LibImpl$ if:
  \begin{itemize}
    \item The library events are generated by the implementation:
      \begin{align*}
      \A c \in \CallId.
        \A (e_0\dots e_n) \in \enum[\EvOfCid{i}]{\po}.&
            \actOf(e_0) \dots \actOf(e_n) \in \LibImpl[op](\callOf(c))
      \\
        \A (e_0\dots e_n) \in \enum[\EvOfCid{{}\recoveryId}]{\po}.&
            \actOf(e_0) \dots \actOf(e_n) \in \LibImpl[rec]
      \end{align*}
    \item Calls are sequential:
        for all $e_1,e_2,e_3 \in E$ with
        $\cidOf(e_1) = \cidOf(e_3) \ne \bot$
        and $ \cidOf(e_2) \ne \cidOf(e_1) $:
      \[
          e_1 \po-> e_3 \lor e_3 \po-> e_1
          \text{\quad and\quad}
          e_1 \po-> e_3
          \land
          e_2 \po-> e_3
          \implies
          e_2 \po-> e_1
      \]
    \item Returns come last, if at all:
      \[
        \sameCid \inters (\ev{\Rets}\seq\po) = \emptyset
        \text{\quad and\quad}
        \cidOf(e') \ne \cidOf(e) \ne \bot \land \smash{e \po!-> e'}
        \implies
          e \in \Rets
      \]
      We write $ \retOf(e) \is v$ if there is
      an event $e' \in E$ with $ e' \of \Ret{v} $ and
      $\cidOf(e)=\cidOf(e')$;
      $ \retOf(i) \is \bot$ otherwise.
    \item Recovery \po-precedes all other events:
      \[
        G.\EvOfCid{\recoveryId} \times (E \setminus (\Init \union \EvOfCid{\recoveryId}))
          \subs \po
        \text{\quad and\quad}
        G.\EvOfCid{\recoveryId} \ne \emptyset
      \]
    \item Client and library work on disjoint locations:
      \[
        \A e \in \Event\setminus \LibEv, e' \in \LibEv.
          \locOf(e) \ne \locOf(e').
      \]
  \end{itemize}
We also say $G$ is an execution of~$\LibImpl[rec]$
  if it is an execution of~$\tup{\fun \wtv.\emptyset, \LibImpl[rec]}$.
  A \emph{chain of}~$\LibImpl$ is a Px86-consistent chain of executions of~$\LibImpl$.
\end{definition}

\begin{lemma}[Basic properties of library specifications]
\label{lm:basic-lib-specs}
  For any library specification~$\tup{\AbsState, \Delta, \initState}$
  and for every~$q,q' \in \AbsState$,
  we have the following:
  \begin{defenum}[itemsep=.5\baselineskip]
    \item $
            \LegalPath{q}{q'} =
              \Union_{q''\in\AbsState}
               \LegalPath{q}{q''}
               \concat
               \LegalPath{q''}{q'}
          $
    \item $ \LegalFrom{q} $ is prefix-closed
  \end{defenum}
\end{lemma}

\begin{definition}
  Given a library specification $ \tup{\AbsState, \Delta, \initState} $ and
  histories~$ {\h_1,\h_2 \in \Hist} $,
  we define~$ {\h_1 \himplies{\AbsState}{\Delta} \h_2} $
  to hold when
  $
    \A q,q' \in \AbsState.
      \h_1 \in \LegalPath{q}{q'}
      \implies
      \h_2 \in \LegalPath{q}{q'}
  $.
  Note that
  $ \h_1 \hequiv{\AbsState}{\Delta} \h_2 $
  if and only if
  $ \h_1 \himplies{\AbsState}{\Delta} \h_2 $ and
  $ \h_2 \himplies{\AbsState}{\Delta} \h_1 $.
\end{definition}

\begin{lemma}
\label{lm:basic-hist-equiv}
  The relation $ \hequiv{\AbsState}{\Delta} $
  is a congruence with respect to concatenation of histories.
\end{lemma}

 \section{A Proof Technique for Persistent Linearizability}
\label{sec:full-proof-technique}

In this section we provide proofs for the theorems and lemmas stated
in \cref{sec:master-thm}.

\subsection{Decoupling Recovery}

\recdecoupling*
\begin{proof}
  Let and $ G_1 \dots G_n $ be a chain of~$\tup{\LibImpl[rec], \LibImpl[op]}$
  with~$n>0$.
  Define $ M_0 = \emptyset $,
     and $ M_i = \mem{\restr{G_i.\nvo}{G_i.\Persisted}} $.
  Note that for all~$ 0 < i \leq n $, $ G_i.\Init = M_{i-1} $.
  We prove, for all~$0 < i\leq n$, that there exist
  abstract executions
  $ \tup{\abs{G}_1,\lin_1},\dots,\tup{\abs{G}_i,\lin_i} $
  of $G_1 ,\dots, G_i$ respectively,
  such that, for some $q \in \AbsState$,
  $
    \hist(\enum{\lin_1} \concat \dots \concat \enum{\lin_i})
      \in \LegalPath{\initState}{q}
  $ and $
    M_i \in \durable(q)
  $.
  \begin{induction}
    \step[Base case~$i=1$]
      We can choose~$ q = \initState $;
      we then have $  \emptyvec \in \LegalPath{\initState}{\initState} $ and $M_0 = \emptyset \in \durable(\initState)$.
    \step[Induction step~$i+1$]
      The induction hypothesis gives us
      $ \tup{\abs{G}_1,\lin_1},\dots,\tup{\abs{G}_i,\lin_i} $
      and a state~$q_i$
      such that  $
        \hist(\enum{\lin_1} \concat \dots \concat \enum{\lin_i})
          \in \LegalPath{\initState}{q_i}
      $
      and $
        M_i \in \durable(q_i)
      $.
      By \pre\tup{\durable,\recovered}-linearizability of $\LibImpl[op]$,
      and $ G_{i+1}.\Init = M_{i} \in \durable(q_i) $,
      there is a state~$q_{i+1}$ and abstract execution
      $\tup{\abs{G}_{i+1}, \lin_{i+1}}$ of~$G_{i+1}$
      such that
      $ \hist(\enum{\lin}) \in \LegalPath{q_i}{q_{i+1}} $
      and
      $\mem{\restr{G.\nvo}{G.\Persisted}} \in \durable(q_{i+1})$.
      Since
       \[
         \left.
         \begin{matrix*}[r]
          \hist(\enum{\lin_1} \concat \dots \concat \enum{\lin_i})
            \in \LegalPath{\initState}{q_i}
          \\\hist(\enum{\lin_{i+1}})
            \in \LegalPath{q_i}{q_{i+1}}
         \end{matrix*}
         \right\}
          \implies
          \hist(\enum{\lin_1} \concat \dots \concat \enum{\lin_i} \concat \enum{\lin_{i+1}})
            \in \LegalPath{\initState}{q_{i+1}}
      \]
      we proved our statement.
      \qedhere
  \end{induction}
\end{proof}

\subsection{Linearizability Proofs}

Linearization strategies, as defined in \cref{def:lin-strategy},
are not general enough to support advanced proofs like the one we present
in \cref{sec:verif-linkfree}.

Strategies for the volatile linearization argument might need to account for
calls that are linearized in hindsight.
These are calls for which no particular event can be seen as linearizing them,
yet it is possible to find a legal position for them in the linearization.
To express this we add to strategies a component $\hres{p}$ which we call
\emph{hindsight resolution}, which takes in input a sequence of events
and a call id and returns the index (if any) at which the call can be thought
of logically taking effect.

\begin{definition}\label{def:full-lin-strategy}
  A \emph{full linearization strategy} for~$G$
  is a tuple
  $ \tup{\lf,\hres{\lf},r,\rel, \alpha} $
  such that:
  \begin{defenum}
  \item $ \tup{\lf,r,\rel, \alpha} $
    is a linearization strategy for~$G$,
  \item $\hres{\lf} \from \enum{\rel} \to (G.\CallId \pto \Nat)$,
    called \emph{hindsight resolution},
  \item $\A \vec{e} \in \enum{\rel}.
           \dom(\hres{\lf}(\vec{e})) \inters \dom(\lf) = \emptyset $,
    \label{cond:hind-dom}
\end{defenum}
\end{definition}

Calls linearized by hindsight do not modify the state,
but they might not be abstractly enabled at every state.
For example, a failing delete of~$k$ can only be linearized
when the memory encodes some state~$S$ with $k\notin S$.
Therefore, a full strategy validates an execution, if the strategy it extends validates it and if, additionally, the hindsight calls are linearized
at a time where they can abstractly take a step (as mandated by $\Delta$).

\begin{definition}[Validating full strategy]
\label{def:valid-strategy}
  A \emph{linearization full strategy}
  $ \tup{\lf,\hres{\lf},r,\rel, \alpha} $
  \emph{\pre\recEndOf-validates}~$G$ if,
  the strategy
  $ \tup{\lf,r,\rel, \alpha} $
  \emph{\pre\recEndOf-validates}~$G$ and,
  for all $\vec{e} \in {\enum{\rel}}$
  and all $c \in \dom(\hres{\lf}(\vec{e}))$:
  \begin{defenum}
    \item $
        \hres{\lf}(\vec{e})(c) > \recEndOf(\vec{e})
      $, and\label{cond:hind-after-rec}
    \item $
      \A q.
        \mem*{\upto{\vec{e}}{\hres{\lf}(\vec{e})(c)-1}} \in \alpha(q)
          \implies
            (q,q) \in \Delta(\callOf(c), r(c))
      $.\label{cond:hind-linpt}\end{defenum}
\end{definition}

Since strategies are special cases of full strategies (with $\hres{\lf}(\wtv)=\emptyset$)
we abuse of terminology and call both just ``strategies''.

We define functions
to extract histories from sequences
of events, through (full) linearization strategies.

\begin{definition}
\label{def:history-maps}
  Let $ \tup{\lf,r,\rel, \alpha} $ be a linearization strategy.
We define the \emph{call-identifier history function}
  $ \cidHistOf[\lf] \from (\Event \dunion \CallId)^* \to \CallId^* $
  as follows:
  \begin{equation*}
    \cidHistOf[\lf](\vec{s}) \is
      \begin{cases}
        \emptyvec
          \CASE \vec{s} = \emptyvec
        \\
        \cidHistOf[\lf](\vec{s}')
        \concat c
          \CASE \vec{s} = \vec{s}' \concat c
          \land c \in \CallId
        \\
        \cidHistOf[\lf](\vec{s}')
        \concat c
          \CASE \vec{s} = \vec{s}' \concat e
          \land \lf(c) = e \in \Event
        \\
        \cidHistOf[\lf](\vec{s}')
          \CASE \vec{s} = \vec{s}' \concat e
          \land e \in \Event
          \land \A c. \lf(c) \ne e
      \end{cases}
  \end{equation*}
  Note that the function is well-defined since~$\lf$ is injective,
  so there is no ambiguity as for which~$c$ can satisfy the condition of the third case.

  From a sequence of call identifiers~$ c_1 \dots c_n \in \CallId^* $,
  we can extract the history of calls:
  \begin{align*}
    \histOf{r}(c_1 \dots c_n) & \is
      \tup{\callOf(c_1), r(c_1)}
        \dots
      \tup{\callOf(c_n), r(c_n)}
    \\
    \histOf[\lf]{r}(\vec{s}) & \is
      \histOf{r}(\cidHistOf[\lf](\vec{s}))
  \end{align*}

  Let $\hres{\lf}$ be a hindsight resolution, and $n=\len{\vec{e}}$;
  we define:
  \begin{align*}
    \cidHistOf[\lf,\hres{\lf}](\vec{e}) & \is
      \cidHistOf[\lf](\hres{\lf}[\vec{e}])
    &
    \hres{\lf}(\vec{e})_{@i} &\is
      \set{c | \hres{\lf}(\vec{e})(c) = i}
    \\
    \histOf[\lf,\hres{\lf}]{r}(\vec{e}) & \is
      \histOf{r}(\cidHistOf[\lf,\hres{\lf}](\vec{e}))
    &
    \hres{\lf}[\vec{e}] &\is
      \hres{\lf}(\vec{e})_{@0} \concat \vec{e}(0)
      \cdots
      \hres{\lf}(\vec{e})_{@n} \concat \vec{e}(n)
  \end{align*}
  Note that we are implicitly coercing the set $ \hres{\lf}(\vec{e})_{@i} $
  to a sequence of its elements.
  Since the actual sequence of choice is irrelevant,
  we assume one is picked arbitrarily, and any property
  stated for $ \hres{\lf}(\vec{e})_{@i} $ holds for any choice of sequence.
\end{definition}

\begin{definition}
  \label{def:orig}
  We define the \emph{origin} index of a call:
  \[
    \orig_{\lf,\hres{\lf}\!}(\vec{e}, c) \is
    \begin{cases}
      i
      \CASE \vec{e}(i) = \lf(c)
      \\
      \hres{\lf}(\vec{e})(c)
      \CASE c \in \dom(\hres{\lf}(\vec{e}))
    \end{cases}
  \]
\end{definition}

\begin{lemma}
\label{lm:orig-mono}
  Let $ \vec{c} $ be a scattered subsequence of
  $ \cidHistOf[\lf,\hres{\lf}](\vec{e}) $.
If $
    \orig_{\lf,\hres{\lf}\!}(\vec{e}, \vec{c}(i_1))
    < \orig_{\lf,\hres{\lf}\!}(\vec{e}, \vec{c}(i_2))
  $ then $
    i_1 < i_2.
  $
\end{lemma}
\begin{proof}
  Since scattered subsequences preserve the order, it is enough to prove
  the claim for $ c = \cidHistOf[\lf,\hres{\lf}](\vec{e}) $,
  which follows straightforwardly from
  the definition of $ \cidHistOf[\lf,\hres{\lf}] $.
\end{proof}

\Cref{lm:validates-legal} is a simple corollary of the following
lemmas, which in addition take hindsight resolution into account.

\begin{lemma}
\label{lm:linpt-validates-legal}
  Let~$G$ be an execution of~$\tup{\LibImpl[op],\LibImpl[rec]}$
  for some \pre\tup{\durable,\recovered}-sound~$\LibImpl[rec]$,
  with $G.\Init \in \durable(q)$.
  Assume the linearization strategy
  $\tup{\lp, \hres{\lp}, r, \vo, \volatile}$
  \pre\recEndOf-validates~$G$.
  Then,
  for all~$\vec{e} \in G.\enum{\vo}$,
  ${
    \histOf[\lp,\hres{\lp}]{r}(\vec{e}) \in \LegalPath{q}{q'}
  }$
  for some~$q'$.
\end{lemma}

\begin{proof}
  We prove that, for all $i < \len{\vec{e}}$,
  $
    \E q_i.
      \histOf[\lp,\hres{\lp}]{r}(\upto{\vec{e}}{i})
        \in \LegalPath{q}{q_i}
  $,
  which implies our goal.

  Note that, by condition \ref{cond:hind-after-rec},
  $\hres{\lp}(\vec{e})(c) > \recEndOf(\vec{e}) $,
  for all~$c\in\dom(\hres{\lp})$.

  We consider two cases:
  \begin{casesplit}
    \case[$ i < \recEndOf(\vec{e}) $]
      Since $\A c.\lp(c) \ne \recoveryId$, we know that $
        \histOf[\lp,\hres{\lp}]{r}(\upto{\vec{e}}{i}) =
        \emptyvec \in \LegalPath{q}{q}
      $
      so we can pick~$q_i=q$.
    \case[$ i \geq \recEndOf(\vec{e}) $]
      We strengthen the claim by also showing
      $ \mem{\upto{\vec{e}}{i}} \in \volatile(q_i) $
      and proceed by induction on~$i$.
      \begin{induction}
        \step[Base case~$ i = \recEndOf(\vec{e}) $]
          As we argued above, we can pick~$q_i=q$
          since $
            \histOf[\lp,\hres{\lp}]{r}(\upto{\vec{e}}{i}) =
            \emptyvec \in \LegalPath{q}{q}.
          $
          By $G.\Init \in \durable(q)$
          and soundness of $\LibImpl[rec]$, we also have that
          \[
            \mem{\upto{\vec{e}}{i}}
            = \mem{\restr{G.\tso}{G.\Init\union\EvOfCid{\recoveryId}}}
            \in \recovered(q) \subs \volatile(q).
          \]
        \step[Inductive step~$ i > \recEndOf(\vec{e}) $]
          By induction hypothesis we have,
          for some $q_{i-1}$,
          $ \mem{\upto{\vec{e}}{i-1}} \in \volatile(q_{i-1}) $
          and $
            \histOf[\lp,\hres{\lp}]{r}(\upto{\vec{e}}{i-1})
              \in \LegalPath{q}{q_{i-1}}.
          $
          First note that
          \[
            \histOf[\lp,\hres{\lp}]{r}(\upto{\vec{e}}{i}) =
            \histOf[\lp,\hres{\lp}]{r}(\upto{\vec{e}}{i-1})
            \concat
            \hres{\lp}(\vec{e})_{@i}
            \concat
            \var{op}
          \]
          where~$\var{op}$ is such that either:
          \begin{enumerate}[label=\roman*),leftmargin=3em]
            \item $\var{op} = \emptyvec$
              if $\vec{e}(i)$ is not a linearization point, or
            \item $\var{op} = \tup{\callOf(c),r(c)}$
              if $\vec{e}(i)$ is the linearization point of~$c$.
          \end{enumerate}
          Before we analyse the two cases,
          we observe that by condition~\ref{cond:hind-linpt} and the induction hypothesis,
          $
            \histOf[\lp,\hres{\lp}]{r}(\upto{\vec{e}}{i-1})
            \concat
            \hres{\lp}(\vec{e})_{@i}
            \in
            \LegalPath{q}{q_{i-1}}.
          $

          We now consider the two cases above.
          \begin{casesplit}
            \case[$ \A c.\lp(c) \ne \vec{e}(i) $]
              Then we can pick~$q_i = q_{i-1}$.
By condition~\ref{cond:stutter}
              we know~$ \mem{\upto{\vec{e}}{i}} \in \volatile(q_{i-1}) $.
              Moreover we get
              \[
                \histOf[\lp,\hres{\lp}]{r}(\upto{\vec{e}}{i})
                = \histOf[\lp,\hres{\lp}]{r}(\upto{\vec{e}}{i-1})
                  \concat
                  \hres{\lp}(\vec{e})_{@i}
                \in \LegalPath{q}{q_{i}}.
              \]
            \case[$ \E c.\lp(c) = \vec{e}(i) $]
              By condition~\ref{cond:linpt-trans}
we know that there is a~$q'$
              such that $ \mem{\upto{\vec{e}}{i}} \in \volatile(q') $
              and $ (q_{i-1},q') \in \Delta(\callOf(c), r(c)) $.
              By letting $q_{i} = q'$ we obtain
              \[
                \histOf[\lp,\hres{\lp}]{r}(\upto{\vec{e}}{i})
                = \histOf[\lp,\hres{\lp}]{r}(\upto{\vec{e}}{i-1})
                  \concat
                  \hres{\lp}(\vec{e})_{@i}
                  \concat
                  \tup{\callOf(c),r(c)}
                \in \LegalPath{q}{q_{i}}.
              \qedhere
              \]
          \end{casesplit}
      \end{induction}
  \end{casesplit}
\end{proof}

\begin{lemma}
\label{lm:perst-validates-legal}
  Consider an execution~$G$ of~$\tup{\LibImpl[op],\LibImpl[rec]}$
  for some \pre\tup{\durable,\recovered}-sound~$\LibImpl[rec]$,
  with $G.\Init \in \durable(q)$.
  Assume the linearization strategy
  $\tup{\pt, r, \restr{\nvo}{G.\Persisted}, \durable}$
  \pre\initOf-validates~$G$.
  Then, for all~$\vec{e} \in \enum[G.\Persisted]{\nvo}$,
  for some~$q'$,
  $
    \histOf[\pt]{r}(\vec{e}) \in \LegalPath{q}{q'}
  $ and $
    \mem{\vec{e}} \in \durable(q').
  $
\end{lemma}

\begin{proof}
  By a straightforward adaptation of the proof of \cref{lm:linpt-validates-legal}.
\end{proof}

\begin{lemma}
\label{lm:abs-hb-to-hb}
  Given~$G$ and a set~$C$ of completion events for~$G$,
  let~$\ret_1, \ret_2 \in G.\Rets \dunion C$.
  We have:
  \[
   \ret_1 \hb[{\abs[C]{G}}]-> \ret_2
   \quad\implies\quad
   \ret_1 \hb[G]-> \min\nolimits_{\po}(G.\EvOfCid{\cidOf(\ret_2)}) .
  \]
\end{lemma}
\begin{proof}
  Let us first prove that
  $\hat{e} = \min\nolimits_{\po}(G.\EvOfCid{\cidOf(\ret_2)})$
  exists and is unique.
  Its uniqueness follows from the fact that calls are sequential
  (cf.~\cref{def:lib-impl-exec}).
  We consider two cases:
  \begin{casesplit}
  \case[$ \ret_2 \notin C $]
    Then $ \ret_2 \in G.\EvOfCid{\cidOf(\ret_2)} \ne \emptyset$,
    so $\hat{e}$ exists.
    Consider the edges justifying
    $(\ret_1, \ret_2) \in \hb[{\abs[C]{G}}]$:
    \[
      \ret_1 =
      e_0 \relto{{\abs[C]{G}}.\po \union {\abs[C]{G}}.\rfe}
      \dots
      \relto{{\abs[C]{G}}.\po \union {\abs[C]{G}}.\rfe}
      e_{n-1} \relto{{\abs[C]{G}}.\po \union {\abs[C]{G}}.\rfe}
      e_n
      \po[G]->
      \ret_2.
    \]
    Since elements of~$C$ have no outgoing edges,
    $ e_0,\dots,e_n \in G.E \setminus \LibEv $.
    As a consequence,
    all the edges between them are also edges of $\hb[G]$:
    \[
      \ret_1 =
      e_0 \hb[G]->
      \dots
      \hb[G]->
      e_{n-1} \hb[G]->
      e_n
      \po[G]->
      \ret_2.
    \]
    Now we have $ \cidOf(e_n) \ne \cidOf(\hat{e}) = \cidOf(\ret_2) $,
    and~$ \hat{e} \po[G]-> \ret_2 $.
    Therefore, by sequentiality of calls,
    $ \cidOf(e_n) \po[G]-> \hat{e} $, which proves the claim.

  \case[$ \ret_2 \in C $]
    Then if there was no other event of the same call,
    we would have no incoming edge in $\ret_2$
    which contradicts $(\ret_1, \ret_2) \in \hb[{\abs[C]{G}}]$.
    Therefore, there is some $e' \in G.\EvOfCid{\cidOf(\ret_2)} \ne \emptyset$, so $\hat{e}$ exists.
    Since the only edges to $\ret_2$ are the ones added by $C_{\po}$
    from events with call id~$\cidOf(\ret_2)$,
    we have that $(\ret_1,e') \in \hb[{\abs[C]{G}}]$.
    Now consider the last edge to reach $e'$:
    \[
      \ret_1
      \hb[G]->
      e''
      \po[G]->
      e'
      \relto{C_{\po}}
      \ret_2.
    \]
    Note that the edge between~$e'' \notin \LibEv$ and $e'\in \LibEv$
    cannot be a $\rfe$ edge as clients do not share locations with libraries.

    Now, if $\hat{e} = e'$ we are done.
    Otherwise we have $ \cidOf(e'') \ne \cidOf(\hat{e}) = \cidOf(e') $,
    and~$ \hat{e} \po[G]-> e' $.
    Therefore, by sequentiality of calls,
    $ \cidOf(e'') \po[G]-> \hat{e} $, which proves the claim.
    \qedhere
  \end{casesplit}
\end{proof}

\begin{definition}[Completion induced by~$\PC$ and~$r$]
\label{def:compl-from-pt}
  Given an execution~$G$,
  a set of persisted call ids
  $\PC \subs G.\CallId$ and
  $r \from G.\CallId \pto \Val$,
  we say a set~$C$ of completion events for~$G$ is
  \emph{induced by}~$\PC$ and~$r$ if
  $ C = \set{ (e_1 \of \Ret{r(c_1)}), \dots, (e_n \of \Ret{r(c_n)}) } $
  where
  $ \PC \setminus \cidOf(G.\Rets) = \set{c_1, \dots, c_n} $,
  and $\A 0<i\leq n.\cidOf(e_i) = c_i$.
Since the set of completion events induced by some~$\PC$ and~$r$
  is unique up to the identity of events,
  we write~$\Compl{\PC}{r}$ for such a set using an arbitrary choice for
  the identity of the events.
\end{definition}

\subsection{The \thename\ Theorem}
\label{sec:full-master-theorem}

\Cref{th:master} is a special case of the following general
theorem.

\begin{theorem}[General \thename\ Theorem]
\label{th:master-full}
  Consider a library with deterministic
  specification~$\tup{\AbsState, \Delta, \initState}$
  and operations implementation~$ \LibImpl[op] $.
  Let~$ \durable,\recovered,\volatile \from \AbsState \to \powerset(\Mem) $
  with $ \A q.\recovered(q) \subs \volatile(q) $.
  To prove
  $\LibImpl[op]$ is \pre\tup{\durable,\recovered}-linearizable,
  it is sufficient to prove the following.
  Fixing arbitrary
\pre\tup{\durable,\recovered}-sound $\LibImpl[rec]$,
          $q\in \AbsState$, and
          $G$ execution of~$\tup{\LibImpl[op], \LibImpl[rec]}$
          with $G.\Init \in\durable(q)$,
find:
\begin{multicols}{2}
  \begin{itemize}
    \item a volatile order $\vo$
    \item a set of events $\HBEv \subs G.E$
    \item $ \lp \from G.\CallId \pto \HBEv $
    \item $ \hres{\lp} \from \enum[G.E]{\vo} \to (G.\CallId \pto \Nat) $
\columnbreak
    \item $ \pt \from G.\CallId \pto G.\Persisted $
    \item $ r \from G.\CallId \pto \V $
    \item a persisted readers set\\$ \PRd \subs G.\CallId \setminus \dom(\pt) $
  \end{itemize}
\end{multicols}
  such that:
  \begin{defenum}[itemsep=.4\baselineskip]
    \item $
          \dom(\pt) \subs \dom(\lp).
        $
      \label{cond:pers-subs-lp}
    \item $
        \A \vec{e} \in {\enum[G.E]{\vo}}.
          \dom(\hres{\lp}(\vec{e})) =
            \PRd \setminus \dom(\lp).
        $
      \label{cond:hind-subs-rets}
    \item $
          \A c \in \PRd.\;
            \Delta(\callOf(c),r(c)) \subs \idOn{\AbsState}.
        $
\label{cond:perst-ret-id}
    \item $
        \ev{\HBEv} \seq \hb \seq \ev{\HBEv}
          \subs \vo
      $
      \label{cond:vo-pres-hb}
    \item $
        \A c \in \dom(\lp).
          \E e_1,e_2 \in {G.\EvOfCid{c}}.
            \smash{e_1 \hb?-> \lp(c) \hb?-> e_2}.
       $
      \label{cond:linpt-hb}
    \item For any~$\vec{e} \in {\enum[G.E]{\vo}}$ and
          $c\in\dom(\hres{\lp}(\vec{e}))$:
      \begin{itemize}
      \item $
          \E i_1,i_2.
            \vec{e}(i_1), \vec{e}(i_2) \in \HBEv
            \land
            c=\cidOf(\vec{e}(i_1))=\cidOf(\vec{e}(i_2))
            \land
            i_1 \leq \hres{\lp}(\vec{e})(c) \leq i_2.
      $
      \end{itemize}
      \label{cond:hind-hb}
    \item $\tup{\lp, \hres{\lp}, r, \vo, \volatile}$
      \pre\recEndOf-validates~$G$.
      \label{cond:linpt-valid}
    \item $
        \cidOf(G.\Rets) \subs \dom(\pt) \union \PRd.
      $
      \label{cond:rets-persist}
    \item For any~$c \in \dom(\lp) \setminus (\dom(\pt) \union \PRd)$,
          and all $\vec{e} \in {\enum[G.E]{\vo}}$:
      \begin{itemize}
      \item If $
          \hres{\lp}[\vec{e}] = \vec{e}' \concat \lp(c) \concat \vec{e}''
        $ then $
            \tup{\callOf(c), r(c)}
        $ is
        \pre \h-voidable,\\
        where $
          \h = \restr{\histOf[\lp]{r}(\vec{e}'')}{(\dom(\pt) \union \PRd)}
        $.
      \end{itemize}
      \label{cond:voided-linpt-commute}
    \item For any~$c, c'\in \dom(\pt)$, either:
      \begin{itemize}
      \item $
          \tup{\callOf(c), r(c)}
            \comm{\AbsState}{\Delta}
          \tup{\callOf(c'), r(c')}
        $, or
      \item $
          \smash{\pt(c) \nvo-> \pt(c')}
          \implies
          \smash{\lp(c) \vo-> \lp(c')}.
        $
      \end{itemize}
      \label{cond:vo-nvo-agree-commute}\item $
      \histOf[\pt]{r}(\enum[G.\Persisted]{\nvo})
      \in \LegalFrom{q}
        \implies
          \tup{\pt, r, \restr{\nvo}{G.\Persisted}, \durable}
      $ \pre\initOf-validates~$G$.
      \label{cond:perst-valid}
  \end{defenum}
\end{theorem}

Before producing a proof, let us explain the conditions of the theorem.
Most conditions are similar to \cref{th:master}.
To prove \pre\tup{\durable,\recovered}-linearizability of~$G$,
we have to provide two linearization strategies:
the volatile one~$\tup{\lp, \hres{\lp}, r, \vo, \volatile}$
and the persistent one~$\tup{\pt, r, \restr{\nvo}{G.\Persisted}, \durable}$.
Since the persistency points must be durable events,
$\dom(\pi)$ only represents the persisted ``writer'' calls,
\ie calls that induce an abstract state change.
The persisted readers set $\PRd$ indicates which of the ``reader'' calls
should be considered (logically) persisted.
Finally, we need to provide a set of events~$X$ which includes all the linearization events; we will explain the meaning of this set shortly.

Then the theorem asks us to check a number of conditions.
Condition \ref{cond:pers-subs-lp} checks
that every persisted writer has a linearization point.
Condition \ref{cond:hind-subs-rets}
says that all the readers that persisted which do not have a linearization point
are linearized by hindsight.
Condition \ref{cond:perst-ret-id}
ensures that the persistent readers are indeed readers.

The goal of conditions
\ref{cond:vo-pres-hb}, \ref{cond:linpt-hb} and \ref{cond:hind-hb}
is to make sure the final linearization respects
the~$\hb$ order between the calls, as required by \cref{def:abs-exec}.
In particular, we do not want the linearization to contradict
the~$\po$ ordering between calls.
In fact~$\vo$ might not preserve~$\hb$ in general:
typically $\vo$ reconstructs a global notion of time which
might contradict some $\po$ edges,
for examples the ones between write and read events.
Condition \ref{cond:vo-pres-hb}
states that~$\vo$ has to preserve~$\hb$ on~$X$.
Typically, linearization points are reads or updates,
and one can set $X = \UReads$.
Condition \ref{cond:linpt-hb}
requires a linearization point to be \hb-between two events of the call it linearizes.
This ensures that an~$\hb$ edge between calls implies an $\hb$ edge between
the respective linearization points.
If linearization points are local to the call they linearize the condition
is trivially satisfied.
Condition \ref{cond:hind-hb}
achieves the same effect for hindsight linearizations.

The conditions so far represent well-formedness constraints
that usually hold by construction.
The conditions that follow are the ones where the actual algorithmic insight is involved.

Condition \ref{cond:linpt-valid}
requires to prove regular volatile linearizability
using the volatile linearization strategy.
This is typically a straightforward adaptation
of the proof of a non-durable version of the same data structure
(\eg Harris list for the link-free set).
The adaptation simply needs to ensure that the events added
to make the data structure durable preserve the encoded state when executed.

When making a linearizable data structure durable,
after having decided how the persisted memory represents an abstract state (through~$\durable$),
the main design decision is which flushes to issue and when.
Conditions~\ref{cond:rets-persist},
\ref{cond:voided-linpt-commute}, and
\ref{cond:vo-nvo-agree-commute}
are the ones that check that all the flushes needed for correctness are issued
by the operations.
This can help guiding optimizations,
as redundant flushes would be the ones not contributing to the proofs of these conditions.

Condition~\ref{cond:rets-persist} reflects the
\emph{unbuffered} nature of durable linearizability:
it requires every returned call to be considered as persisted.
The only way to enforce a writer operation is persisted before returning
is to issue a (synchronous) flush on the address of the persistency point of
the call, which we dub the ``flush before return'' policy.

Condition \ref{cond:voided-linpt-commute}
handles voided calls, \ie calls that have been linearized by the volatile argument somewhere the sequence, but have not persisted.
What the condition asks is to prove that voided calls are voidable
with respect to the rest of the linearization ahead of them.
More precisely, one considers the linearization induced by
$\vo$ and the hindsight resolution ($\hres{\lp}[\vec{e}]$)
and finds the linearization point of a voided call $c$.
Then one needs to check that the call is voidable with respect to
the volatile history of non-voided calls ahead ($\h$).
This is typically enforced by making sure that possibly conflicting calls
in $\h$ issue a flush on the address of the persistency point of~$c$
before persisting themselves: if that is the case,
then either $\h$ does not contain conflicting calls, or
$c$ has been persisted, which means it is not voided.

Condition \ref{cond:vo-nvo-agree-commute}
considers the case of two persisted calls that do not abstractly commute
(\eg two calls on the same key).
For those pairs of calls,
we have to prove that the persistent linearization strategy
and the volatile one agree on the ordering of the calls.
This is also achieved by issuing flushes strategically:
an operation has to make sure to flush the persistency point of earlier
calls that do not commute with it.

When the conditions presented so far hold,
the legality of the volatile linearization
implies the legality of the persistent linearization.
What remains to prove is that the persistency points modify the persisted memory
so that it encodes the output state expected given the legal linearization.
This is checked by condition \ref{cond:perst-valid},
which allows us to assume the persisted linearization is legal
(a fact that follows from the other conditions)
and asks us to prove by induction on $\nvo$ that each persistency point
modifies the memory so that the encoded state reflects the expected change.
Proving this does not involve flushes,
but merely checks that the persistency points enact changes that are compatible
with the durable state interpretation~$\durable$.

\begin{proof}[Proof of the {\hyperref[th:master-full]{General \thename\ Theorem}}]
  We prove that the assumptions of the theorem
  imply $\LibImpl[op]$ is \pre\tup{\durable,\recovered}-linearizable,
  by finding a~$q'$ such that
  there is an abstract execution $\tup{\abs{G}, \lin}$ of~$G$
  with $ \hist(\enum{\lin}) \in \LegalPath{q}{q'} $, and
  $\mem{\restr{G.\nvo}{G.\Persisted}} \in \durable(q')$.
  \begin{proofsteps}
  \step[Step~1: a legal \vo-induced linearization]
    Pick an arbitrary~$ \vec{e}_1 \in G.\enum{\vo} $
    (which exists by acyclicity of $\vo$).
    By applying \cref{lm:linpt-validates-legal} to \cref{cond:linpt-valid},
    we obtain that
    \begin{equation}
      \histOf[\lp,\hres{\lp}]{r}(\vec{e}_1) \in \LegalFrom{q}.
      \label{prop:master:e1-legal}
    \end{equation}
    Let $ \vec{c}_1 = \cidHistOf[\lp](\hres{\lp}[\vec{e}_1]) $;
    by definition, we have $
      \histOf[\lp,\hres{\lp}\!]{r}(\vec{e}_1) = \histOf{r}(\vec{c}_1).
    $
  \step[Step~2: filtering voided calls]
    Let $ {V = \dom(\lp) \setminus (\dom(\pt) \union \PRd)} $,
    \ie the set of calls that were voided by the crash.
    Note that, by \ref{cond:hind-subs-rets},
    no call that is linearized by~$\hres{\lp}$ is voided.
    Let $ \vec{c}_2 = \vec{c}_1 \setminus V $,
    \ie the calls that have been persisted and survive the crash.
    We will now show that
    \begin{equation}
      \histOf{r}(\vec{c}_2) \in \LegalFrom{q}
      \label{prop:master:c2-legal}
    \end{equation}

    The idea is that for each of the events in~$V$,
    \ref{cond:voided-linpt-commute} applies,
    and we can use it to iteratively remove them from the sequence.

    To prove it, first note that $
    \histOf{r}(\vec{c}_2)
    =
    \histOf[\lp,\hres{\lp}\!]{r}
      (\vec{e}_1 \setminus \lp(V))
    =
    \histOf[\lp]{r}
      (\hres{\lp}[\vec{e}_1] \setminus \lp(V)).
    $
    For any $n\leq \card{V}$, let $v_1, \dots, v_n \in \lp(V)$ be the last~$n$
    occurrences of elements of~$\lp(V)$ in $\vec{e}_1$, ordered from last to first.
    We prove by induction that for all $n\leq \card{V}$,
    $
      \histOf[\lp]{r}
        (\hres{\lp}[\vec{e}_1] \setminus \set{v_1,\dots,v_n})
      \in \LegalFrom{q}.
    $
    \begin{induction}
    \step[Base case $n=0$] Trivial from \eqref{prop:master:e1-legal}.
    \step[Inductive step~$n>0$]
      Let~$j$ be such that $\hres{\lp}[\vec{e}_1](j) = v_n$.
      We have, by induction hypothesis, that
      $
        \histOf[\lp]{r}
          (\hres{\lp}[\vec{e}_1] \setminus \set{v_1,\dots,v_{n-1}})
        \in \LegalFrom{q}.
      $
      Moreover,
      \[
        \hres{\lp}[\vec{e}_1] \setminus \set{v_1,\dots,v_{n-1}}
        =
          \upto{\hres{\lp}[\vec{e}_1]}{j-1}
        \concat
        \lp(v_n) \concat
        \bigl(
          \tailfrom{\hres{\lp}[\vec{e}_1]}{j+1}.
            \setminus \set{v_1,\dots,v_{n-1}}
        \bigr)
      \]
      This implies that there is some $q' \in \AbsState$ such that
      $
        \histOf[\lp]{r}
        (\upto{\hres{\lp}[\vec{e}_1]}{j-1})
          \in \LegalPath{q}{q'}
      $ and $
        \histOf[\lp]{r}(
          \lp(v_n) \concat (
            \tailfrom{\hres{\lp}[\vec{e}_1]}{j+1}
            \setminus \set{v_1,\dots,v_{n-1}}
          ))
          \in \LegalFrom{q'}.
      $
      Notice that by construction
      \[
        \histOf[\lp]{r}(
          \tailfrom{\hres{\lp}[\vec{e}_1]}{j+1}
          \setminus \set{v_1,\dots,v_{n-1}}
        )
        =
        \restr{
          \histOf[\lp]{r}(
            \tailfrom{\hres{\lp}[\vec{e}_1]}{j+1}
          )
        }{(\dom{\pt}\union \PRd)}.
      \]
      By \ref{cond:voided-linpt-commute}, we then have
      $
        \histOf[\lp]{r}
          (\tailfrom{\hres{\lp}[\vec{e}_1]}{j+1}
            \setminus \set{v_1,\dots,v_{n-1}})
            \in \LegalFrom{q'}.
      $
      We therefore have:
      \begin{align*}
        \histOf[\lp]{r}
          (\hres{\lp}[\vec{e}_1] \setminus \set{v_1,\dots,v_n})
        &=
        \histOf[\lp]{r}
          (\upto{\hres{\lp}[\vec{e}_1]}{j-1}
          \concat
          (\tailfrom{\vec{e}_1}{j+1} \setminus \set{v_1,\dots,v_{n-1}}))
        \\&=
          \histOf[\lp]{r}(\upto{\hres{\lp}[\vec{e}_1]}{j-1})
          \concat
          \histOf[\lp]{r}(\tailfrom{\vec{e}_1}{j+1} \setminus \set{v_1,\dots,v_{n-1}})
          \\&
        \in \LegalFrom{q}
      \end{align*}
    \end{induction}
  \step[Step~3: the abstract execution]
    We now construct an abstract execution of~$G$ from $\vec{c}_2$.
    Let $C = \Compl{\dom(\pt) \union \PRd}{r}$,
    and $P_{\Rets} = G.\Rets \union C$.
    The sequence~$\vec{c}_2$ induces
    the relation~$\lin  \subs P_{\Rets} \times P_{\Rets}$
    defined as
    \[
      \lin \is
        \set{ (\ret_1,\ret_2) |
                \E i,j.i < j\land
                \vec{c}_2(i) = \cidOf(\ret_1)\land
                \vec{c}_2(j) = \cidOf(\ret_2)
}.
    \]
    By \ref{cond:pers-subs-lp}, \ref{cond:hind-subs-rets}, and \ref{cond:rets-persist}, we have:
    \begin{equation}
      \ret \in P_{\Rets}
      \iff
      \E i. \vec{c}_2(i) = \cidOf(\ret),
      \label{cond:linpt-verif:lin-c2-dom}
    \end{equation}
    which makes~$\lin$ a strict total order on~$P_{\Rets}$.
    Moreover:
    \begin{equation}
      \hist(\enum{\lin}) = \histOf{r}(\vec{c}_2)
\label{cond:linpt-verif:lin-c2}
    \end{equation}
    To show $ \tup{\abs[C]{G}, \lin} $ is a strong
    abstract execution of~$G$ we have to prove that
    $ \restr{\abs[C]{G}.\hb}{P_{\Rets}} \subs \lin $.
    Take $(\ret_1, \ret_2) \in \restr{\abs[C]{G}.\hb}{P_{\Rets}}$.
    By \eqref{cond:linpt-verif:lin-c2-dom} there are~$i_1$ and~$i_2$ such that
    $ \vec{c}_1(i_1) = \cidOf(\ret_1) $
    and
    $ \vec{c}_2(i_2) = \cidOf(\ret_2) $.
    To prove $(\ret_1, \ret_2) \in \lin$ we have to prove~$i_1 < i_2$.
    By \cref{lm:orig-mono}, it suffices to show that
    $
      \orig_{\lp,\hres{\lp}\!}(\vec{e}_1, \vec{c}_2(i_1))
      < \orig_{\lp,\hres{\lp}\!}(\vec{e}_1, \vec{c}_2(i_2))
    $, that is the linearization points of
    $\vec{c}_2(i_1)$ and $\vec{c}_2(i_2)$
    appear in the required order in~$\vec{e}_1$.
    We will show that there exist some~$i_1',i_2'$ such that:
    \begin{align}
\orig_{\lp,\hres{\lp}\!}(\vec{e}_1, \vec{c}_2(i_1)) \leq i_1' \land{}&
        \vec{e}_1(i_1') \in \HBEv
        \land
        \vec{e}_1(i_1') \hb[G]?-> \ret_1
      \label{cond:linpt-verif:post-i1}
      \\
i_2' \leq \orig_{\lp,\hres{\lp}\!}(\vec{e}_1, \vec{c}_2(i_2)) \land{}&
        \vec{e}_1(i_2') \in \HBEv
        \land
        \min_{\po}(G.\EvOfCid{\cidOf(\ret_2)}) \hb[G]?-> \vec{e}_1(i_2')
      \label{cond:linpt-verif:pre-i2}
    \end{align}
    Then, by \cref{lm:abs-hb-to-hb}, we have
    $ (\ret_1, \min_{\po}(G.\EvOfCid{\cidOf(\ret_2)})) \in \hb[G] $,
    which combined with the above gives us
    $ {(\vec{e}_1(i_1'), \vec{e}_1(i_2')) \in \hb[G]} $;
    by \ref{cond:vo-pres-hb} we obtain
    $ {(\vec{e}_1(i_1'), \vec{e}_1(i_2')) \in \vo} $
    which implies $ i_1' < i_2' $.
    All in all, we get $ \orig_{\lp,\hres{\lp}\!}(\vec{e}_1, \vec{c}_2(i_1)) \leq i_1' < i_2' \leq \orig_{\lp,\hres{\lp}\!}(\vec{e}_1, \vec{c}_2(i_2)) $
    as desired.
    To conclude the step we then prove \eqref{cond:linpt-verif:post-i1} and \eqref{cond:linpt-verif:pre-i2}.
    \begin{casesplit}
      \case*[Proof of~\eqref{cond:linpt-verif:post-i1}]
        Let $ \orig_{\lp,\hres{\lp}\!}(\vec{e}_1, \vec{c}_2(i_1)) = j $.
        According to the definition of~$\orig$ (\cref{def:orig}),
        we must consider two cases.
        \begin{casesplit}
        \case[$ \vec{e}_1(j) = \lp(\vec{c}_2(i_1)) $]
          By~\ref{cond:linpt-hb} we have that there is some $ e $
          with $ \cidOf(e) = \vec{c}_2(i_1) $ and
          $ \lp(\vec{c}_2(i_1)) \hb?-> e $.
          Since $e$ has the same call id of $\ret_1$ we also know
          $ e \po?-> \ret_1 $.
          Overall, we obtain that $\lp(\vec{c}_2(i_1)) \in \HBEv$
          and $ \lp(\vec{c}_2(i_1)) \hb[G]-> \ret_1 $,
          and therefore we can set~$i_1' = j$ and prove the claim.
        \case[$ \hres{\lp}(\vec{e}_1)(\vec{c}_2(i_1)) = j $]
          By~\ref{cond:hind-hb} we have that there is some $ j' $
          with
            $ \vec{e}_1(j') \in \HBEv $,
            $ \cidOf(\vec{e}_1(j')) = \vec{c}_2(i_1) $, and
            $ j \leq j' $.
          From the call ids, we know $ \vec{e}_1(j') \po-> \ret_1 $
          and so in particular $ \vec{e}_1(j') \hb[G]-> \ret_1 $.
          Therefore we can set~$i_1' = j'$ and prove the claim.
        \end{casesplit}
      \case*[Proof of~\eqref{cond:linpt-verif:pre-i2}]
        Similar to the previous case,
        by doing a symmetric case analysis.
    \end{casesplit}

  \step[Step~4: irrelevance of readers]
    Before we can compare the sequence~$ \vec{c}_2 $ with that
    induced by~$\nvo$, we need to remove from it the ``persisted readers'',
    since they are not meaningfully ordered by~$\nvo$.
    Let~$ \vec{c}_3 \is \vec{c}_2 \setminus \PRd $.
    By \cref{cond:perst-ret-id},
removing calls from~$\PRd$ preserves legality of histories:
    \begin{equation}
      \histOf{r}(\vec{c}_2)
        \himplies{\AbsState}{\Delta}
      \histOf{r}(\vec{c}_3).
      \label{prop:master:c2-equiv-c3}
    \end{equation}
    Note that, by \ref{cond:rets-persist}, \ref{cond:pers-subs-lp} and
    \ref{cond:hind-subs-rets},
    the calls in $\vec{c}_3$ are exactly those in~$\dom(\pt)$.

  \step[Step~5: the \nvo-induced linearization is legal]
    Let~$ \vec{e}_\pers = G.\enum[\rng(\pt)]{\nvo} $,
    be the \nvo-ordered sequence of persisted persistency points.
    By letting $
      \vec{c}_\pers = \cidOf(\vec{e}_\pers(0)) \dots \cidOf(\vec{e}_\pers(n))
    $, where $n=\len{\vec{e}_\pers}-1$,
    we have:
    \[
      \histOf[\pt]{r}(\enum{G.\restr{\nvo}{G.\Persisted}})
      =
      \histOf[\pt]{r}(\vec{e}_\pers)
      =
      \histOf{r}(\vec{c}_\pers)
    \]
We now prove that
    \begin{equation}
      \histOf{r}(\vec{c}_\pers)
        \hequiv{\AbsState}{\Delta}
      \histOf{r}(\vec{c}_3).
      \label{prop:master:cp-equiv-c3}
    \end{equation}
Note that, since~$\nvo$ orders all events in $\dom(\pt)$,
    $\vec{c}_3$ is a permutation of $\vec{c}_\pers$.
    We prove, for all~$i \leq n$,
    that there exists a $\vec{c}$
    which is a scattered subsequence of $ \vec{c}_\pers $
    with $ \len{\vec{c}} = i $
    such that:\footnote{We abuse of notation and in an expression
      $ \vec{c}_1 \setminus \vec{c}_2 $
      we implicitly coerce~$\vec{c}_2$ to the set of its items.}
    \begin{equation}
      \A \pr{c} \in (\vec{c}_3 \setminus \vec{c}),\dpr{c} \in \vec{c}.
        \smash{\lp(\pr{c}) \vo/-> \lp(\dpr{c})}
      \quad\text{and}\quad
      \histOf{r}(\vec{c} \concat (\vec{c}_3 \setminus \vec{c}))
        \hequiv{\AbsState}{\Delta}
      \histOf{r}(\vec{c}_3).
      \label{goal:linpt-verif:insert}
    \end{equation}
We proceed by induction on~$i$.
    Intuitively, $\vec{c}$ represents an \nvo-respecting
    subsequence of calls constructed so far and we are inserting in it
    the elements of $ \vec{c}_3 $ one by one.
    Every time we have to swap calls to insert the element,
    the swap will be shown to preserve the equivalence of the histories
    by appealing to \cref{cond:vo-nvo-agree-commute}.
    \begin{induction}
      \step[Base case~$i=0$]
        Trivial as $
          \emptyvec \concat (\vec{c}_3\setminus\emptyset)
          = \vec{c}_3.
        $
      \step[Induction step~$i+1$]
        The induction hypothesis gives us that if $i\leq n$ then
        there exists a $\vec{c}^i$ with $ \len{\vec{c}^i}=i $
        such that
        \begin{gather}
          \A \pr{c} \in (\vec{c}_3 \setminus \vec{c}^i),\dpr{c} \in \vec{c}^i.
            \smash{\lp(\pr{c}) \vo/-> \lp(\dpr{c})}
          \label{prop:master:ind-hyp-vo}
          \\
          \histOf{r}(\vec{c}^i \concat (\vec{c}_3 \setminus \vec{c}^i))
            \hequiv{\AbsState}{\Delta}
          \histOf{r}(\vec{c}_3).
          \label{prop:master:ind-hyp-equiv}
        \end{gather}
        We now construct a~$\vec{c}$ that proves the statement for~$i+1$.
        Assume $ i+1 \leq n $; then $\len{\vec{c}^i}=i<n$ and therefore
        $ \len{(\vec{c}_3 \setminus \vec{c}^i)} > 0 $.
        Let $ c = (\vec{c}_3 \setminus \vec{c}^i)(0) $.
        First note that, by construction,
        the sequence of linearization points of the calls in $\vec{c}_3$
        respects $\vo$,
        and thus:
        \[
          \A \pr{c}\in\vec{c}_3\setminus(\vec{c}^i \union \set{c}).
          \lp(\pr{c}) \vo/-> \lp(c).
        \]
        and so if we construct~$\vec{c}$ by inserting~$c$ in $\vec{c}^i$,
        the first conjunct of \eqref{goal:linpt-verif:insert} will hold.

        Our goal is to insert~$c$ in~$\vec{c}^i$ so that
        the resulting sequence is \nvo-respecting and
        keeps the histories equivalent.
        We do this by considering~$ \vec{c}^i \concat c $ and
        swapping~$c$ with its predecessor until we do not violate any
        $\nvo$ edges.
        If~$i=0$ that is trivial as $ \vec{c}^i = \emptyvec $
        and the desired sequence is just~$c$.
        Otherwise,
        let~$c'$ be the last item of $\vec{c}^i$,
        \ie $c' = \vec{c}^i(i-1)$.
        If $ \pt(c') \nvo-> \pt(c) $, $ \vec{c}^i \concat c $ is the desired sequence.
        If not, then $ \pt(c) \nvo-> \pt(c') $, so by \ref{cond:vo-nvo-agree-commute}
        we have two cases:
        either $ {\lp(c) \vo-> \lp(c')} $
        or $
            \tup{\callOf(c), r(c)}
              \comm{\AbsState}{\Delta}
            \tup{\callOf(c'), r(c')}
          $.
        The first case is not possible as it contradicts
        the induction hypothesis~\eqref{prop:master:ind-hyp-vo}.
        In the remaining case, we can swap $c$ and $c'$ obtaining
        some sequence~$\upto{\vec{c}^i}{i-2} \concat c c'$
        with
        $
          \histOf{r}(
              \upto{\vec{c}^i}{i-2} \concat c c' \concat
              (\vec{c}_3 \setminus (\vec{c}^i \union \set{c}))
            )
            \hequiv{\AbsState}{\Delta}
          \histOf{r}(\vec{c}^i \concat (\vec{c}_3 \setminus \vec{c}^i)).
        $
        By iterating the argument on~$c$ and $ \upto{\vec{c}^i}{i-2} $
        we can find a suitable position for~$c$ while preserving the history equivalences.
    \end{induction}

  \step[Step~6: persisted memory encodes the final state]
    By \eqref{prop:master:c2-legal},
       \eqref{prop:master:c2-equiv-c3}, and
       \eqref{prop:master:cp-equiv-c3},
    we obtain
    \[
      \LegalFrom{q} \ni
      \histOf{r}(\vec{c}_2)
        \himplies{\AbsState}{\Delta}
      \histOf{r}(\vec{c}_3)
        \hequiv{\AbsState}{\Delta}
      \histOf{r}(\vec{c}_\pers)
        =
      \histOf[\pt]{r}(\enum{G.\restr{\nvo}{G.\Persisted}})
\]
    This implies that there exists a~$q'$ such that
    \begin{align*}
      \histOf{r}(\vec{c}_2)
      \in \LegalPath{q}{q'}
      &&
      \histOf[\pt]{r}(\enum{\restr{G.\nvo}{G.\Persisted}})
      \in \LegalPath{q}{q'}
    \end{align*}
    We can then apply \ref{cond:perst-valid} to get that
    $ \tup{\pt, r, G.\restr{\nvo}{G.\Persisted}, \durable} $
    \pre\initOf-validates~$G$.
    By \cref{lm:perst-validates-legal} this implies
    that for some~$q''$:
    \begin{align*}
      \histOf[\pt]{r}(\enum{\restr{G.\nvo}{G.\Persisted}}) \in \LegalPath{q}{q''}
      &&
      \mem{\enum{\restr{G.\nvo}{G.\Persisted}}} \in \durable(q'').
    \end{align*}
    Since $\Delta$ is assumed to be deterministic, we have~$q' = q''$.
    From the equivalences above and \eqref{cond:linpt-verif:lin-c2}
    we then have:
    \begin{align*}
      \hist(\enum{\lin}) =
      \histOf{r}(\vec{c}_2) \in \LegalPath{q}{q'}
      &&
      \mem{\enum{\restr{G.\nvo}{G.\Persisted}}}
      \in \durable(q').
    \end{align*}
    which completes the proof that
    $\LibImpl[op]$ is \pre\tup{\durable,\recovered}-linearizable.
  \qedhere
  \end{proofsteps}
\end{proof}

\subsection{The Persist-First \thename\ Theorem}
\label{sec:full-pers-first-master-theorem}

\begin{theorem}[Persist-First \thename\ Theorem]
\label{th:pf-master}
  Consider a library with deterministic
  specification~$\tup{\AbsState, \Delta, \initState}$
  and operations implementation~$ \LibImpl[op] $.
  Let~$ \durable,\recovered,\volatile \from \AbsState \to \powerset(\Mem) $
  with $ \A q.\recovered(q) \subs \volatile(q) $.
  To prove
  $\LibImpl[op]$ is \pre\tup{\durable,\recovered}-linearizable,
  it is sufficient to prove the following.
  Fixing arbitrary:
  \begin{itemize}
    \item \pre\tup{\durable,\recovered}-sound $\LibImpl[rec]$,
    \item $q\in \AbsState$, and
    \item $G$ execution of~$\tup{\LibImpl[op], \LibImpl[rec]}$
          with $G.\Init \in\durable(q)$,
  \end{itemize}
  find:
  \begin{itemize}
    \item a volatile order $\vo$
    \item a set of events $\HBEv \subs G.E$
    \item $ \lp \from G.\CallId \pto \HBEv $
    \item $ \hres{\lp} \from \enum[G.E]{\vo} \to (G.\CallId \pto \Nat) $
    \item $ \pt \from G.\CallId \pto G.\Persisted $
    \item $ r \from G.\CallId \pto \V $
    \item a persisted readers set $ \PRd \subs G.\CallId \setminus \dom(\pt) $
  \end{itemize}
  such that:
  \begin{defenum}[itemsep=.4\baselineskip]
    \item $
          \dom(\lp) \subs \dom(\pt).
        $
      \label{pf-master:pers-subs-lp}
    \item $
        \A \vec{e} \in {\enum[G.E]{\vo}}.
          \dom(\hres{\lp}(\vec{e})) =
            \PRd \setminus \dom(\lp).
        $
      \label{pf-master:hind-subs-rets}
    \item $
          \A c \in \PRd.\;
            \Delta(\callOf(c),r(c)) \subs \idOn{\AbsState}.
        $
\label{pf-master:perst-ret-id}
    \item $
        \ev{\HBEv} \seq \hb \seq \ev{\HBEv}
          \subs \vo
      $
      \label{pf-master:vo-pres-hb}
    \item $
        \A c \in \dom(\lp).
          \E e_1,e_2 \in {G.\EvOfCid{c}}.
            \smash{e_1 \hb?-> \lp(c) \hb?-> e_2}.
       $
      \label{pf-master:linpt-hb}
    \item For any~$\vec{e} \in {\enum[G.E]{\vo}}$ and
          $c\in\dom(\hres{\lp}(\vec{e}))$:
      \begin{itemize}
      \item $
          \E i_1,i_2.
            \vec{e}(i_1), \vec{e}(i_2) \in \HBEv
            \land
            c=\cidOf(\vec{e}(i_1))=\cidOf(\vec{e}(i_2))
            \land
            i_1 \leq \hres{\lp}(\vec{e})(c) \leq i_2.
      $
      \end{itemize}
      \label{pf-master:hind-hb}
    \item $\tup{\lp, \hres{\lp}, r, \vo, \volatile}$
      \pre\recEndOf-validates~$G$.
      \label{pf-master:linpt-valid}
    \item $
        \cidOf(G.\Rets) \subs \dom(\lp) \union \PRd.
      $
      \label{pf-master:rets-persist}
    \item For all $c \in \dom(\pt) \setminus \dom(\lp)$
        and $\vec{e} \in {\enum[G.E]{\vo}}$,
        if $ \vec{e} = \vec{e}' \concat \pt(c) \concat \vec{e}'' $ then:
        \begin{itemize}
          \item $\tup{\callOf(c), r(c)}$ is \pre\histOf[\lp]{r}(\vec{e}'')-appendable.
          \item if\/ $
            \histOf[\lp]{r}(\vec{e}')
              \in \LegalPath{q}{q'}
            $ for some~$q'$,
then $
              \E q''.
                (q',q'') \in \Delta(\callOf(c), r(c))
            $.
        \end{itemize}
        Additionally, for all
        $ c, c' \in \dom(\pt) \setminus \dom(\lp) $,
        if $ \pt(c') \nvo-> \pt(c) $
        then $\tup{\callOf(c), r(c)}$ is \pre\tup{\callOf(c'), r(c')}-appendable.
      \label{pf-master:appendable}
    \item For any~$c, c'\in \dom(\pt)$, either:
      \begin{itemize}
      \item $
          \tup{\callOf(c), r(c)}
            \comm{\AbsState}{\Delta}
          \tup{\callOf(c'), r(c')}
        $, or
      \item
        $ c,c'\in \dom(\lp) $
        and
        $
          \smash{\pt(c) \nvo-> \pt(c')}
          \implies
          \smash{\lp(c) \vo-> \lp(c')}
        $, or
      \item
        $ c\in\dom(\lp), c'\notin \dom(\lp) $
        and
        $ \smash{\pt(c) \nvo-> \pt(c')} $, or
      \item
        $ c, c' \notin \dom(\lp) $.
\end{itemize}
      \label{pf-master:vo-nvo-agree-commute}\item $
      \histOf[\pt]{r}(\enum[G.\Persisted]{\nvo})
      \in \LegalFrom{q}
        \implies
          \tup{\pt, r, \restr{\nvo}{G.\Persisted}, \durable}
      $ \pre\initOf-validates~$G$.
      \label{pf-master:perst-valid}
  \end{defenum}
\end{theorem}

\begin{proof}
  We prove that the assumptions of the theorem
  imply $\LibImpl[op]$ is \pre\tup{\durable,\recovered}-linearizable,
  by finding a~$q'$ such that
  there is an abstract execution $\tup{\abs{G}, \lin}$ of~$G$
  with $ \hist(\enum{\lin}) \in \LegalPath{q}{q'} $, and
  $\mem{\restr{G.\nvo}{G.\Persisted}} \in \durable(q')$.
  \begin{proofsteps}
  \step[Step~1: a legal \vo-induced linearization]
    Pick an arbitrary~$ \vec{e}_1 \in G.\enum{\vo} $
    (which exists by acyclicity of $\vo$).
    By applying \cref{lm:linpt-validates-legal} to \cref{pf-master:linpt-valid},
    we obtain that, for some~$q_1$:
    \begin{equation}
      \histOf[\lp,\hres{\lp}]{r}(\vec{e}_1)
        \in \LegalPath{q}{q_1}.
      \label{prop:pf-master:e1-legal}
    \end{equation}
    Let $ \vec{c}_1 = \cidHistOf[\lp](\hres{\lp}[\vec{e}_1]) $;
    by definition, we have $
      \histOf[\lp,\hres{\lp}\!]{r}(\vec{e}_1) = \histOf{r}(\vec{c}_1).
    $
  \step[Step~2: adding prematurely persisted calls]
  By~\ref{pf-master:appendable} we know each call~$c$
  in $\dom(\pt) \setminus \dom(\lp)$ (the prematurely persisted)
  is enabled at $q_1$,
  that is~$\E q_c. (q_1,q_c) \in \Delta(\callOf(c), r(c))$.
  Let
  $\vec{c}_1' = \enum{\rel[cidnvo]} $
  where
  $
    \rel[cidnvo] = \set{
      (c,c') | c,c'\in\dom(\pt) \setminus \dom(\lp),
               \pt(c) \nvo-> \pt(c')
    }.
  $
  A simple induction using the second part of~\ref{pf-master:appendable}
  shows that $ \vec{c}_1' \in \LegalFrom{q_1} $.
  This sequence can be appended to $\vec{c}_1$ obtaining
  $ \vec{c}_2 = \vec{c}_1 \concat \vec{c}_1' $ with
  \begin{equation}
    \histOf{r}(\vec{c}_2) \in \LegalFrom{q}
    \label{prop:pf-master:c2-legal}
  \end{equation}
  \step[Step~3: the abstract execution]
    We now construct an abstract execution of~$G$ from $\vec{c}_2$.
    Let $C = \Compl{\dom(\pt) \union \PRd}{r}$,
    and $P_{\Rets} = G.\Rets \union C$.
    The sequence~$\vec{c}_2$ induces
    the relation~$\lin  \subs P_{\Rets} \times P_{\Rets}$
    with
    \begin{equation}
      \hist(\enum{\lin}) = \histOf{r}(\vec{c}_2)
\label{prop:pf-master:lin-c2}
    \end{equation}
    as in the proof of \masterthm.
    To show $ \tup{\abs[C]{G}, \lin} $ is a strong
    abstract execution of~$G$ we can proceed exactly
    as in the proof of \masterthm,
    with the added observation that
    by \ref{pf-master:rets-persist} there are no
    edges of $\restr{\abs[C]{G}.\hb}{P_{\Rets}}$
    starting from a prematurely persisted call.

  \step[Step~4: irrelevance of readers]
    Before we can compare the sequence~$ \vec{c}_2 $ with that
    induced by~$\nvo$, we need to remove from it the ``persisted readers'',
    since they are not meaningfully ordered by~$\nvo$.
    Let~$ \vec{c}_3 \is \vec{c}_2 \setminus \PRd $.
    By \cref{pf-master:perst-ret-id},
removing calls from~$\PRd$ preserves legality of histories:
    \begin{equation}
      \histOf{r}(\vec{c}_2)
        \himplies{\AbsState}{\Delta}
      \histOf{r}(\vec{c}_3).
      \label{prop:pf-master:c2-equiv-c3}
    \end{equation}
    Note that, by \ref{pf-master:rets-persist}, \ref{pf-master:pers-subs-lp} and
    \ref{pf-master:hind-subs-rets},
    the calls in $\vec{c}_3$ are exactly those in~$\dom(\pt)$.
    More precisely,
    $ \vec{c}_3 = \vec{c}_3' \concat \vec{c}_1'$
    where $\vec{c}_3' \subs \dom(\lp) \inters \dom(\pt)$.

  \step[Step~5: the \nvo-induced linearization is legal]
    Let~$ \vec{e}_\pers = G.\enum[\rng(\pt)]{\nvo} $,
    be the \nvo-ordered sequence of persisted persistency points.
    By letting $
      \vec{c}_\pers = \cidOf(\vec{e}_\pers(0)) \dots \cidOf(\vec{e}_\pers(n))
    $, where $n=\len{\vec{e}_\pers}-1$,
    we have:
    \[
      \histOf[\pt]{r}(\enum{G.\restr{\nvo}{G.\Persisted}})
      =
      \histOf[\pt]{r}(\vec{e}_\pers)
      =
      \histOf{r}(\vec{c}_\pers)
    \]
We now prove that
    \begin{equation}
      \histOf{r}(\vec{c}_\pers)
        \hequiv{\AbsState}{\Delta}
      \histOf{r}(\vec{c}_3).
      \label{prop:pf-master:cp-equiv-c3}
    \end{equation}
    Note that, since~$\nvo$ orders all events in $\dom(\pt)$,
    $\vec{c}_3$ is a permutation of $\vec{c}_\pers$.

    We can proceed exactly as in the proof of \masterthm,
    to obtain a permutation $\vec{c}_3''$ of $\vec{c}_3'$ with
    \begin{equation}
      \A i<j<\len{\vec{c}_3''}. \pt(\vec{c}_3''(i)) \nvo-> \pt(\vec{c}_3''(j))
      \quad\text{and}\quad
      \histOf{r}(\vec{c}_3')
        \hequiv{\AbsState}{\Delta}
      \histOf{r}(\vec{c}_3'').
      \label{prop:pf-master:c3p-equiv-c3pp}
    \end{equation}
    Then, by induction on the length of $\vec{c}_1'$,
    we can insert the items of $\vec{c}_1'$ into
    $ \vec{c}_3'' $, respecting the order they have in~$\vec{c}_\pers$.

    Formally, we prove, for all~$ \len{\vec{c}_3''} \leq i \leq n$,
    that there exists a $\vec{c}$
    which is a scattered subsequence of $ \vec{c}_\pers $
    with $ \len{\vec{c}} = i $
    such that:
    \begin{equation}
      \A c' \in \vec{c} \setminus \dom(\lp),
         c'' \in (\vec{c}_1' \setminus \vec{c}).
        \pt(c') \nvo-> \pt(c'')
      \quad\text{and}\quad
      \histOf{r}(\vec{c} \concat (\vec{c}_1' \setminus \vec{c}))
        \hequiv{\AbsState}{\Delta}
      \histOf{r}(\vec{c}_3).
      \label{goal:pf-master:insert}
    \end{equation}
We proceed by induction on~$i$.
    Intuitively, $\vec{c}$ represents an \nvo-respecting
    subsequence of calls constructed so far and we are inserting in it
    the elements of $ \vec{c}_1' $ one by one.
    Every time we have to swap calls to insert the element,
    the swap will be shown to preserve the equivalence of the histories
    by appealing to \cref{pf-master:vo-nvo-agree-commute}.
    \begin{induction}
      \step[Base case~$i=0$]
        By letting $ \vec{c} = \vec{c}_3'' $ and using
        \eqref{prop:pf-master:c3p-equiv-c3pp}.
      \step[Induction step~$i+1$]
        The induction hypothesis gives us that if $i\leq n$ then
        there exists a $\vec{c}^i$ with $ \len{\vec{c}^i}=i $
        such that
        \begin{gather}
          \A c' \in \vec{c}^i \setminus \dom(\lp),
             c'' \in (\vec{c}_1' \setminus \vec{c}^i).
            \pt(c') \nvo-> \pt(c'')
          \label{prop:pf-master:ind-hyp-nvo}
          \\
          \histOf{r}(\vec{c}^i \concat (\vec{c}_1' \setminus \vec{c}^i))
            \hequiv{\AbsState}{\Delta}
          \histOf{r}(\vec{c}_3).
          \label{prop:pf-master:ind-hyp-equiv}
        \end{gather}
        We now construct a~$\vec{c}$ that proves the statement for~$i+1$.
        Assume $ i+1 \leq n $; then $\len{\vec{c}^i}=i<n$ and therefore
        $ \len{(\vec{c}_3 \setminus \vec{c}^i)} > 0 $.
        Let $ c = (\vec{c}_1' \setminus \vec{c}^i)(0) $.
        First note that, by construction,
        the sequence of persistency points of the calls in $\vec{c}_1'$
        respects $\nvo$,
        and thus:
        \[
          \A \pr{c}\in\vec{c}_1'\setminus(\vec{c}^i \union \set{c}).
          \lp(\pr{c}) \nvo-> \lp(c).
        \]
        and so if we construct~$\vec{c}$ by inserting~$c$ in $\vec{c}^i$,
        the first conjunct of \eqref{goal:pf-master:insert} will hold.

        Our goal is to insert~$c$ in~$\vec{c}^i$ so that
        the resulting sequence is \nvo-respecting and
        keeps the histories equivalent.
        We do this by considering~$ \vec{c}^i \concat c $ and
        swapping~$c$ with its predecessor until we do not violate any
        $\nvo$ edges.
        If~$i=0$ that is trivial as $ \vec{c}^i = \emptyvec $
        and the desired sequence is just~$c$.
        Otherwise,
        let~$c'$ be the last item of $\vec{c}^i$,
        \ie $c' = \vec{c}^i(i-1)$.
        If $ \pt(c') \nvo-> \pt(c) $, $ \vec{c}^i \concat c $ is the desired sequence.
        If not, recall~$c \notin \dom(\lp)$,
        therefore we have two cases, either
        $ c' \notin \dom(\lp) $ or $c'\in\dom(\lp)$.
        In the former case, by induction hypothesis we have
        $\pt(c') \nvo-> \pt(c)$ and we are done.
        In the latter case we have $c'\in\dom(\lp)$ and
        by \ref{pf-master:vo-nvo-agree-commute},
        this gives us two possibilities:
        either $ {\lp(c') \nvo-> \lp(c)} $
        or $
            \tup{\callOf(c), r(c)}
              \comm{\AbsState}{\Delta}
            \tup{\callOf(c'), r(c')}
          $.
        In first case we are again done.
        In the remaining case, we can swap $c$ and $c'$ obtaining
        some sequence~$\upto{\vec{c}^i}{i-2} \concat c c'$
        with
        $
          \histOf{r}(
              \upto{\vec{c}^i}{i-2} \concat c c' \concat
              (\vec{c}_1' \setminus (\vec{c}^i \union \set{c}))
            )
            \hequiv{\AbsState}{\Delta}
          \histOf{r}(\vec{c}^i \concat (\vec{c}_3 \setminus \vec{c}^i)).
        $
        By iterating the argument on~$c$ and $ \upto{\vec{c}^i}{i-2} $
        we can find a suitable position for~$c$ while preserving the history equivalences.
    \end{induction}

  \step[Step~6: persisted memory encodes the final state]
    By \eqref{prop:pf-master:c2-legal},
       \eqref{prop:pf-master:c2-equiv-c3}, and
       \eqref{prop:pf-master:cp-equiv-c3},
    we obtain
    \[
      \LegalFrom{q} \ni
      \histOf{r}(\vec{c}_2)
        \himplies{\AbsState}{\Delta}
      \histOf{r}(\vec{c}_3)
        \hequiv{\AbsState}{\Delta}
      \histOf{r}(\vec{c}_\pers)
        =
      \histOf[\pt]{r}(\enum{G.\restr{\nvo}{G.\Persisted}})
\]
    This implies that there exists a~$q'$ such that
    \begin{align*}
      \histOf{r}(\vec{c}_2)
      \in \LegalPath{q}{q'}
      &&
      \histOf[\pt]{r}(\enum{\restr{G.\nvo}{G.\Persisted}})
      \in \LegalPath{q}{q'}
    \end{align*}
    We can then apply \ref{pf-master:perst-valid} to get that
    $ \tup{\pt, r, G.\restr{\nvo}{G.\Persisted}, \durable} $
    \pre\initOf-validates~$G$.
    By \cref{lm:perst-validates-legal} this implies
    that for some~$q''$:
    \begin{align*}
      \histOf[\pt]{r}(\enum{\restr{G.\nvo}{G.\Persisted}}) \in \LegalPath{q}{q''}
      &&
      \mem{\enum{\restr{G.\nvo}{G.\Persisted}}} \in \durable(q'').
    \end{align*}
    Since $\Delta$ is assumed to be deterministic, we have~$q' = q''$.
    From the equivalences above and \eqref{prop:pf-master:lin-c2}
    we then have:
    \begin{align*}
      \hist(\enum{\lin}) =
      \histOf{r}(\vec{c}_2) \in \LegalPath{q}{q'}
      &&
      \mem{\enum{\restr{G.\nvo}{G.\Persisted}}}
      \in \durable(q').
    \end{align*}
    which completes the proof that
    $\LibImpl[op]$ is \pre\tup{\durable,\recovered}-linearizable.
  \qedhere
  \end{proofsteps}
\end{proof} \section{Verification of Link-Free Set}
\label{sec:verif-linkfree}

In this section we formally prove durable linearizability
of a persistent key-value store implementation called
the \acronym{link-free set}{\LinkFree},
proposed in~\cite{SOFT:oopsla}.
The pseudo-code is reproduced in \cref{fig:lf-set-code}.
The only deviations from the original are the omission of C11-relevant memory order annotations,and the representation of validity as a single bit instead of two.
The latter change is justified by the assumption that the memory manager (which we do not model explicitly) initializes allocated memory with zeroes.

\begin{figure}[p]
  \small
\begin{tabular}{l@{\hspace{3em}}l}
  \begin{sourcecode}[gobble=2,lineskip=-2pt]
  record Node:
    key: $\Nat \dunion \set{+\infty,-\infty}$
    nxt: $\Bool \times \Addr_\bot$
    val: $\Val$
    valid: $\Bool$
    insFl: $\Bool$
    delFl: $\Bool$

  def init():
    t = alloc(Node)
    t.key = $+\infty$; t.valid = 1
    t.nxt = <|0,$\nullptr$|>
    h = alloc(Node)
    h.key = $-\infty$; h.valid = 1
    h.nxt = <|0,t|>
    return h

  def find(h, k):
    p = h
    <|_,c|> = p.nxt
    while(1):
      if c.nxt == <|0,_|>:
        if c.key >= k:
          return <|p,c|>
        p = c
      else:
        trim(p, c)
      <|_,c|> = c.nxt

  def trim(p, c):
    flushDel(c)
    <|_,s|> = c.nxt
    CAS(p.nxt, <|0,c|>, <|0,s|>)

  def makeValid(c):
    if c.valid == 0:
      c.valid = 1
  \end{sourcecode}
  &
  \begin{sourcecode}[gobble=2,firstnumber=last,lineskip=-2pt]
  def insert(h, k, v):
    while(1):
      <|p,c|> = find(h, k)
      if c.key == k:
        makeValid(c) @\label{pt:lfset:insno}@
        flushIns(c)
        return false
      n = alloc(Node)
      n.key = k
      n.val = v
      n.nxt = <|0,c|>
      if CAS(p.nxt, <|0,c|>, <|0,n|>):
        makeValid(n) @\label{pt:lfset:insok}@
        flushIns(n)
        return true

  def delete(h, k):
    while(1):
      <|p,c|> = find(h, k) @\label{pt:lfset:insno}@
      if c.key != k:
        return false
      <|_,n|> = c.nxt
      makeValid(c)
      if CAS(c.nxt, <|0,n|>, <|1,n|>): @\label{pt:lfset:delok}@
        trim(p, c)
        return true


  def flushIns(c):
    if c.insFl == 0:
      flush(c)
      c.insFl = 1

  def flushDel(c):
    if c.delFl == 0:
      flush(c)
      c.delFl = 1
  \end{sourcecode}
\end{tabular}
   \caption{
    The optimized Link-free set implementation.
  }
  \label{fig:lf-set-code}
\end{figure}

\subsection{Key-value Stores}

We write~$\Key$ and~$\Val$ for the set of
possible \emph{keys} and \emph{values} respectively.
$\Key$ is assumed to be totally ordered by $<$.
In the original implementation they are
both represented as~\code{long} values.
The only constraint on the choice of the concrete types
is that the whole node structure must fit in a single cache line.

\begin{definition}[$\KVS$ Specification]
  The abstract states form the set
  $
    \KVS \is \Key \pto \Val.
  $
  The transition relation is defined as:
  \begin{align*}
    \Delta_\KVS(\p{insert}, k, \p{true}) &=
      \set{(\kvs, \kvs \dunion \set{\tup{k,v}}) | k \not\in \dom(\kvs) }
    \\
    \Delta_\KVS(\p{delete}, k, \p{true}) &=
      \set{(\kvs, \kvs\setminus\set{\tup{k,\kvs(k)}}) | k \in \dom(\kvs) }
    \\
    \Delta_\KVS(\p{insert}, k, \p{false}) &=
      \set{(\kvs, \kvs) | k \in \dom(\kvs) }
    \\
    \Delta_\KVS(\p{delete}, k, \p{false}) &=
      \set{(\kvs, \kvs) | k \notin \dom(\kvs)}
\end{align*}
\end{definition}

We omit the \p{contains} operation as its treatment is subsumed by
the treatment of \p{delete}, which also needs hindsight linearization,
see \cref{sec:contains}.

While the above is a perfectly adequate specification,
we do not prove linearizability directly against it,
but to a specification that annotates each key with the address of the node
holding it.
This allows us to strengthen our invariants,
and entails the original goal since linearizability against the more detailed specification implies linearizability against the more abstract one.
More specifically, the enhanced specification allows us to
prove that between two contiguous operations on the same key~$k$,
the same node is continuously representing~$k$ in the set.

\begin{definition}[$\KVS'$ Specification]
  The abstract states form the set
  $
    \KVS' \is \Key \pto (\Addr \times \Val).
  $
  To avoid notational noise we use $\kvs$ to range over $\KVS'$ too.
  The transition relation is defined as:
  \begin{align*}
    \Delta(\p{insert}, \tup{k,v}, \p{true}) &=
      \set{ (\kvs, \kvs\dunion\tup{k,\tup{x,v}}) |
                k \not\in \dom(\kvs), x \text{ fresh in } \kvs }
    \\
    \Delta(\p{insert}, \tup{k,v}, \p{false}) &=
      \set{
        (\kvs, \kvs) | k \in \dom(\kvs)
      }
    \\
    \Delta(\p{delete}, k, \p{true}) &=
      \set{
        (\kvs, \kvs\setminus\tup{k,\tup{x,v}}) | \kvs(k) = \tup{x,v}
      }
    \\
    \Delta(\p{delete}, k, \p{false}) &=
      \set{
        (\kvs, \kvs) | k \notin \dom(\kvs)
      }
\end{align*}
\end{definition}

\begin{remark}
\label{rm:spec-determinism}
The caveat in using $\KVS'$ is that while $\KVS$ is deterministic,
the choice of the new addresses in inserts makes $\KVS'$ non-deterministic.
This is not a fundamental problem:
the non-determinism we introduce is inessential.
Formally,
we can define an equivalence relation on~$\KVS'$,
\[
  (\kvs \sim \kvs') \is
    \A k,v. \kvs(k) = \tup{\wtv, v} \iff \kvs_i(k) = \tup{\wtv, v}
\]
obtaining that the original specification is equivalent
to the quotient of the enhanced one:
$ (\KVS, \Delta_\KVS) = (\KVS', \Delta) / {\sim} $.
The quotient of a transition system is defined as
having equivalence classes as states and the quotient
transition relation relates two classes if any of the elements of the classes are related in the original transition relation.

Our \masterthm\ can be applied to non-deterministic specifications
to prove durable linearizability with respect to a deterministic quotient
of the specifications.
\end{remark}

\begin{lemma}[Voidability in $\KVS'$]
\label{lm:kvs-voidable}
  For the $\KVS'$ specification, the following hold
  \begin{itemize}
\item
    $\tup{\p{insert},k,v,\p{false}}$ is \emph{\pre h-voidable}
      for every~$h$
  \item
    $\tup{\p{delete},k,\p{false}}$ is \emph{\pre h-voidable}
      for every~$h$
  \item
    $\tup{\p{insert},k,v,\p{true}}$ is \emph{\pre h-voidable}
      if and only if
        $h$ contains no calls to operations on the key~$k$,
        or $\tup{\p{insert},k,v,\p{true}} \concat h$ is not legal.
  \item
    $\tup{\p{delete},k,\p{true}}$ is \emph{\pre h-voidable}
      if and only if
        $h$ contains no calls to operations on the key~$k$,
        or $\tup{\p{delete},k,\p{true}} \concat h$ is not legal.
  \end{itemize}
\end{lemma}

\subsection{Global-Happens-Before}

\begin{definition}[Global Happens Before]
We define the following derived orders between events:
\label{def:ghb}
 \begin{align*}
   \rfe & \is \rf \setminus \po &
\ppo & \is
     \restr{\bigl(\po \setminus ((\Writes \union \Flushes) \times \Reads)\bigr)}{
       \Event\setminus \Rets
     }
\\
   \fr & \is \inv{\rf}\seq\mo &
\ghb & \is \tr{(\ppo \union \mo \union \fr \union \rfe)}
 \end{align*}
\end{definition}

The $ \ppo $ relation is the ``preserved program order'',
\ie the program order edges that are obeyed globally.
The $ \rb $ relation is the ``read before'' relation,
also known as the ``from-read'' relation.

\begin{lemma}[\citet{Alglave12}]
\label{lm:ghb-acyclic}
 For every Px86-consistent execution~$G$, $G.\ghb$ is irreflexive.
\end{lemma}

\begin{lemma}
  \label{lm:wrw-ghb}
  $
    \ev{\UWrites} \seq \rf \seq \ev{\Reads} \seq \po \seq \ev{\UWrites}
      \subs \ghb \inters \tso.
  $
\end{lemma}

\begin{lemma}
  \label{lm:rwr-ghb}
  $
    \ev{\Reads} \seq \po \seq \ev{\UWrites} \seq \rf \seq \ev{\Reads}
      \subs \ghb \inters \tso.
  $
\end{lemma}

\begin{lemma}
  \label{lm:wrfl-tso}
  $
    \ev{\UWrites} \seq \rf \seq \ev{\Reads} \seq \po \seq \ev{\Flushes}
      \subs \tso.
  $
\end{lemma}

\lmhbimpliesghb*

\begin{proof}
  Define $\hb(n)$ to be the relation such that
  $(e,e') \in \hb(n)$ if there is a sequence of edges in $ \po \union \rfe $
  leading from~$e$ to~$e'$ that contains exactly~$n$ $\rfe$ edges.
  Clearly, $\hb = \Union_{n\in\Nat} \hb(n)$.
  The statement of the lemma is then equivalent to proving
  $
    \A n. \A e,e'\in\UReads.
      (e,e') \in \hb(n) \implies (e,e') \in \ghb $.
  We prove this by induction on~$n$.
  \begin{induction}
    \step[Base case~$n=0$]
      We have~$\hb(0) = \po$, and
      $ \ev{\UReads}\seq \po \seq \ev{\UReads} \subs \ppo \subs \ghb $.
    \step[Induction step~$n+1$]
      Consider $e,e' \in \UReads$ with $ (e,e') \in \hb(n+1) $.
      There is a path from $e$ to $e'$ containing at least one $\rfe$ edge;
      let us consider the first such edge:
      \[
        e \relto{\maybe{\hb(0)}} w \rfe-> r \relto{\maybe{\hb(n)}} e'.
      \]
      Since $w\in \UWrites$ and $\hb(0)=\po$,
      we can conclude that either $e=w$ or $ (e,w) \in \ppo \subs \ghb $.
      Moreover, since $ r \in \UReads $,
      we have $ (w,r) \in \rfe \subs \ghb $.
      Finally, either $r=e'$, or $ (r,e') \in \hb(n) $, which by induction hypothesis implies $ (r,e') \in \ghb $.
      We obtain
      \[
        e \ghb?-> w \ghb-> r \ghb?-> e'
      \]
      which implies $ (e,e') \in \ghb $.
      \qedhere
  \end{induction}
\end{proof}

\begin{lemma}
\label{lm:uread-ghb-last}
  Let $\vec{e} \in G.\enum{\ghb} $ and
  $ (\vec{e}(j),\vec{e}(i)) \in \rf $ for~$0<j<i<\len{\vec{e}}$.
  Then $\mem{\upto{\vec{e}}{i-1}}(\locOf(\vec{e}(i))) = \rvalOf(\vec{e}(i))$.
\end{lemma}
\begin{proof}
  Assume, towards a contradiction, that~$j$ is not the most recent write to
  $\loc = \locOf(\vec{e}(i))$ in $\upto{\vec{e}}{i-1}$,
  \ie that there is some~$k$ with $j < k < i$ such that
      $\vec{e}(k) \in \UWrites_{\loc}$.
  We have two cases:
  \begin{casesplit}
  \case[$ \vec{e}(j) \mo-> \vec{e}(k) $]
    We have $ \vec{e}(i) \rb-> \vec{e}(k) $
    which implies $ \vec{e}(i) \ghb-> \vec{e}(k) $
    which entails~$i<k$ yielding a contradiction.
  \case[$ \vec{e}(k) \mo-> \vec{e}(j) $]
    We have $ \vec{e}(k) \ghb-> \vec{e}(j) $
    which entails~$k<j$ yielding a contradiction.
  \qedhere
  \end{casesplit}
\end{proof}

\begin{lemma}[Updates read from last write]
\label{lm:upd-read-last}
  Let $\vec{e} \in G.\enum{\ghb} $ and
  $ \vec{e}(i) \in \Updates_{\loc} $ for~$0<i<\len{\vec{e}}$.
  Then $\mem{\upto{\vec{e}}{i}}(\loc) = \wvalOf(\vec{e}(i))$
  and  $\mem{\upto{\vec{e}}{i-1}}(\loc) = \rvalOf(\vec{e}(i))$.
\end{lemma}
\begin{proof}
  That $\mem{\upto{\vec{e}}{i}}(\loc) = \wvalOf(\vec{e}(i))$
  follows directly from \cref{def:mem}.
  $\mem{\upto{\vec{e}}{i-1}}(\loc) = \rvalOf(\vec{e}(i))$
  follows from \cref{lm:uread-ghb-last}
  noting that $ \rf\seq\ev{\Updates} \subs \ghb $
  and therefore if $(\vec{e}(j), \vec{e}(i)) \in \rf$ then~$j<i$.
\end{proof}

\subsection{Node structure}

The heap is assumed to be a uniform collection of nodes of type \p{Node}:

\begin{center}
\begin{tabular}{c@{\hspace{5em}}c}
{\begin{sourcecode}
struct Node {
  long key;
  byte valid;
  bool insFl;
  bool delFl;
  struct Node* nxt;
}
\end{sourcecode}}
&
{\begin{sourcecode*}
record Node {
  key   : $\Key \dunion \set{+\infty, -\infty}$
  nxt   : $ \Bool \times \Addr_\nullptr $
  valid : $\Bool$
  insFl : $\Bool$
  delFl : $\Bool$
}
\end{sourcecode*}}
\end{tabular}
\end{center}

On the left is the actual C structure.
On the right is the record we are going to use to represent each node.

The algorithm uses the maximum and minimum possible \code{long} values
as sentinel values in the \p{key} field.
We therefore represent \p{key} as a natural
plus the sentinel values $\pm\infty$;
the set data structure will only be able to store elements with keys in $\Nat$.
The \p{nxt} field stores a pointer to the next node
and a deletion mark as its least significant bit.
The record therefore represents \p{nxt} as a pair $ \tup{b, p} $
where~$b\in\Bool$ is the mark bit
and~$p\in\Addr_\nullptr \is \Addr \dunion \set{\nullptr}$ is the actual address of the next node, or null~$(\nullptr)$.
We say a node at $n$ is \emph{marked (for deletion)} if
$ \loc[n.nxt] = \tup{1, \wtv} $.
The \p{valid} field in the actual implementation is used to store two bits,
and consider the node as valid if
the two bits coincide and $\p{nxt} \ne \nullptr$;
we simplify the scheme by having a single bit
and let a node be valid if its \p{valid} field is 1;
since we assume allocated memory is zeroed,
any newly minted node will be automatically initially invalid.

The initial (empty) state consists of two nodes, the head~$\var{head}$ and tail~$t$ with
$\loc[\var{head}.nxt] = t$, $\loc[\var{head}.key] = -\infty$, and $\loc[t.key] = +\infty$.
Whether they are persisted or not is irrelevant, they are simply ignored and freshly re-allocated by the recovery.
The address $\var{head}$ of the head node acts as the entry address for the entire set data structure.

\subsection{Events generated by the implementation}

We characterise the possible sets of events produced by the library operations
in \cref{fig:traces-find,fig:traces-ins,fig:traces-del}.
We write them in pseudocode where vertical sequences of events signify
\po!-ordering, and any free variable is to be considered existentially quantified.
Implicitly, the traces are also closed by truncating them with an error event
in case of null dereferences.
That is, each line that accesses some address~$x$ is to be understood as a shortcut for a check for~$x$ being null followed by the actual access,~e.g.:
\[
  \begin{traceblock*}
  \IF x = \nullptr\\
  \begin{traceblock}
    \Err
  \end{traceblock}\\
  \ELSE\\
  \begin{traceblock}
    \R{x}{f}{\wtv} \WHERE \text{or write, or update}\\
    \COMMENT \text{rest of code}
  \end{traceblock}
  \end{traceblock*}
\]
Part of the correctness proof will show that executions leading to $\Err$ are impossible.

Let~$\pref(X)$ be the set of prefixes of the sequences in~$X$.
We then let
$ \LinkFreeImpl[op] $
be such that
\begin{align*}
  \LinkFreeImpl[op](\p{insert}, k, v) &\is
    \pref(\ref{op:insert-ok}(k,v) \union \ref{op:insert-no}(k,v))
  \\
  \LinkFreeImpl[op](\p{delete}, k) &\is
    \pref(\ref{op:delete-ok}(k) \union \ref{op:delete-no}(k))
\end{align*}

\begin{figure*}[p]
\begin{align*}
&
\begin{trace}{FIND}{k, p, c}
\label{op:find}
  \COMMENT n_0 = \var{head}\\
  \FOR i = {0,\dots, m{-}1}\\
  \begin{traceblock}
    \R{n_i}{nxt}{\tup{0,\wtv}} \\
    \R{n_i}{key}{k_i} \WHERE k_i < k \\
    \ref{op:garbage}(n_i,n_{i+1})
  \end{traceblock}\\
  \COMMENT p = n_m \\
  \LINE{r_{\p{p}}}\label{ev:find:rp}
  \R{p}{nxt}{\tup{0,\wtv}}\\
  \R{p}{key}{k_m} \WHERE k_m < k \\
  \ref{op:garbage}(p,c)\\
  \LP{\p{\tiny FIND}}\label{lp:find}
  \R{c}{nxt}{\tup{0,\wtv}}\\
  \R{c}{key}{k'} \WHERE k \leq k' \\
\end{trace} \qquad&
&
\begin{trace}{GARBAGE}{n, n'}
\label{op:garbage}
  \COMMENT n_0 = n, d \geq 1\\
  \R{n_0}{nxt}{\tup{\wtv,n_1}}\\
  \FOR i = {1,\dots, d{-}1}\\
  \begin{traceblock}
    \R{n_i}{nxt}{\tup{1,\wtv}}\\
    \ref{op:trim}(n_0, n_i)\\
    \R{n_i}{nxt}{\tup{\wtv, n_{i+1}}}
  \end{traceblock}\\
  \COMMENT n_d = n'
\end{trace} \\[.5\baselineskip]
&
\begin{trace}{TRIM}{p, c}
\label{op:trim}
  \ref{op:flush-del}(c)\\
  \R{c}{nxt}{\tup{\wtv,s}}\\
\begin{either}
  \LINE{w_{\tiny\p{TRIM}}}\label{ev:trim-w-nxt}
      \U{p}{nxt}{\tup{0,c}}{\tup{0,s}}\\
      \R{p}{nxt}{\tup{b,c'}}\WHERE b=1 \lor c'\ne c
    \end{either}
\end{trace} &&
\begin{trace}{FLUSH\_DEL}{n}
\label{op:flush-del}
  \R{n}{delFl}{b}\\
  \IF b=0\\
  \begin{traceblock}
    \FL{n}\\
    \W{n}{delFl}{1}\\
  \end{traceblock}
\end{trace} \end{align*}
\caption{Auxiliary procedures}
\label{fig:traces-find}
\end{figure*}
\begin{figure*}[p]
\begin{align*}
&
\begin{trace}{INSERT\_OK}{k, v}
\label{op:insert-ok}
\FOR i = {1,\dots,m} \\
  \begin{traceblock}
    \ref{op:find}(k, p_i, c_i)\\
    \R{c_i}{key}{k_i}
      \WHERE k_i \ne k\\
    \ref{op:newnode}(\wtv, k, v, c_i)\\
    \R{p_i}{nxt}{\tup{b_i, n_i}}
      \WHERE b_i=1 \lor n_i\ne c_i\\
  \end{traceblock}\\
\ref{op:find}(k, p, c)\\
  \R{c}{key}{k''} \WHERE k\ne k''\\
  \ref{op:newnode}(n, k, v, c)\\
  \LP{1}\label{lp:insert-ok}
  \U{p}{nxt}{\tup{0, c}}{\tup{0, n}}\\
  \ref{op:makevalid}(n)\\
  \ref{op:flush-ins}(n)\\
  \Ret{\p{true}}
\end{trace}
 \qquad&
&
\begin{trace}{INSERT\_NO}{k, v}
\label{op:insert-no}
\FOR i = {1,\dots,m} \\
    \begin{traceblock}
      \ref{op:find}(k, p_i, c_i)\\
      \R{c_i}{key}{k_i}
        \WHERE k_i \ne k\\
      \ref{op:newnode}(\wtv, k, v, c_i)\\
      \R{p_i}{nxt}{\tup{b_i, n_i}}
        \WHERE b_i=1 \lor n_i\ne c_i\\
    \end{traceblock}
  \\
\LP{2}\label{lp:insert-no}
  \ref{op:find}(k, p, c)\\
  \R{c}{key}{k}\\
  \ref{op:makevalid}(c)\\
  \ref{op:flush-ins}(c)\\
  \Ret{\p{false}}
\end{trace}
 \end{align*}
\begin{align*}
&
\begin{trace}{NEWNODE}{n, k, v, c}
\label{op:newnode}
  \Alloc{n}\\
\W{n}{key}{k}\\
  \W{n}{val}{v}\\
  \LINE{w_{\p{nxt}}}\label{ev:newnode-w-nxt}
  \W{n}{nxt}{\tup{0,c}}
\end{trace}
 &&
\begin{trace}{MAKEVALID}{n}
\label{op:makevalid}
  \R{n}{valid}{b}\\\IF b = 0\\
   \PT{1}\label{pt:makevalid}
   \begin{traceblock}
     \W{n}{valid}{1}\\\end{traceblock}
\end{trace} &&
\begin{trace}{FLUSH\_INS}{n}
\label{op:flush-ins}
  \R{n}{insFl}{b}\\
  \IF b=0\\
  \begin{traceblock}
    \FL{n}\\
    \W{n}{insFl}{1}\\
  \end{traceblock}
\end{trace}
 \end{align*}
\caption{Events generated by insertions}
\label{fig:traces-ins}
\end{figure*}

\begin{figure*}[tp]
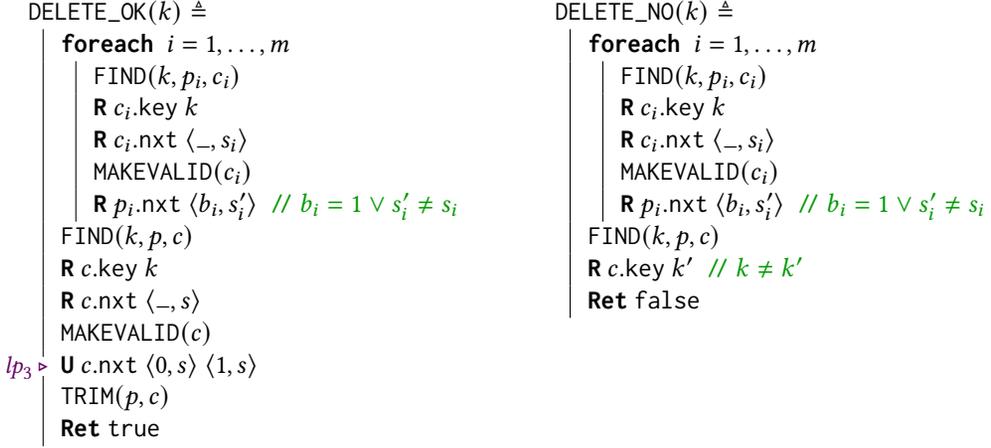

\begin{align*}
&
\begin{trace}{DELETE\_OK}{k}
\label{op:delete-ok}
  \FOR i = {1,\dots,m} \\
  \begin{traceblock}
    \ref{op:find}(k, p_i, c_i)\\
    \R{c_i}{key}{k}\\
    \R{c_i}{nxt}{\tup{\wtv,s_i}}\\
    \ref{op:makevalid}(c_i)\\
    \R{p_i}{nxt}{\tup{b_i, s'_i}}
      \WHERE b_i=1 \lor s'_i\ne s_i\\
  \end{traceblock}\\
  \ref{op:find}(k, p, c)\\
  \R{c}{key}{k}\\
  \R{c}{nxt}{\tup{\wtv,s}}\\
  \ref{op:makevalid}(c)\\
  \LP{3}\label{lp:delete-ok}
  \U{c}{nxt}{\tup{0, s}}{\tup{1, s}}\\
  \ref{op:trim}(p,c)\\
  \Ret{\p{true}}
\end{trace} \qquad&
&
\begin{trace}{DELETE\_NO}{k}
\label{op:delete-no}
  \FOR i = {1,\dots,m} \\
  \begin{traceblock}
    \ref{op:find}(k, p_i, c_i)\\
    \R{c_i}{key}{k}\\
    \R{c_i}{nxt}{\tup{\wtv,s_i}}\\
    \ref{op:makevalid}(c_i)\\
    \R{p_i}{nxt}{\tup{b_i, s'_i}}
      \WHERE b_i=1 \lor s'_i\ne s_i\\
  \end{traceblock}\\
  \ref{op:find}(k, p, c)\\
  \R{c}{key}{k'} \WHERE k \ne k' \\
  \Ret{\p{false}}
\end{trace}
 \end{align*}
\caption{Events generated by deletions}
\label{fig:traces-del}
\end{figure*}

\subsection{State representation}

We now define the durable, volatile, and recovered
representation functions.

\begin{definition}[Durable state representation for \LinkFree]
\label{def:linkfree:recoverable}
  We define $ \durable \from \KVS' \to \powerset(\Mem) $
  as the function such that
  $ M \in \durable(\kvs) $
  if and only if
  $
    M =
      \durM_{\lbl{s}} \dunion
      \durM_{\lbl{d}} \dunion
      \durM_{\lbl{g}}
  $
  where
  $
    \kvs(k) = \tup{x_k, v_k}
  $ and
  \begin{align*}
    \durM_{\lbl{s}} &=
      \Dunion_{k\in\dom(\kvs)} \left[
      \begin{matrix*}[l]
        \loc[x_k.key] \mapsto k,
        &
        \loc[x_k.val] \mapsto v_k,
        &
        \loc[x_k.nxt] \mapsto \tup{0,\wtv},
        \\
        \loc[x_k.valid] \mapsto 1,
        &
        \loc[x_k.insFl] \mapsto \wtv,
        &
        \loc[x_k.delFl] \mapsto 0
      \end{matrix*}
      \right]
    \\
    \durM_{\lbl{d}} &=
      \Dunion_{y\in X_{\lbl{d}}} [
        \loc[y.nxt] \mapsto \tup{1,\wtv},
        \loc[y.valid] \mapsto 1,
        \loc[y.$\wtv$] \mapsto \wtv
      ]
    \\
    \durM_{\lbl{g}} &=
      \Dunion_{y\in X_{\lbl{g}}} [
        \loc[y.nxt] \mapsto \tup{0,\wtv},
        \loc[y.valid] \mapsto 0,
        \loc[y.delFl] \mapsto 0,
        \loc[y.$\wtv$] \mapsto \wtv
      ]
  \end{align*}
  for some sets of addresses
$ X_{\lbl{d}}$, and
  $X_{\lbl{g}} $.
  Intuitively,
  $ \durM_{\lbl{s}} $ collects the nodes representing members of the store,
  $ \durM_{\lbl{d}} $ collects deleted nodes,
  $ \durM_{\lbl{g}} $ collects garbage nodes.
\end{definition}

We define the recovered states as the ones where the
volatile and durable representations agree on the encoded set.
Formally we define some volatile state representation~$\volatile$
and define~$\recovered(q) \is \durable(q) \inters \volatile(q)$.

\begin{definition}[Volatile and recovered state representation for \LinkFree]
\label{def:linkfree:volatile}
\label{def:linkfree:recovered}
  We define $ \volatile \from \KVS' \to \powerset(\Mem) $
  as the function shuch that
  $ M \in \volatile(\kvs) $
  if and only if
  $
    M =
      \volM_{\lbl{s}} \dunion
      \volM_{\lbl{d}} \dunion
      \volM_{\lbl{u}}
  $
  where
  $
    \kvs(k) = \tup{x_k, v_k}
  $,
  $ x_{-\infty} = \var{head} $,
  $ x_{+\infty} = \var{tail} $,
  $
    \dom(\kvs) \union \set{+\infty,-\infty} = \set{k_1,\dots,k_n}
  $, and
  \begin{align*}
    \volM_{\lbl{s}} &=
      \Dunion_{1 \leq i \leq n} \left[
      \begin{matrix*}[l]
        \loc[x_{k_i}.key] \mapsto {k_i},
        &
        \loc[x_{k_i}.val] \mapsto v_{k_i},
        &
        \loc[x_{k_i}.nxt] \mapsto \tup{0,\wtv},
        \\
        \loc[x_{k_i}.valid] \mapsto \wtv,
        &
        \loc[x_{k_i}.insFl] \mapsto \wtv,
        &
        \loc[x_{k_i}.delFl] \mapsto 0
      \end{matrix*}
      \right]
    \\
    \volM_{\lbl{d}} &=
      \Dunion_{y\in X_{\lbl{d}}} [
        \loc[y.nxt] \mapsto \tup{1,\wtv},
        \loc[y.valid] \mapsto 1,
        \loc[y.$\wtv$] \mapsto \wtv
      ]
    \\
    \volM_{\lbl{u}} &=
      \Dunion_{y\in X_{\lbl{u}}} [
        \loc[y.nxt] \mapsto \tup{0,\wtv},
        \loc[y.valid] \mapsto 0,
        \loc[y.delFl] \mapsto 0,
        \loc[y.$\wtv$] \mapsto \wtv
      ]
  \end{align*}
  for some sets of addresses
  $X_{\lbl{d}}$,
  $X_{\lbl{u}} $, and
  $ X_{\lbl{s}} = \set{x_{k_1},\dots,x_{k_n}} $
  and such that
  \begin{gather}
    \A x,y \in \Addr.
      \bigl(
        \volM(\loc[x.key]) < +\infty \land
        \volM(\loc[x.nxt]) = \tup{\wtv, y}
      \bigr)
        \implies
          \volM(\loc[x.key]) < \volM(\loc[y.key])
    \label{cond:linkfree:sorted-links}
    \\
    \begin{multlined}
    \A i<n.\E m\geq 0.
    \E y_1,\dots,y_m\in X_{\lbl{d}}.
    \E y_{m+1} = x_{k_{i+1}}.
    \\
    \textstyle
      \volM(\loc[x_{k_i}.nxt]) = \tup{0,y_1}
      \land
      \bigl(
      \LAnd_{1\leq j \leq m}
        \volM(\loc[y_j.nxt]) = \tup{1,y_{j+1}}
      \bigr)
\end{multlined}
    \label{cond:linkfree:members-reachable}
    \\
    \A y \in X_{\lbl{d}}.
    \E y' \in X_{\lbl{s}} \union X_{\lbl{d}}.
      \volM(\loc[y.nxt]) = \tup{\wtv,y'}
    \label{cond:linkfree:marked-nxt-init}
  \end{gather}
  That is,
  $ \volM_{\lbl{s}} $ is a sorted linked list of
  valid unmarked nodes representing members of the store,
  possibly interleaved with deleted notes;
  $ \volM_{\lbl{d}} $ represents deleted nodes, and
  $ \volM_{\lbl{u}} $ represents uninitialised nodes.
  Note that even links in
  $ \volM_{\lbl{d}} $ and
  $ \volM_{\lbl{u}} $ are sorted,
  although they are not required to form a list.
  Moreover, the sortedness constraint~\eqref{cond:linkfree:sorted-links}
  implies~$k_i < k_{i+1}$ for all~$i < n$.

  The recovered state representation for \LinkFree{} is
  $ \recovered(\kvs) \is \durable(\kvs) \inters \volatile(\kvs) $.
\end{definition}

\begin{lemma}
\label{lm:linkfree:volatile-members-reach}
  For all~$M \in \volatile(\kvs)$,
  if $ \kvs(k) = \tup{x,v} $ then
  $ \reach(M,x) $ and $ M(\loc[x.nxt]) = \tup{0,\wtv} $
\end{lemma}
\begin{proof}
  By \eqref{cond:linkfree:sorted-links}
  and \eqref{cond:linkfree:members-reachable}.
\end{proof}

\begin{lemma}
\label{lm:linkfree:volatile-reach-sd}
  For all~$M \in \volatile(\kvs)$,
  if~$ \reach(M,x) $ then
  $ x \in X_{\lbl{s}} \dunion X_{\lbl{d}} $
  (for $X_{\lbl{s}},X_{\lbl{d}}$ as in \cref{def:linkfree:volatile}).
\end{lemma}
\begin{proof}
  Trivial by condition~\eqref{cond:linkfree:members-reachable} and
  the fact that~$\var{head} \in X_{\lbl{s}}$
  from \cref{def:linkfree:volatile}.
\end{proof}

From here on, $G$ will refer to some arbitrary execution
of $\tup{\LinkFreeImpl[op],\LibImpl[rec]}$
for some~\pre\tup{\durable,\recovered}-sound recovery $\LibImpl[rec]$,
with $G.\Init \in \durable(\kvs)$ for some $\kvs\in\KVS'$.
In particular, for executions where the
recovery has not returned we do not need to prove anything
(soundness of recovery takes care of those).
So, for each such~$G$ we can assume, for some $\kvs\in\KVS'$, that
$ {\mem{\restr{G.\tso}{G.\Init\union G.\EvOfCid{\recoveryId}}} \in \recovered(\kvs)} $
which means that:
\begin{equation}
  \A e \in \enum[G.E]{\ghb}.
    \mem{\upto{\vec{e}}{\initOf(\vec{e})}} \in \durable(\kvs)
    \land
    \mem{\upto{\vec{e}}{\recEndOf(\vec{e})}} \in \recovered(\kvs)
  \label{prop:linkfree:init-and-rec}
\end{equation}

Since we are treating recovery opaquely,
we will treat the events in $G.\Init\union G.\EvOfCid{\recoveryId}$ uniformly,
and thus we will use the shortcut:
\[
  G.E_0 \is G.\Init\union G.\EvOfCid{\recoveryId}
\]

Furthermore, all the considered addresses will be library addresses,
unless otherwise specified.

\subsection{Invariants}
\label{sec:invariants}

We begin our formal analysis of \LinkFree\ by proving some baseline invariants,
which we will use in proving linearizability.
Informally, we want to prove the following invariants on \ghb-induced memory:
\begin{enumerate}[label=(I\arabic*),leftmargin=2.5\parindent]
\item For each pair of nodes $n_1,n_2$ such that $n_1\p{.nxt} =\tup{\wtv, n_2}$,
        we have $ n_1\p{.key} < n_2\p{.key} $.
        \label{inv:sorted}
  \item Every node that is unmarked and pointed at, is reachable from the head.
        \label{inv:reach}
  \item If $ \loc[n.insFl] = 1 $ then $ n $ is valid
  \item If $ \loc[n.delFl] = 1 $ then $ n $ is valid and marked
\item The fields $ \p{key}, \p{val} $ are assigned to exactly once.
    \label{inv:key-val-once}
  \item Once marked a node stays marked and the next pointer of a marked node is never updated.
    \label{inv:marking-irrev}
  \item Once valid a node stays valid.
    \label{inv:valid-irrev}
\end{enumerate}

\begin{lemma}[Single assignments~\ref{inv:key-val-once}]
\label{lm:unique-writes}
  In $G.E \setminus G.E_0$,
  there is at most one event $ (e_1\of\W{x}{key}{\wtv}) $,
  at most one event $ (e_2\of\W{x}{val}{\wtv}) $, and
  at most one event $ (e_3\of\W{x}{nxt}{\wtv}) $,
  for each address~$x$.
  However, these events can only exist if
  $ G.E_0 \inters \UWrites_x = \emptyset $.
  Moreover, if these events exist, then
    $ (e_1,e_2) \in \po! $ and
    $ (e_2,e_3) \in \po! $.
\end{lemma}
\begin{proof}
  All such write events are generated by~\ref{op:newnode}
  and apply to a freshly allocated address.
  Note that all other events affecting \p{nxt} are updates.
\end{proof}

\begin{definition}[Initialised address]
\label{def:init-addr}
  We say an address~$x$ is \emph{initialised at~$e \in G.E$} if
  there exists~$e_0 \in G.E$ such that
  $ e_0 \in G.E_0 \inters \UWrites_{\loc[x.nxt]} $ or
  \[
    \underbrace{
    \Alloc{x} \po-> \W{x}{nxt}{\wtv}
    }_{\ref{op:newnode}(x,\wtv[3])}
    \po!->
    (e_0 \of \U{y}{nxt}{\wtv}{\tup{0,x}})
    \ghb?-> e
  \]
\end{definition}

\begin{lemma}
\label{lm:uread-nxt-ghb-last}
  Let $\vec{e} \in G.\enum{\ghb} $ and
  $ \vec{e}(i) \in \Reads_{\loc[x.nxt]} $
  for some~$\recEndOf(\vec{e})<i<\len{\vec{e}}$.
  Then $\mem{\upto{\vec{e}}{i-1}}(\loc[x.nxt]) = \rvalOf(\vec{e}(i))$.
\end{lemma}
\begin{proof}
  Let $ (\vec{e}(j),\vec{e}(i)) \in \rf $.
  If $ \vec{e}(j) \in G.E_0 \union \Updates $ or $ (\vec{e}(j),\vec{e}(i)) \in \rfe $ then $(\vec{e}(j),\vec{e}(i)) \in \ghb$ so $j<i$
  and \cref{lm:uread-ghb-last} proves the claim.
  Otherwise, $(\vec{e}(j),\vec{e}(i)) \in \rfi \subs \po$ and $\vec{e}(j) \in \Writes$. The only writes to \p{nxt} are generated by \ref{op:newnode},
  which implies~$x$ is initialized at~$ \vec{e}(i) $
  and we have
  $ \vec{e}(j) \po-> (e_0 \of \U{y}{nxt}{\wtv}{\wtv}) \po-> \vec{e}(i) $,
  from which $(\vec{e}(j),\vec{e}(i)) \in \ghb$ so $j<i$
  and \cref{lm:uread-ghb-last} again proves the claim.
\end{proof}

\begin{lemma}
\label{lm:uread-key-ghb-last}
  Let $\vec{e} \in G.\enum{\ghb} $ and
  $ \vec{e}(i) \in \Reads_{\loc[x.key]} $
  for some~$\recEndOf(\vec{e})<i<\len{\vec{e}}$.
  Then $\mem{\upto{\vec{e}}{i-1}}(\loc[x.key]) = \rvalOf(\vec{e}(i))$.
\end{lemma}

\begin{proof}
  Analogous to the proof of \cref{lm:uread-nxt-ghb-last}.
\end{proof}

\begin{lemma}[Pointed-at nodes are initialised]
\label{lm:pointed-at-init}
  Let~$\vec{e} \in G.\enum{\ghb}\!$,
  and~$\recEndOf(\vec{e})<i<\len{\vec{e}}$.
  If\/ $ (\vec{e}(i) \of \R{x}{nxt}{\tup{\wtv,y}}) $
  then~$ y $ is initialised at~$\vec{e}(i)$.
\end{lemma}
\begin{proof}
  By induction on~$i$.
Assume the statement holds for every~$i'$
      with $\recEndOf(\vec{e}) < i' < i$.
      We do a case analysis on the label
      of the possible source $e$ of the \rf-edge
      to~$(\vec{e}(i)  \of \R{x}{nxt}{\tup{\wtv,y}})$.
      \begin{casesplit}
        \case[$ e \in G.E_0$]
          By assumption~\eqref{prop:linkfree:init-and-rec},
          $\mem{\upto{\vec{e}}{\recEndOf(\vec{e})}} = \wvalOf(e)$;
          From $\mem{\upto{\vec{e}}{\recEndOf(\vec{e})}} \in \recovered(\kvs)$
          we have another event
            $ e' \in G.E_0 \inters \Writes_{\loc[y.nxt]} $,
          which proves~$y$ is initialised at~$\vec{e}(i)$.
\case[$ (e \of \Alloc{x}) $]
          We prove this is impossible, by considering two subcases:
          \begin{casesplit}
            \case*[$x=\var{head}$]
              Trivial as, according to $\recovered$,
              $\var{head}$ is allocated in $G.E_0$.
            \case*[$x\ne \var{head}$]
              Every read to $x.\p{nxt}$ (except the ones on~$x=\var{head}$)
              is \po-preceded by a read
              $ r \of \R{z}{nxt}{\tup{\wtv,x}} $ for some~$z$.
              Since $(r,\vec{e}(i)) \in \po \subs \ghb$,
              $ r = \vec{e}(j) $ for some~$j < i$.
              We can then apply the induction hypothesis on~$r$
              and get that~$x$ is initialised at~$r$.
              $ x $ cannot be part of the nodes created at the initial events,
              since those do not have allocation events, and allocation events
              cannot reuse addresses from the initialisation.
              This gives us
              \[
                (e \of \Alloc{x})
                \po->
                (e_1 \of \W{x}{nxt}{\wtv})
                \po!->
                (e_0 \of \U{y}{nxt}{\wtv}{\tup{0,x}})
                \ghb->
                r
                \ghb->
                \vec{e}(i)
              \]
              Since
              $ (e,\vec{e}(i)) \in \rf $
              and $ (e,e_1) \in \mo $,
              we have $ (\vec{e}(i),e_1) \in \fr \subs \ghb $,
              which is a contradiction as that would cause
              $ \ghb $ to have a cycle.
          \end{casesplit}
\case[$ (e \of \W{x}{nxt}{\tup{0,y}}) $
              from \ref{op:newnode}]
          By construction~$e$ is \po-preceded (in \ref{op:find})
          by an event $ (r \of \R{z}{nxt}{\tup{\wtv,y}}) $.
          We have $ r \po-> e \rf-> \vec{e}(i) $ and so
          $ r \ghb-> \vec{e}(i) $.
          This means we can apply the induction
          hypothesis on~$r$ obtaining that~$y$ is initialised at~$r$,
          and hence at~$\vec{e}(i)$.
\case[$ (e \of \U{x}{nxt}{\tup{\wtv,\wtv}}{\tup{\wtv,y}}) $
              from \ref{op:insert-ok}]
          We have~$(e, \vec{e}(i)) \in \ghb$ by~\cref{lm:hb-implies-ghb}.
          The update is \po-preceded by \ref{op:newnode}$(y,\wtv[3])$
          which gives us the desired events to show~$y$ is initialised
          at~$\vec{e}(i)$.
\case[$ (e \of \U{x}{nxt}{\tup{\wtv,z}}{\tup{\wtv,y}}) $
              from \ref{op:trim}]
          The update is \po-preceded by an event
          $ (r \of \R{x}{nxt}{\tup{\wtv,y}}) $.
          Since~$ r \ghb-> e \ghb-> \vec{e}(i) $,
          we can apply the induction hypothesis on~$r$
          and obtain that~$y$ is initialised at~$e$ and hence at~$\vec{e}(i)$.
\case[$ (e \of \U{x}{nxt}{\tup{0,y}}{\tup{1,y}}) $
              from \ref{op:delete-ok}]
          Similar to the previous case as~$e$ is also~\po-preceded by
          an event $ (r \of \R{x}{nxt}{\tup{\wtv,y}}) $.
\qedhere
      \end{casesplit}
\end{proof}

\begin{lemma}[Reached nodes are initialised]
\label{lm:reach-init}
  For all~$\vec{e} \in G.\enum{\ghb}$,
  for all~$\recEndOf(\vec{e})<i<\len{\vec{e}}$,
  if $ \vec{e}(i) \in \Reads_{\loc[x.nxt]} \union \Reads_{\loc[x.key]} $,
  then~$ x $ is initialised at~$\vec{e}(i)$.
\end{lemma}
\begin{proof}
  We consider two cases:
  \begin{casesplit}
    \case[$x=\var{head}$]
      By construction $\var{head}$ is initialised
      in~$G.E_0$.
    \case[$x\ne\var{head}$]
      Every such read is \po-preceded by a read
      $ (r \of \R{z}{nxt}{\tup{\wtv,x}}) $.
      Note that $(r,\vec{e}(i)) \in \ghb$.
      By~\cref{lm:pointed-at-init} we have that~$x$
      is initialised at~$r$ and hence at~$\vec{e}(i)$ too.
      \qedhere
  \end{casesplit}
\end{proof}

\begin{lemma}[Next is only set to initialised nodes]
\label{lm:w-nxt-init}
  For all~$\vec{e} \in G.\enum{\ghb}$,
  and all~$\recEndOf(\vec{e})<i<\len{\vec{e}}$,
  and all~$x,y \in \Addr$,
  if $ \vec{e}(i) \in \UWrites_{\loc[x.nxt]}$
  and $ \wvalOf(\vec{e}(i)) = \tup{\wtv,y} $,
  then~$ y $ is initialised at~$\vec{e}(i)$.
\end{lemma}
\begin{proof}
  We consider all possible writes to~$\loc[x.nxt]$.
  \begin{casesplit}
    \case[$ (\vec{e}(i) \of \Alloc{x}) $]
      Trivial as $ \wvalOf(\vec{e}(i)) = \tup{\wtv,\nullptr} $.
    \case[$ (\vec{e}(i) \of \W{x}{nxt}{\tup{0,y}}) $
          from \ref{op:newnode}]
      The event is \po-preceded by a read
      $ r \of \R{z}{nxt}{\tup{\wtv,y}} $.
      Since $ (r,\vec{e}(i)) \in \ghb $,
      by \cref{lm:pointed-at-init}~$y$ is initialised at~$r$
      and thus at $\vec{e}(i)$.
    \case[$ (\vec{e}(i) \of \U{x}{nxt}{\wtv}{\tup{\wtv,y}}) $
          from \ref{op:insert-ok}]
      The events generated by the \ref{op:newnode}$(y,\wtv[3])$
      \po-preceding $\vec{e}(i)$ are the witnesses
      to show~$y$ is initialised.
    \case[$ (\vec{e}(i) \of \U{x}{nxt}{\wtv}{\tup{\wtv,y}}) $
          from \ref{op:trim}]
      Similar to the \ref{op:newnode} case.
    \case[$ (\vec{e}(i) \of \U{x}{nxt}{\tup{0,y}}{\tup{1,y}}) $
          from \ref{op:delete-ok}]
      Similar to the \ref{op:newnode} case.
    \qedhere
  \end{casesplit}
\end{proof}

\begin{lemma}[Next is only updated on initialised nodes]
\label{lm:u-init}
  For all~$\vec{e} \in G.\enum{\ghb}$,
  and all~$\recEndOf(\vec{e})<i<\len{\vec{e}}$,
  and all~$x \in \Addr$,
  if\/ $ \vec{e}(i) \in \Updates_{\loc[x.nxt]}$
  then~$ x $ is initialised at~$\vec{e}(j)$ for some~$j$
  with $\recEndOf(\vec{e})<j<i$,
  and $ \vec{e}(j) \ghb-> \vec{e}(i) $.
  Note that his implies~$x$ is initialised at $\vec{e}(i)$ too.
\end{lemma}
\begin{proof}
  By case analysis on the update events.
  \begin{casesplit}
    \case[$ (\vec{e}(i) \of \U{x}{nxt}{\wtv}{\wtv}) $
          from \ref{op:insert-ok}]
      The event is \po-preceded by \ref{op:find}$(\wtv,x,\wtv)$.
      If $x=\var{head}$ the statement is proven as~$\var{head}$
      is always initialised.
      Otherwise, there is an event~$ (r\of \R{z}{nxt}{\tup{\wtv,x}}) $
      \po-preceding~$\vec{e}(i)$.
      We thus have~$ r \ghb-> \vec{e}(i) $ and $ r = \vec{e}(j) $
      for some $\recEndOf(\vec{e})<j<i$.
      By~\cref{lm:pointed-at-init} we have~$x$ is initialised at~$\vec{e}(j)$.
    \case[$ (\vec{e}(i) \of \U{x}{nxt}{\wtv}{\tup{\wtv,y}}) $
          from \ref{op:trim}]
      The event is \po-preceded in all instances of \ref{op:trim},
      by a read event $ (r \of \R{x}{nxt}{\wtv}) $.
      So, similarly to the previous case,
      we can apply~\cref{lm:reach-init} and conclude~$x$ is initialised at~$r$.
    \case[$ (\vec{e}(i) \of \U{x}{nxt}{\tup{0,y}}{\tup{1,y}}) $
          from \ref{op:delete-ok}]
      The event is \po-preceded by a read event $ (r \of \R{x}{nxt}{\wtv}) $.
      So, similarly to the previous case,
      we can apply~\cref{lm:reach-init} and conclude~$x$ is initialised at~$r$.
    \qedhere
  \end{casesplit}
\end{proof}

\begin{lemma}
\label{lm:newnode-init}
  For all~$\vec{e} \in G.\enum{\ghb}$,
  and all~$\recEndOf(\vec{e})<i<\len{\vec{e}}$,
  if $ \addrOf(\vec{e}(i)) = n \in \Addr  $ then
  either $ \vec{e}(i) $ is generated by \ref{op:newnode}$(n, k, v, c)$,
  or~$n$ is initialised at~$\vec{e}(i)$.
\end{lemma}
\begin{proof}
  For the events of
  \ref{op:newnode}$(n,k,v,c)$
  the claim holds trivially.
Any event~$e$ generated in
  \ref{op:makevalid}$(n)$ and \ref{op:makevalid}$(n)$
  in \ref{op:insert-ok}$(k,v)$,
  is \po-preceded by a call to
  \ref{op:newnode}$(n,k,v,c)$
  and comes \po-after \ref{lp:insert-ok};
  that makes~$n$ initialised at~$e$.
Every other event that acts on
  $ \loc[n.key] $, $ \loc[n.val] $, $ \loc[n.valid] $,
  $ \loc[n.insFl] $, and $ \loc[n.delFl] $
  is \po-preceded by a read to~$ \loc[n.nxt] $,
  so the claim follows from \cref{lm:reach-init}.
The claim holds for reads and updates acting on $\loc[n.nxt]$
  by \cref{lm:u-init,lm:reach-init}.
\end{proof}

\begin{definition}[Memory safety]
\label{def:mem-safety}
  A call with id~$c \in G.\CallId \dunion \set{\recoveryId} $
  is \emph{memory safe}
  if for all $e \in G.\EvOfCid{c}$,
  $\actOf(e) \ne \Err$ and if $ e \in \UWrites $ then
  there is some~$e_0$ with $e_0 \in G.\Init \lor (e_0 \of \Alloc{\wtv})$
  and $(e_0,e) \in \mo?$. \end{definition}

\begin{lemma}[Memory safety]
\label{lm:mem-safety}
  All calls in~$G.\CallId$ are memory safe.
\end{lemma}
\begin{proof}
  Absence of error can be proven by proving that
    for all~$\vec{e} \in G.\enum{\ghb}$ and
    all~$i<\len{\vec{e}}$,
    $ \vec{e}(i) \ne \Err $,
  using an induction on~$i$,
  combined with~\cref{lm:pointed-at-init,lm:w-nxt-init}.

  That every written-to address is allocated
  is straightforward consequence of \cref{lm:newnode-init}.
\end{proof}

\begin{lemma}
\label{lm:init-key-val}
  For all~$\vec{e} \in G.\enum{\ghb}$,
  for all~$\recEndOf(\vec{e})<i<\len{\vec{e}}$,
  if~$y$ is initialised at $ \vec{e}(i) $
  then for all~$j$ with $
    i \leq j \leq \len{\vec{e}}
  $
  we have
  $
    \mem{\upto{\vec{e}}{i}}(\loc[y.key]) =
    \mem{\upto{\vec{e}}{j}}(\loc[y.key])
  $
  and
  $
    \mem{\upto{\vec{e}}{i}}(\loc[y.val]) =
    \mem{\upto{\vec{e}}{j}}(\loc[y.val])
  $.
\end{lemma}
\begin{proof}
  We prove the statement for~\p{key}, the one for~\p{val} is analogous.
  Since $y$ is initialised at~$\vec{e}(i)$,
  either~$\loc[y.key]$ is written to by an event in $G.E_0$,
  or by the write event in \ref{op:newnode}$(y,\wtv[3])$.
  In either case there is a~$i_0<i$ such that
  $(\vec{e}(i_0) \of \W{y}{key}{k})$.
  By \cref{lm:unique-writes} $\vec{e}(i_0)$ is the only such write event
  of~$\vec{e}$, from which we get
  $ \mem{\upto{\vec{e}}{i}}(\loc[y.key]) = k =
    \mem{\upto{\vec{e}}{j}}(\loc[y.key]) $.
\end{proof}

\begin{lemma}
\label{lm:reading-nxt-key-val}
  For all~$\vec{e} \in G.\enum{\ghb}$,
  for all~$\recEndOf(\vec{e})<i<\len{\vec{e}}$,
  if\/ $ (\vec{e}(i) \of \R{x}{nxt}{\tup{\wtv,y}}) $
  then for all~$j$ with $
    i \leq j \leq \len{\vec{e}}
  $
  we have
  $
    \mem{\upto{\vec{e}}{i}}(\loc[y.key]) =
    \mem{\upto{\vec{e}}{j}}(\loc[y.key]) \ne \bot
  $
  and
  $
    \mem{\upto{\vec{e}}{i}}(\loc[y.val]) =
    \mem{\upto{\vec{e}}{j}}(\loc[y.val]) \ne \bot
  $.
\end{lemma}
\begin{proof}
  By \cref{lm:pointed-at-init,lm:init-key-val}.
\end{proof}

\begin{lemma}
\label{lm:immutable-key}
  For all~$\vec{e} \in G.\enum{\ghb}$,
  and all~$\recEndOf(\vec{e})<i<\len{\vec{e}}$,
  if $ (\vec{e}(i) \of \R{x}{key}{k}) $
  then for all~$j$ with $
    i \leq j \leq \len{\vec{e}}
  $
  we have
  $
    \mem{\upto{\vec{e}}{j}}(\loc[x.key]) = k
  $.
\end{lemma}
\begin{proof}
  Every read of~$\loc[x.key]$ is either
  \po-preceded by a read $ \R{y}{nxt}{\tup{\wtv,x}} $
  or~$x=\var{head}$.
  In the former case we can apply~\cref{lm:reading-nxt-key-val} and we are done.
  In the latter case the statement is true as the only events
  writing to~$\loc[\var{head}.key]$ are in $G.E_0$.
\end{proof}

\begin{lemma}[Visible links]
\label{lm:visible-links}
  For all~$\vec{e} \in G.\enum{\ghb}$,
  and all~$\recEndOf(\vec{e})<i<\len{\vec{e}}$,
  if
$ \vec{e}(i) \in \UReads_{\loc[x.nxt]} $
  and
  $ \rvalOf(\vec{e}(i)) = \tup{b,y} $,
  then there is a~$\recEndOf(\vec{e}) \leq j \leq i$ such that
  $
    \mem{\upto{\vec{e}}{j}}(\loc[x.nxt]) = \tup{b,y}
  $,
  and both~$x$ and~$y$ are initialised at $\vec{e}(j)$.
\end{lemma}
\begin{proof}
  By \cref{lm:reach-init},~$x$ is initialised at~$\vec{e}(i)$.
  Consider the $\rf$ edge to~$\vec{e}(i)$,
  $ (w,\vec{e}(i)) \in \rf $.
  We consider three cases:
  \begin{casesplit}
    \case[$w\in G.E_0$]
      Then, by construction,~$y$ is also initialised in $G.E_0$,
      and so both~$x$ and~$y$ are initialised at~$j=\recEndOf(\vec{e})<i$
      and~$ \mem{\upto{\vec{e}}{\recEndOf(\vec{e})}}(\loc[x.nxt]) = \wvalOf(w) = \rvalOf(\vec{e}(i)) = \tup{b,y} $.
    \case[$w\in\Writes \setminus G.E_0$]
      Then~$w$ is generated by~\ref{op:newnode}$(x,\wtv[3])$
      and so we have
      \[
        (w \of \W{x}{nxt}{\tup{0,y}})
        \po!->
        (u \of \U{z}{nxt}{\wtv}{\tup{0,x}})
        \ghb->
        \vec{e}(i).
      \]
      There exists then~$\recEndOf(\vec{e})<j<i$ such that~$ \vec{e}(j) = u $.
      By \cref{lm:unique-writes} and \cref{lm:u-init}
      $\mem{\upto{\vec{e}}{j}}(\loc[x.nxt]) = y$
      and~$y$ is initialised, as desired.
    \case[$w \in \Updates \setminus G.E_0$]
      If $w\in\Updates$ we have $(w,\vec{e}(i))\in\ghb$,
      so $ w=\vec{e}(j) $ for some~$\recEndOf(\vec{e})<j<i$.
      By \cref{lm:u-init,lm:w-nxt-init}, $x$ and~$y$ are initialised at~$w$, and
      clearly $\mem{\upto{\vec{e}}{j}}(\loc[x.nxt]) = \tup{b,y}$, as desired.
      \qedhere
  \end{casesplit}
\end{proof}

\begin{lemma}[Links are sorted~\ref{inv:sorted}]
\label{lm:sorted-links}
  Let~$ \vec{e} \in G.\enum{\ghb} $, and
  let~$ M_i = \mem{\upto{\vec{e}}{i}} $.
  For every
      ${\recEndOf(\vec{e}) \leq i<\len{\vec{e}}}$,
  and every $x,y \in \Addr$,
  if~$ M_i(\loc[x.nxt]) = \tup{\wtv,y} $
  then
     $ M_i(\loc[x.key]) < M_i(\loc[y.key]) $.
\end{lemma}
\begin{proof}
  By induction on~$i$.
  \begin{induction}
    \step[Base case~$i=\recEndOf(\vec{e})$]
      The statement holds by the assumption that
      $M_{\recEndOf(\vec{e})} \in \recovered(\kvs)$ for some~$\kvs$.
    \step[Induction step~$i>\recEndOf(\vec{e})$]
      Assume the statement holds for every~$\recEndOf(\vec{e}) \leq j < i$.
      If $\vec{e}(i) \in \Reads \union \Flushes \union \Rets$
      then $M_i = M_{i-1}$ and the induction hypothesis yields the lemma.
      For any given~$x$ and~$y$ we can also invoke the induction hypothesis if
      $\locOf(\vec{e}(i)) \inters \set{\loc[x.nxt], \loc[x.key], \loc[y.nxt], \loc[y.key]} = \emptyset$.
      We then do a case analysis on the remaining possible write events.
      For each event modifying address~$a$ we need to check
      the invariant holds for $ x=a $ and for~$y=a$.
      \begin{casesplit}
      \case[$ (\vec{e}(i) \of \Alloc{a}) $ from \ref{op:newnode}]
        We show that
          $ M_i(\loc[a.nxt]) = \tup{0,\nullptr} $, and
          $ \A y. {M_i(\loc[y.nxt]) \ne \tup{\wtv,a}} $.
        In both cases the statement is trivially satisfied.
        \begin{casesplit}
          \case*[{$ M_i(\loc[a.nxt]) = \tup{0,\nullptr} $}]
            By \cref{def:mem}.
          \case*[{$ \A y. M_i(\loc[y.nxt]) \ne \tup{\wtv,a} $}]
            Assume by contradiction that $M_i(\loc[y.nxt]) = \tup{\wtv,a}$
            for some~$y$.
            Then there is a~$j<i$ such that
            $ \vec{e}(j) \in \UWrites_{\loc[y.nxt]}$
            and $ \wvalOf(\vec{e}(j)) = \tup{\wtv,a} $.
            If~$j \leq \recEndOf(\vec{e})$,
            by $M_{\recEndOf(\vec{e})} \in \recovered(\kvs)$
            we would have a double allocation of~$a$.
            Otherwise, by~\cref{lm:w-nxt-init}~$a$ is initialised at $\vec{e}(i)$
            which implies that either there are two allocations of~$a$,
            or that~$j>i$.
            In both cases we reach a contradiction.
        \end{casesplit}
\case[$ (\vec{e}(i) \of \W{a}{key}{k}) $ from \ref{op:newnode}]
        Analogously to the previous case, we show that
        $ M_i(\loc[a.nxt]) = \tup{\wtv,\nullptr} $, and
        $ \A y. {M_i(\loc[y.nxt]) \ne \tup{\wtv,a}} $.
        \begin{casesplit}
          \case*[{$ M_i(\loc[a.nxt]) = \tup{\wtv,\nullptr} $}]
            Assume, by contradiction, that
            $M_i(\loc[a.nxt]) = \tup{\wtv,z}$ for some~$z$.
            Then there is a~$j<i$ such that
            $ \vec{e}(j) \in \UWrites_{\loc[a.nxt]}$
            and $ \wvalOf(\vec{e}(j)) = \tup{\wtv,a} $.
            From \cref{lm:unique-writes} we know the only write to
            $\loc[a.nxt]$ can be the one \po-\emph{following} $\vec{e}(i)$ in \ref{op:newnode}, which means it cannot be~$ \vec{e}(j) $.
            So~$\vec{e}(j)$ must be an update.
            From \cref{lm:u-init} we know that~$a$ must then be initialised
            at~$\vec{e}(j)$, which implies either that~$a$
            has been allocated twice,
            or that~$j>i$.
            In both cases we reach a contradiction.
          \case*[{$ \A y. M_i(\loc[y.nxt]) \ne \tup{\wtv,a} $}]
            Analogous to the corresponding subcase of the allocation case.
        \end{casesplit}
\case[$ (\vec{e}(i) \of \W{a}{nxt}{\tup{0,c}}) $ from \ref{op:newnode}]
        Similarly to the previous cases one can show
        $ \A y. M_i(\loc[y.nxt]) \ne \tup{\wtv,a} $.
        It remains to show that $M_i(\loc[a.key]) < M_i(\loc[c.key])$.
        First,
          $\vec{e}(i)$ is \po-preceded by the write to $\loc[a.key]$
          in \ref{op:newnode}$(a,k,\wtv[2])$, so~$M_i(\loc[a.key]) = k$.
        Second,
          $\vec{e}(i)$ is \po-preceded by two read events,
          $(r_0 \of \R{c}{key}{k'})$ with~$k \leq k'$ in \ref{op:find},
          and
          $(r_1 \of \R{c}{key}{k''})$ with~$k \ne k''$ in \ref{op:insert-ok}.
        By \cref{lm:immutable-key} $M_i(\loc[c.key])=k'=k''$.
        As required we obtain $k < k''$.
\case[$ (\vec{e}(i) \of \U{a}{nxt}{\wtv}{\tup{\wtv,n}}) $
            from \ref{op:insert-ok}]
        We need to show $ M_i(\loc[a.key]) < M_i(\loc[n.key]) $.
        The update \po-follows the events
        $ (r_a \of \R{a}{key}{k_a}) $ in \ref{op:find}$(k,a,\wtv)$ with~$ k_a < k $,
        and
        $ (w_c \of \W{n}{key}{k}) $ in \ref{op:newnode}$(n,k,\wtv[2])$.
        As in the previous case, by \cref{lm:immutable-key}
        we obtain $ M_i(\loc[a.key]) = k_a < k = M_i(\loc[n.key]) $.
\case[$ (\vec{e}(i) \of \U{a}{nxt}{\tup{\wtv,c}}{\tup{\wtv,s}}) $
            from \ref{op:trim}]
        We need to show $ M_i(\loc[a.key]) < M_i(\loc[s.key]) $.
        The event $\vec{e}(i)$ is \po-preceded by the event
        $ (\vec{e}(i_c) \of \R{c}{nxt}{\tup{\wtv,s}}) $.
        By applying \cref{lm:visible-links} to $\vec{e}(i)$ and $\vec{e}(i_c)$
        we obtain~$ \recEndOf(\vec{e}) \leq j_a \leq i $
              and $ \recEndOf(\vec{e}) \leq j_c \leq i_c < i $
        with
        \begin{align*}
          M_{j_a}(\loc[a.nxt]) &= c&
          M_{j_c}(\loc[c.nxt]) &= s
        \end{align*}
        and
          $a,c$ initialised at~$\vec{e}(j_a)$ and
          $c,s$ initialised at~$\vec{e}(j_c)$.
        By \cref{lm:init-key-val}, we have
        \begin{align*}
          M_{j_a}(\loc[a.key]) &=
          M_{i}(\loc[a.key])
          \\
          M_{j_a}(\loc[c.key]) &=
          M_{j_c}(\loc[c.key]) =
          M_{i}(\loc[c.key])
          \\
          M_{j_c}(\loc[s.key]) &=
          M_{i}(\loc[s.key])
        \end{align*}
        Since~$\recEndOf(\vec{e}) \leq j_c<i$ and~$\recEndOf(\vec{e}) \leq j_a<i$ we can apply the induction hypothesis
        on both and obtain:
        \[
          M_{j_a}(\loc[a.key]) <
          M_{j_a}(\loc[c.key]) =
          M_{j_c}(\loc[c.key]) <
          M_{j_c}(\loc[s.key])
        \]
        which implies
        $
          M_{i}(\loc[a.key]) <
          M_{i}(\loc[s.key])
        $
        as desired.
\case[$ (\vec{e}(i) \of \U{a}{nxt}{\tup{0,c}}{\tup{1,c}}) $
            from \ref{op:delete-ok}]
        Immediate from the induction hypothesis.
\qedhere
      \end{casesplit}
  \end{induction}
\end{proof}

\begin{lemma}[Reading links orders keys]
\label{lm:read-link-ord-keys}
  Let~$ \vec{e} \in G.\enum{\ghb} $,
  $\recEndOf(\vec{e}) < i$,
  $ \vec{e}(i) \in \UReads_{\loc[c.nxt]} $, and
  $ \rvalOf(\vec{e}(i)) = \tup{\wtv, s} $.
  Then for all~$j\geq i$
  $
    \mem{\upto{\vec{e}}{j}}(\loc[c.key]) < \mem{\upto{\vec{e}}{j}}(\loc[s.key]).
  $
\end{lemma}
\begin{proof}
  If $ \vec{e}(i) \in \Updates $ the claim follows from~
  \cref{lm:upd-read-last,lm:sorted-links}.
If $ \vec{e}(i) \in \Reads $,
  by \cref{lm:reach-init} $c$ is initialised at $\vec{e}(i)$,
  so there is a~$\recEndOf(\vec{e}) \leq j_0<i$ with
  $ \mem{\upto{\vec{e}}{j_0}}(\loc[c.nxt]) = \tup{\wtv,s} $.
  By \cref{lm:sorted-links}
  \[
    \mem{\upto{\vec{e}}{j_0}}(\loc[c.key])
      < \mem{\upto{\vec{e}}{j_0}}(\loc[s.key]).
  \]
  By \cref{lm:reading-nxt-key-val} for all $j > j_0$,
  $ \mem{\upto{\vec{e}}{j}}(\loc[c.key]) < \mem{\upto{\vec{e}}{j}}(\loc[s.key]) $
  as desired.
\end{proof}

\begin{definition}[Reachable node]
\label{def:reach-node}
  Given some memory~$M$,
  an address~$x\in \Addr$ we write~$\reach(M,x,y)$,
  if there exist~$z_0,\dots,z_n \in \Addr$
  with $z_0 = x$, $z_n = y$ and
    $ M(\loc[z_j.nxt]) = \tup{\wtv,z_{j+1}} $
  for all~$j<n$.
  We write~$\reach(M,x)$ for $\reach(M,\var{head},x)$.
\end{definition}

\begin{lemma}[Marking is irreversible~\ref{inv:marking-irrev}]
\label{lm:marking-irrev}
  Let~$ \vec{e} \in G.\enum{\ghb} $,
  and $\recEndOf(\vec{e}) < i$.
  If~$ \mem{\upto{\vec{e}}{i}}(\loc[x.nxt]) = \tup{1,y} $
  then for all~$j\geq i$,
    $ \mem{\upto{\vec{e}}{j}}(\loc[x.nxt]) = \tup{1,y} $.
\end{lemma}
\begin{proof}
  At allocation, a node is unmarked.
  By \cref{lm:u-init} the write event generated by~\ref{op:newnode}
  overwrites an unmarked node and leaves it unmarked.
  Every update event reads a tuple~$\tup{0,z}$,
  so by \cref{lm:upd-read-last} they never update a marked node.
\end{proof}

\begin{lemma}[\ref{op:trim} is called on marked nodes]
\label{lm:trim-on-marked}
  Let~$ \vec{e} \in G.\enum{\ghb} $, and
  $(\vec{e}(i) \of \U{p}{nxt}{\tup{0,c}}{\tup{0,s}})$
  be generated at~\ref{ev:trim-w-nxt} in \ref{op:trim}$(p,c)$.
  Then~$ \mem{\upto{\vec{e}}{i-1}}(\loc[c.nxt]) = \tup{1,s} $.
\end{lemma}
\begin{proof}
  \ref{op:trim} is used in two contexts.
  \begin{casesplit}
    \case*[In \ref{op:garbage}]
      We have
      $
        (e_c \of \R{c}{nxt}{\tup{1,s'}})
        \po->
        (e_s \of \R{c}{nxt}{\tup{\wtv,s}})
        \po->
        (\vec{e}(i) \of \U{p}{nxt}{\tup{0,c}}{\tup{0,s}}).
      $
      The $\po$ edges are contained in~$\ghb$ so
      $e_c = \vec{e}(i_c)$ and
      $e_s = \vec{e}(i_s)$
      for some $ i_c < i_s < i $.
      By \cref{lm:marking-irrev,lm:uread-nxt-ghb-last} we obtain
      $ s = s' $
      and
      $ \mem{\upto{\vec{e}}{i}}(\loc[c.nxt]) = \tup{1,s} $.

    \case*[In \ref{op:delete-ok}]
      $\vec{e}(i)$ is \po-preceded by an event
      $\vec{e}(i_c) \of \U{c}{nxt}{\tup{0,s}}{\tup{1,s}}$ with $i_c<i$.
      Therefore $ \mem{\upto{\vec{e}}{i_c}}(\loc[c.nxt]) = \tup{1,s} $
      and so by \cref{lm:marking-irrev}
      $ \mem{\upto{\vec{e}}{i}}(\loc[c.nxt]) = \tup{1,s} $.
      \qedhere
  \end{casesplit}
\end{proof}

\begin{lemma}[Member nodes are reachable~\ref{inv:reach}]
\label{lm:member-reach}
  Let~$ \vec{e} \in G.\enum{\ghb} $, and
  let~$ M_i = \mem{\upto{\vec{e}}{i}} $.
  For every
      ${\recEndOf(\vec{e}) \leq i<\len{\vec{e}}}$,
  and~$x \in \Addr$,
  if
    $x$ is initialised at~$\vec{e}(j)$ for some~$j\leq i$, and
    $ M_i(\loc[x.nxt]) = \tup{0,\wtv} $,
  then $\reach(M_i, x)$.
\end{lemma}
\begin{proof}
  By induction on~$i$.
  \begin{induction}
    \step[Base case~$i=\recEndOf(\vec{e})$]
      The statement holds by the assumption that
      $M_{\recEndOf(\vec{e})} \in \recovered(\kvs)$ for some~$\kvs$.
    \step[Induction step~$i>\recEndOf(\vec{e})$]
      Assume the statement holds at~$i-1$.
      The invariant is clearly preserved by any event that does not
      modify the $\p{nxt}$ field.
      We consider all events that do.
      \begin{casesplit}
        \case[$ (\vec{e}(i) \of \Alloc{x}) $ or
              $ (\vec{e}(i) \of \W{x}{key}{k}) $ from \ref{op:newnode}]
          Trivial as $x$ is not initialised at~$\vec{e}(i)$.
\case[$ (\vec{e}(i) \of \U{x}{nxt}{\tup{0,c}}{\tup{0,n}}) $
              from \ref{op:insert-ok}]
          By \cref{lm:upd-read-last,lm:u-init}
          we have that~$x$ is initialised at $\vec{e}(j)$
          for some $\recEndOf(\vec{e}) < j < i$,
          and~$M_{i-1}(\loc[x.nxt])=\tup{0,c}$.
          Therefore, by induction hypothesis,
          $\reach(M_{i-1},x)$.
          Since~$n$ is initialised at~$\vec{e}(i)$ (and not before)
          we know~$M_i(\loc[n.nxt]) = \tup{0,c}$.
          The update gives us~$M_i(\loc[x.nxt]) = \tup{0,n}$.
          Every node reachable at~$i-1$ is still reachable at~$i$,
          and~$n$ is now initialised and reachable via~$x$.
\case[$ (\vec{e}(i) \of \U{x}{nxt}{\tup{0,c}}{\tup{0,s}}) $
              from \ref{op:trim}]
          By \cref{lm:upd-read-last},
          we have that
$M_{i-1}(\loc[x.nxt])=\tup{0,c}$.
By \cref{lm:trim-on-marked},
          $M_i(\loc[c.nxt])=\tup{1,s}$.
          The update gives us~$M_i(\loc[x.nxt]) = \tup{0,s}$.
          Every node that was reachable at~$i-1$, except~$c$,
          is still reachable at~$i$.
          Since~$c$ is marked at~$i$, the statement is satisfied.
\case[$ (\vec{e}(i) \of \U{x}{nxt}{\tup{0,c}}{\tup{1,c}}) $
              from \ref{op:delete-ok}]
          Trivial as the next pointers are not modified.
\qedhere
      \end{casesplit}
  \end{induction}
\end{proof}

\begin{lemma}
\label{lm:valid1-irrev}
For all~$\vec{e} \in G.\enum{\ghb}$,
  and all~$\recEndOf(\vec{e}) \leq i < j < \len{\vec{e}}$,
  if
  $\mem{\upto{\vec{e}}{i}}(\loc[x.valid]) = 1 $
  and
  $ \mem{\upto{\vec{e}}{j}}(\loc[x.valid]) = b $
  then
  $ b=1 $.
\end{lemma}
\begin{proof}
  By \cref{lm:newnode-init}, $x$ is initialised at~$\vec{e}(i)$.
  The only events writing~$0$ to $ \loc[x.valid] $
  are allocations of~$x$ or events in $G.E_0$.
  In both cases such events must appear before~$i$ in $\vec{e}$,
  therefore $
    b = \mem{\upto{\vec{e}}{j}}(\loc[x.valid]) =
    \mem{\upto{\vec{e}}{i}}(\loc[x.valid]) = 1
  $.
\end{proof}

\begin{lemma}
\label{lm:makevalid-irrev}
  For all~$\vec{e} \in G.\enum{\ghb}$,
  and all~$\recEndOf(\vec{e})<i<\len{\vec{e}}$,
  if~$ \vec{e}(i) \notin \Reads $
  is \po-after the last event generated by some call
  \ref{op:makevalid}$(x)$,
then there is an event
  $(e \of \W{n}{valid}{1}) \in G.E$
with
$ (e,\vec{e}(i)) \in \ghb \inters \tso $.
  Moreover, for all $ i \leq j < \len{\vec{e}}  $,
  $ \mem{\upto{\vec{e}}{j}}(\loc[x.valid]) = 1 $.
\end{lemma}
\begin{proof}
  The only events that write~$0$ to $ \loc[x.valid] $ are
  initial events or allocation events, which are both unique per-address.
  We do a case analysis on the last event~$e'$ of \ref{op:makevalid}$(x)$,
  to prove there is some~$ i_1 < i $ such that~$x$ and $ \wvalOf(\vec{e}(i_1)) = 1$.
  \begin{casesplit}
  \case[{$ (e' \of \W{x}{valid}{1}) $}]
    Then $ (e', \vec{e}(i)) \in \po \subs \ghb \inters \tso $ and
    $ \vec{e}(i_1) = e' $ for some~$ i_1 < i $.
    The rest of the claim follows by \cref{lm:valid1-irrev}.
  \case[{$ (e' \of \R{x}{valid}{1}) $}]
    Consider the edge $ (w,e') \in \rf $.
    Then\footnote{In principle, $w$ can be generated by the recovery which might choose to do an update instead of an (equivalent) write. We assume, for simplicity,
      that it uses a write, but none of the proofs depending on this lemma
      fundamentally rely on this assumption.
    }
    $(w \of \W{n}{valid}{1})$ and
    by \cref{lm:wrw-ghb} $(w,\vec{e}(i)) \in \ghb \inters \tso$
    so $ \vec{e}(i_1) = w $ for some~$ \recEndOf(\vec{e}) \leq i_1 < i $.
    The rest of the claim follows by \cref{lm:valid1-irrev}.
    \qedhere
  \end{casesplit}
\end{proof}

\begin{lemma}
\label{lm:makevalid-unmarked}
  For all~$\vec{e} \in G.\enum{\ghb}$,
  and all~$\recEndOf(\vec{e})<i<\len{\vec{e}}$,
  if
  $ (\vec{e}(i) \of \W{x}{valid}{1}) $
  then either
  $
    \mem{\upto{\vec{e}}{i-1}}(\loc[x.valid]) = 1
  $, or
  $
    \mem{\upto{\vec{e}}{i-1}}(\loc[x.nxt]) = \tup{0,\wtv}
  $.
\end{lemma}
\begin{proof}
  We have $ \vec{e}(i) $ is generated by \ref{op:makevalid}$(x)$.
  Each such call is \po-preceded by either:
  \begin{itemize}
  \item $ (e_0 \of \W{x}{nxt}{\tup{0,\wtv}}) $
    generated by \ref{op:newnode} in
    \ref{op:insert-ok}.
  \item $ (e'_0 \of \R{x}{nxt}{\tup{0,\wtv}}) $
    generated by \ref{op:find} in
    \ref{op:insert-no},
    \ref{op:delete-ok},
    \ref{op:delete-no}.
    Then let $e_0 \in \UWrites_{\loc[x.nxt]}$
    be such that~$ (e_0, e'_0) \in \rf $.
    Note that if $e_0 \in G.E_0$ then for every other event $ e_0'' \in G.E_0 \inters \UWrites_{\loc[x.nxt]}$ we have $ e_0'' \ghb-> e_0 $.
  \end{itemize}
  In both cases we have
  $ e_0 \in \UWrites_{\loc[x.nxt]} $,
  $ \wvalOf(e_0) = \tup{0,\wtv} $, and
  $ (e_0, \vec{e}(i)) \in \ghb $
  so there is a~$i_0 < i$ with~$ \vec{e}(i_0) = e_0 $.
If $ \mem{\upto{\vec{e}}{i-1}}(\loc[x.nxt]) = \tup{0,\wtv} $
  we are done.
  Otherwise, there is a largest $ i_1 < i $
  such that~${\wvalOf(\vec{e}(i_1)) = \tup{1,\wtv}}$.
  If~$ i_1 < i_0 $ then we would have
  $ \mem{\upto{\vec{e}}{i-1}}(\loc[x.nxt]) = \tup{0,\wtv} $
  which is a contradiction.
  Since~$ i_1 > i_0 $,
  $\vec{e}(i_1) \notin G.E_0$, otherwise we would have~$e_0 \in G.E_0$
  and $ e_0 \ghb-> \vec{e}(i_1) $.
  The only event generated by the program compatible with~$ \vec{e}(i_0) $
  is an update $ \U{x}{nxt}{\tup{0,\wtv}}{\tup{1,\wtv}} $
  (from \ref{lp:delete-ok}).
  Such an event is \po-preceded by a call to \ref{op:makevalid}.
  Therefore, by \cref{lm:makevalid-irrev} we have
  $\mem{\upto{\vec{e}}{i-1}}(\loc[x.valid]) = 1$
  which proves our claim.
\end{proof}

\subsection{Linearizability}

In this section we focus on proving
\cref{cond:perst-valid} of the \masterthm.
That is, we provide~$\lp$, $\hres{\lp}$, $r$ and~$\volatile$ such that
$\tup{\lp, \hres{\lp}, r, \ghb, \volatile}$ \pre\recEndOf-validates
any execution of~$\LinkFreeImpl[op]$
(assuming soundness of $\LibImpl[rec]$).

\begin{definition}[Linearization points for \LinkFree]
\label{def:linkfree:linpt}
  Given an execution~$G$ of $ \LinkFreeImpl[op] $,
  we define the finite partial functions
  $
    \lp \from G.\CallId \pto G.\UReads
  $
  and
  $
    r \from G.\CallId \pto \V
  $
  as the smallest such that:
  \begin{itemize}
    \item if $\callOf(c) = \tup{\p{insert}, \wtv}$ then
      \begin{itemize}
        \item If there is an event $e\in G.\EvOfCid{c}$ generated
              at~\ref{lp:insert-ok}
              then $ \lp(c) = e $ and $r(c) = \p{true}$.
        \item If there is an event $(r\of\Ret{\p{false}}) \in G.\EvOfCid{c}$
              then there is an event~$ e \in G.\EvOfCid{c} $ generated at
              \ref{lp:find} in \ref{lp:insert-no};
              then we set $ \lp(c) = e $ and $r(c) = \p{false}$.
      \end{itemize}
    \item if $\callOf(c) = \tup{\p{delete}, \wtv}$ then
      \begin{itemize}
        \item If there is an event $e\in G.\EvOfCid{c}$ generated
              at~\ref{lp:delete-ok}
              then $ \lp(c) = e $ and $r(c) = \p{true}$.
        \item If there is an event $(r\of\Ret{\p{false}}) \in G.\EvOfCid{c}$
              then
              $r(c)=\p{false}$ but
              we leave $\lp(c) = \bot$ as this call is handled by~$\hres{\lp}$.
\end{itemize}
  \end{itemize}
\end{definition}

\begin{lemma}
  $\tup{\lp, r, \ghb, \volatile}$ is a linearization strategy.
\end{lemma}
\begin{proof}
  Straightforward by inspecting \cref{def:linkfree:linpt}.
\end{proof}

\begin{theorem}
\label{th:linkfree:lp-validates}
  Assume an arbitrary \pre\tup{\durable,\recovered}-sound~$\LibImpl[rec]$.
  If~$G$ is an execution of $\tup{\LinkFreeImpl[op],\LibImpl[rec]}$,
  then $\tup{\lp, r, \ghb, \volatile}$ \pre\recEndOf-validates~$G$.
\end{theorem}

\begin{proof}
To prove it \pre\recEndOf-validates~$G$,
  consider an arbitrary $ \vec{e} \in G.\enum{\ghb} $,
  an index~$i$ with $ \recEndOf(\vec{e}) < i < \len{\vec{e}} $,
  and a state~$ \kvs \in \KVS' $,
  such that~$ \mem{\upto{\vec{e}}{i-1}} \in \volatile(\kvs) $.
  By \cref{def:linkfree:volatile},
  $
    \mem{\upto{\vec{e}}{i-1}} =
      \volM_{\lbl{s}} \dunion
      \volM_{\lbl{d}} \dunion
      \volM_{\lbl{u}}
  $ for some
    $\volM_{\lbl{s}}$,
    $\volM_{\lbl{d}}$,
    and $\volM_{\lbl{u}}$.
  Let~$ \LinPt = \set{ \lp(c) | c \in \dom(\lp) } $.

  We then check \cref{cond:stutter,cond:linpt-trans}
  by a case analysis on~$\vec{e}(i)$.
  \begin{casesplit}
  \case[$ \vec{e}(i) \notin \LinPt $]
    Then only \cref{cond:stutter} applies.
    We prove that
    $ \mem{\upto{\vec{e}}{i}} \in \volatile(\kvs) $
    by case analysis.
    \begin{casesplit}
\case[$ \vec{e}(i) \in \Reads \union \Flushes $]
      Then, trivially,
      $\mem{\upto{\vec{e}}{i}}=\mem{\upto{\vec{e}}{i-1}} \in \volatile(\kvs)$.
\case[$ (\vec{e}(i) \of \W{n}{insFl}{1}) $ from~\ref{op:flush-ins}]
      If $\loc[n.insFl] \in \dom(\volM_{\lbl{s}})$,
      then \[
        \mem{\upto{\vec{e}}{i}} =
          \volM_{\lbl{s}}[\loc[n.insFl] \mapsto 1] \dunion
          \volM_{\lbl{d}} \dunion
          \volM_{\lbl{u}}
          \in \volatile(\kvs).
      \]
      The cases
      where $\loc[n.insFl] \in \dom(\volM_{\lbl{d}})$
      or $\loc[n.insFl] \in \dom(\volM_{\lbl{u}})$
      are analogous, as none of these sets constrain the \p{insFl} field.
\case[$ (\vec{e}(i) \of \W{n}{delFl}{1}) $ from~\ref{op:flush-del}]
      The events of \ref{op:flush-del} are generated
      in calls to \ref{op:trim}, which occur in~\ref{op:garbage} and \ref{op:delete-ok}.
      The occurrence in~\ref{op:garbage} is \po-preceded by a read
      $ (r\of\R{n}{nxt}{\tup{1,\wtv}}) $,
      so by \cref{lm:uread-nxt-ghb-last,lm:marking-irrev}
      $ \mem{\upto{\vec{e}}{i-1}}(\loc[n.nxt]) = \tup{1,\wtv} $.
      Similarly, the occurrence in \ref{op:delete-ok} is \po-preceded by
      an event $ (w \of \U{n}{nxt}{\tup{0,\wtv}}{\tup{1,\wtv}}) $, and
      consequently $ (w,\vec{e}(i)) \in\ghb $.
      In this case too we conclude, from \cref{lm:marking-irrev},
      that $ \mem{\upto{\vec{e}}{i-1}}(\loc[n.nxt]) = \tup{1,\wtv} $.
      This implies that $ \loc[n.nxt] \in \dom(\volM_{\lbl{d}}) $.
      Since $\volM_{\lbl{d}}$ does not constrain the \p{delFl} field,
      $\mem{\upto{\vec{e}}{i}} \in \volatile(\kvs)$.
\case[$ (\vec{e}(i) \of \U{p}{nxt}{\tup{0,c}}{\tup{0,s}}) $ from~\ref{op:trim}]
      By \cref{lm:trim-on-marked},
      $ \mem{\upto{\vec{e}}{i-1}}(\loc[c.nxt]) = \tup{1,s} $.
      Moreover, by \cref{lm:upd-read-last},
      $ \mem{\upto{\vec{e}}{i-1}}(\loc[p.nxt]) = \tup{0,c}$.
      This implies,
        $ \loc[c.nxt] \in \dom(\volM_{\lbl{d}}) $ and
        $ \loc[p.nxt] \in \dom(\volM_{\lbl{s}}) $.
      By \cref{lm:read-link-ord-keys}
      $
        \mem{\upto{\vec{e}}{i-1}}(\loc[p.key]) <
        \mem{\upto{\vec{e}}{i-1}}(\loc[c.key]) <
        \mem{\upto{\vec{e}}{i-1}}(\loc[s.key]).
      $
      We thus have
      $
        \mem{\upto{\vec{e}}{i}} =
          \volM_{\lbl{s}}[\loc[p.nxt] \mapsto \tup{0,s}] \dunion
          \volM_{\lbl{d}} \dunion
          \volM_{\lbl{u}}
          \in \volatile(\kvs).
      $
\case[$ (\vec{e}(i) \of \W{n}{valid}{1}) $ from~\ref{op:makevalid}]
      By \cref{lm:makevalid-unmarked} we have two cases:
      \begin{casesplit}
      \case[{$ \mem{\upto{\vec{e}}{i-1}}(\loc[n.valid]) = 1 $}]
        Then the claim follows trivially since the memory is unchanged.
      \case[{$ \mem{\upto{\vec{e}}{i-1}}(\loc[n.nxt]) = \tup{0,\wtv} $}]
        By \cref{lm:newnode-init} $n$ is initialised at~$\vec{e}(i)$;
        by \cref{lm:member-reach} we have~$ \reach(\mem{\upto{\vec{e}}{i}}, n) $ holds.
        Since $\reach$ is insensitive to the value of \p{valid},
        we have $ \reach(\mem{\upto{\vec{e}}{i-1}}, n) $.
        By \cref{lm:linkfree:volatile-reach-sd} we have that
        $ n \in X_{\lbl{s}} \dunion X_{\lbl{d}} $.

        We cannot have $ n \in X_{\lbl{d}} $ because all such nodes are marked.
        From $ n \in X_{\lbl{s}} $ we have
        $
          \mem{\upto{\vec{e}}{i}} =
            \volM_{\lbl{s}}[\loc[n.valid] \mapsto 1] \dunion
            \volM_{\lbl{d}} \dunion
            \volM_{\lbl{u}}
            \in \volatile(\kvs).
        $
      \end{casesplit}
\case[$ \vec{e}(i) \in\Writes $ from~\ref{op:newnode}]
      The allocation~$\Alloc{n}$
      has the effect of adding a new node to~$X_{\lbl{u}}$;
      note that since we assume allocated memory is zeroed,
      we automatically have $\loc[n.valid] = 0$.
      The writes to the other fields are irrelevant
      for the abstract key value store as they modify an uninitialised node.
      The write to the \p{nxt} field preserves
      \eqref{cond:linkfree:sorted-links} by virtue of
      \cref{lm:sorted-links}.
    \end{casesplit}
  \case[$ \vec{e}(i) \in \LinPt $]
    Let $c \in \dom(\lp)$ be so that
    $\vec{e}(i) = \lp(c)$.
    \begin{casesplit}
\case[$ \vec{e}(i) $ generated at \ref{lp:find}]
      By \cref{def:linkfree:linpt}, the operations returned
        $ r(c) = \p{false} $,
      and
        $ \callOf(c) = \tup{\p{insert}, \tup{k,v}} $.
Since $\vec{e}(i) \in \Reads$,
      $\mem{\upto{\vec{e}}{i}}=\mem{\upto{\vec{e}}{i-1}} \in \volatile(\kvs)$
      and we have to prove
      $ (\kvs,\kvs) \in \Delta(\p{insert}, \tup{k,v}, \p{false}) $.
It suffices to prove $k \in \dom(\kvs)$.
        Since the operation returned, we have
        an event $r \in G.\EvOfCid{c}$ with
        \[
          \vec{e}(i) = (\lp(c) \of \R{n}{nxt}{\tup{0,\wtv}})
          \ghb->
          (r \of \R{n}{key}{k}) = \vec{e}(j)
        \]
        for some~$j > i$.
        By \cref{lm:reach-init}, $n$ is initialised at $\vec{e}(i)$.
        By \cref{lm:member-reach,lm:uread-nxt-ghb-last,lm:uread-key-ghb-last}
        we have
        $ \mem{\upto{\vec{e}}{i}}(\loc[x.nxt]) = \tup{0,\wtv} $, and
        $ \mem{\upto{\vec{e}}{j}}(\loc[x.key]) = k $, and
        $ \reach(\mem{\upto{\vec{e}}{i}}, n) $.
        By \cref{lm:immutable-key},
        $ \mem{\upto{\vec{e}}{i}}(\loc[x.key]) = k $.
        $\mem{\upto{\vec{e}}{i}}=\mem{\upto{\vec{e}}{i-1}} \in \volatile(\kvs)$,
        we have that~$n$ must belong, reachable and unmarked,
        to $ \volM_{\lbl{s}} $,
        and so $k \in \kvs$.
\case[$ \vec{e}(i) $ generated at \ref{lp:insert-ok}]
      By \cref{def:linkfree:linpt},
        $ r(c) = \p{true} $
      and
        $ \callOf(c) = \tup{\p{insert}, \tup{k,v}} $.
      We show that it must be the case that~$ k \notin \dom(\kvs) $,
      $n$ is fresh in~$\kvs$,
      and that
      $
        \mem{\upto{\vec{e}}{i}}
          \in \volatile(\kvs \dunion \tup{k,\tup{n,v}})
      $
      which would suffice, since
      $
        (\kvs,\kvs \dunion \tup{k,\tup{n,v}}) \in
          \Delta(\p{insert}, \tup{k,v}, \p{true}).
      $

      We have $ (\vec{e}(i) \of \U{p}{nxt}{\tup{0, a}}{\tup{0, n}}) $,
      therefore, by \cref{lm:upd-read-last},
      $ \mem{\upto{\vec{e}}{i-1}}(\loc[p.nxt]) = \tup{0,a} $.

      The event~$\vec{e}(i)$ is \po-preceded by a call to
      \ref{op:newnode}$(n, k, v, a)$,
      but $n$ is not yet initialised at any $\vec{e}(j)$ with $j<i$.
      Therefore all the writes to fields of~$n$
      appearing before~$i$ in~$\vec{e}$ were generated by \ref{op:newnode}$(n, k, v, a)$,
      or there would be a contradiction with \cref{lm:newnode-init}.
      From this we obtain
      \begin{align*}
        \mem{\upto{\vec{e}}{i-1}}(\loc[n.key]) &= k
        &
        \mem{\upto{\vec{e}}{i-1}}(\loc[n.val]) &= v
        &
        \mem{\upto{\vec{e}}{i-1}}(\loc[n.nxt]) &= \tup{0,a}
        \\
        \mem{\upto{\vec{e}}{i-1}}(\loc[n.valid]) &= 0
        &
        \mem{\upto{\vec{e}}{i-1}}(\loc[n.insFl]) &= 0
        &
        \mem{\upto{\vec{e}}{i-1}}(\loc[n.delFl]) &= 0
      \end{align*}

      By \cref{lm:immutable-key} and the reads that \po-precede $\vec{e}(i)$,
      $ \mem{\upto{\vec{e}}{i-1}}(\loc[p.key]) < k < \mem{\upto{\vec{e}}{i-1}}(\loc[p.key]) $.
      By \cref{lm:u-init,lm:member-reach}, we have
      $ \reach(\mem{\upto{\vec{e}}{i-1}}, p) $,
      which, by \cref{lm:linkfree:volatile-reach-sd} implies
      $ p \in X_{\lbl{s}} $.
      Then by \eqref{cond:linkfree:sorted-links}
      any node in $X_{\lbl{s}}$ storing key~$k$ would
      need to be between~$p$ and~$a$ which is a contradiction.
      This proves $ k \notin \dom(\kvs) $.

      Since~$n$ is unmarked and invalid in~$ \mem{\upto{\vec{e}}{i-1}} $,
      it must belong to~$X_{\lbl{u}}$.
      This also implies~$n$ is fresh in~$\kvs$.
      We split $ \volM_{\lbl{u}} = M_0 \dunion M_n$
      such that $ \dom(M_n) = \set{\loc[n.f] | \p{f} \in \Field } $.
      Then:
      \[
        \mem{\upto{\vec{e}}{i}} =
          (\volM_{\lbl{s}}[\loc[p.nxt] \mapsto \tup{0,n}]
          \dunion M_n) \dunion
          \volM_{\lbl{d}} \dunion
          M_0
          \in \volatile(\kvs \dunion \tup{k,\tup{n,v}}).
      \]
\case[$ \vec{e}(i) $ generated at \ref{lp:delete-ok}]
      By \cref{def:linkfree:linpt},
        $ r(c) = \p{true} $
      and
        $ \callOf(c) = \tup{\p{delete}, k} $.
      We show that it must be the case that~$ k \in \dom(\kvs) $,
      and that
      $
        \mem{\upto{\vec{e}}{i}}
          \in \volatile(\kvs \setminus \tup{k,\kvs(k)})
      $
      which would suffice, since
      $
        (\kvs,\kvs \setminus \tup{k,\kvs(k)}) \in
          \Delta(\p{delete}, k, \p{true}).
      $

      We have $ \vec{e}(i) \of \U{a}{nxt}{\tup{0, s}}{\tup{1, s}} $,
      therefore, by \cref{lm:upd-read-last},
      $ \mem{\upto{\vec{e}}{i-1}}(\loc[a.nxt]) = \tup{0,s} $.
      By \cref{lm:u-init,lm:member-reach}, we have
      $ \reach(\mem{\upto{\vec{e}}{i-1}}, a) $,
      which, by \cref{lm:linkfree:volatile-reach-sd} implies
      $ {a \in X_{\lbl{s}}} $.
      Moreover, by \ref{cond:linkfree:members-reachable},
      $s \in \volM_{\lbl{s}} \union \volM_{\lbl{d}}$.
      Since $ \vec{e}(i) $ is \po-preceded by
      a call to \ref{op:makevalid}$(a)$,
      we know that, by \cref{lm:makevalid-irrev},
      $ \mem{\upto{\vec{e}}{i-1}}(\loc[a.valid]) = 1 $.
      We thus can partition $ \volM_{\lbl{s}} = M_0 \dunion M_c$
      such that $ \dom(M_a) = \set{\loc[a.f] | \p{f} \in \Field } $.
      Then:
      \[
        \mem{\upto{\vec{e}}{i}} =
          M_0 \dunion
          (\volM_{\lbl{d}} \dunion M_a^{\strut}[\loc[a.nxt] \mapsto \tup{1,s}]) \dunion
          \volM_{\lbl{u}}
          \in \volatile(\kvs \setminus \tup{k,\kvs(k)}).
      \qedhere
      \]
    \end{casesplit}
  \end{casesplit}
\end{proof}

\subsubsection{Hindsight}

We now handle failed deletes by hindsight linearization.
We start by some basic transition invariant on reachability of nodes.
Then we prove \cref{lm:linkfree:hind-delete}, the ``hindsight lemma''
for \p{delete} which states that there is a point during the execution
of \p{delete}$(h,k)$ where the state encoded in volatile memory
does not contain~$k$.
The proof can, crucially,
rely on having already established by \cref{th:linkfree:lp-validates}
that the \vo-induced linearization of the non-hindsight calls is legal.
We then use the hindsight lemma to define a suitable $\hres{\lp}$.

\begin{lemma}
\label{lm:reach-preserv}
  Let $(\kvs_1, \kvs_2) \in \Delta(\wtv[2])$,
      $ \vec{e} \in \enum{\ghb[G]} $,
      $ M_j = \mem{\upto{\vec{e}}{j}} $, and
      $i\geq\recEndOf(\vec{e})$.
  Moreover, assume
    $M_i \in \volatile(\kvs_1)$,
    $M_{i+1} \in \volatile(\kvs_2)$,
    $\kvs_1(k)=\kvs_2(k)=\tup{y,v}$,
  and $ \reach(M_i,x) $ or $M_i(\loc[x.nxt])=\tup{1,\wtv}$.
  Then, $\reach(M_i,x,y) \implies \reach(M_{i+1},x,y)$.
\end{lemma}
\begin{proof}
  Let $\reach(M_i,x,y)$ be witnessed by
  linked nodes $x_0,x_1,\dots,x_n$
  where $x_0=x$ and $x_n=y$.
  Let~$j \leq n$ be the smallest index such that
  $ M_i(\loc[x_j.nxt]) = \tup{0,\wtv} $ ---
  since~$y$ is unmarked, we know at least one such~$j$ exists.
  From \cref{lm:marking-irrev} we know that the links
  between the nodes from~$x_0$ to $x_{j-1}$ (if any)
  will be preserved in~$M_{i+1}$;
  the claim would therefore follow from proving that
  $ \reach(M_{i+1}, x_j, y) $.

  First we can show~$\reach(M_i, x_j)$.
  We have two cases.
  If $ x = x_j $, then, since the node is unmarked,
  by assumption we have~$\reach(M_i,x)$.
  Otherwise, $ x \ne x_j $ and $ x $ is marked.
  Then, by $M_i \in \volatile(\kvs_1)$ and \eqref{cond:linkfree:marked-nxt-init},
  we have $\reach(M_i, x_j)$.
  
  Then, from $\reach(M_i, x_j)$ we conclude~$ \tup{x_j,k'} \in \kvs_1 $
  for some~$k'\leq k$.
  Since the move from $\kvs_1$ to $\kvs_2$ is legal,
  we have two cases.
  Either $\tup{x_j,k'} \in \kvs_2$,
  in which case $ \reach(M_{i+1}, x_j, y) $ ---
  by $\kvs_2(k)=\tup{y,v}$, $M_{i+1} \in \volatile(\kvs_2)$,
  \eqref{cond:linkfree:sorted-links} and \eqref{cond:linkfree:members-reachable}.
  Otherwise, $\vec{e}_{i+1}$ is the linearization point of a deletion of~$k'$.
  In this case, the event simply marks~$x_j$, thus preserving the links.
\end{proof}

\begin{lemma}[Hindsight for \p{delete}]
\label{lm:linkfree:hind-delete}
  Given any $ \vec{e} \in \enum{\ghb[G]} $,
  if $(\ret \of \Ret{\p{false}}) \in \vec{e}$
  and $\callOf(\ret) = \tup{\p{delete}, k}$,
  then there exist
  $j_1 \leq i \leq j_2$ such that
  $ \cidOf(\vec{e}(j_1)) = \cidOf(\vec{e}(j_2)) = \cidOf(\ret) $,
  $ \vec{e}(j_1),\vec{e}(j_2) \in \Reads $, and
  $ \mem{\upto{\vec{e}}{i}} \in \volatile(\kvs)$
  for some $\kvs$ with $ k \notin \kvs $.
\end{lemma}
\begin{proof}
  We prove the claim by instantiating
    $ \vec{e}(j_1) = \ref{ev:find:rp} \in \Reads $ and
    $ \vec{e}(j_2) = \ref{lp:find} \in \Reads $,
  in the call to $ \ref{op:find}(k,p,c) $
  in \ref{op:delete-no}.
  Note that $ \vec{e}(j_1) \po-> \vec{e}(j_2) $ and thus $j_1 < j_2$.
  Let~$ M_j = \mem{\upto{\vec{e}}{j}} $.
  Let~$n_1$ and $n_2$ be such that
  $ (\vec{e}(j_1) \of \R{n_1}{nxt}{\tup{0,\wtv}}) $
  and
  $ (\vec{e}(j_2) \of \R{n_2}{nxt}{\tup{0,\wtv}}) $.
  Both reads are immediately \po-followed by reads of the keys~$k_1$ and~$k_2$
  of~$n_1$ and~$n_2$, respectively.
  The keys are such that $ k_1 < k < k_2 $.
  By \cref{lm:reach-init} $n_1$ is initialized at $\vec{e}(j_1)$,
  and $n_2$ is initialized at $\vec{e}(j_2)$.
  By \cref{lm:init-key-val} and \cref{lm:uread-key-ghb-last} we have
  $ M_{j_1}(\loc[n_1.key]) = k_1 $,
  $ M_{j_2}(\loc[n_2.key]) = k_2 $.
  Furthermore, by \cref{lm:uread-key-ghb-last}
  $ M_{j_1}(\loc[n_1.nxt]) = \tup{0, \wtv} $, and
  $ M_{j_2}(\loc[n_2.nxt]) = \tup{0, \wtv} $.
  By \cref{lm:member-reach}, we have
  $ \reach(M_{j_1}, n_1) $, and
  $ \reach(M_{j_2}, n_2) $.

  We know from \cref{th:linkfree:lp-validates} that for all~$j_1 \leq j \leq j_2$
  $M_j \in \volatile(\kvs_j)$ for some legal sequence of stores
  $\kvs_j \in \KVS'$.
  Now assume, towards a contradiction, that the required~$i$ does not exist.
  This implies that, for some $ \tup{y,v} $,
  $\kvs_j(k) = \tup{y,v}$, for all~$j_1 \leq j \leq j_2$.
  The pair associated with~$k$ must be constant: to change it in a legal way,
  there would be first a deletion of $k$, which would lead to a contradiction.

  Note that~$n_1$ will still be either unmarked and reachable or marked
  in any $M_j$ with $ j_1 \leq j \leq j_2 $,
  so we can iterate \cref{lm:reach-preserv} to obtain
  $\reach(M_j, n_1, y)$.

  Between $\vec{e}(j_1)$ and $\vec{e}(j_2)$,
  \ref{op:garbage} traverses a (possibly zero) number of marked nodes.
  Consider the first read $ \vec{e}(j_1') \of \R{n_1}{nxt}{\tup{\wtv, n_1'}} $
  in \ref{op:garbage}, with $j_1< j_1' < j_2$.
  From~$\reach(M_{j_1'}, n_1, y)$ and $M_{j_1'}(\loc[n_1.nxt]) = \tup{\wtv, n_1'}$ (by \cref{lm:uread-nxt-ghb-last}) we obtain
  $\reach(M_{j_1'}, n_1', y)$.

  In the case where~$n_1'$ is found marked,
  \ref{op:garbage} would proceed by reading its next pointer
  at some $ (\vec{e}(j_2') \of \R{n_1'}{nxt}{\tup{0,n_2'}}) $.
  Again by iterating \cref{lm:reach-preserv} we obtain
  $\reach(M_{j_2'}, n_1', y)$ and we can repeat the argument.

  Finally, if we find no more garbage ahead, we have $n_1' = n_2$.
  Then, by \cref{lm:marking-irrev} and $M_{j_2}(\loc[n_2.nxt]) = \tup{0, \wtv}$,
  $n_2$ is unmarked in~$M_{j_1'}$.
  By \cref{lm:member-reach} $\reach(M_{j_1'}, n_2)$.
  This leads to a contradiction with $\reach(M_{j_1'}, n_2, y)$ and $ k < k_2 $:
  $n_2$ and $y$ are both unmarked and reachable,
  but $\reach(M_{j_1'}, n_2, y)$ would contradict the sortedness invariant.
\end{proof}

\begin{definition}[Hindsight resolution for \LinkFree]
\label{def:linkfree:hres}
  Given an execution~$G$ of $ \LinkFreeImpl[op] $,
  we define the finite partial function
  $
    \hres{\lp} \from \enum{\ghb[G]} \to \CallId \pto \Nat
  $
  by letting $ \hres{\lp}(\vec{e})(c) \is i $
  if $\callOf(\ret) = \tup{\p{delete}, k}$ and
  $(\ret \of \Ret{\p{false}}) \in \vec{e}$
  and~$i$ is the \emph{minimal} index found in \cref{lm:linkfree:hind-delete};
  $\hres{\lp}(\vec{e})(c)$ is undefined otherwise.
\end{definition}

\begin{theorem}
\label{th:linkfree:hres-validates}
  $\tup{\lp, \hres{\lp}, r, \ghb, \volatile}$ is a linearization strategy.
  Moreover,
  $\tup{\lp, \hres{\lp}, r, \ghb, \volatile}$
  \pre\recEndOf-validates~$G$.
\end{theorem}
\begin{proof}
  Direct consequence of \cref{lm:linkfree:hind-delete,th:linkfree:lp-validates}.
\end{proof}

\subsection{Durable Linearizability}

We now define the persistency thresholds of operations,
and prove durable linearizability.

\begin{definition}[Persistency thresholds for \LinkFree]
\label{def:linkfree:pers-pt}
  Given an execution~$G$ of $ \LinkFreeImpl[op] $,
  we define the finite partial function
  $
    \pt \from G.\CallId \pto G.\Persisted
  $
  and set $ \PRd \subs G.\CallId $,
  as the smallest such that:
  \begin{itemize}
    \item if $\callOf(c) = \tup{\p{insert}, \wtv}$ then
      \begin{itemize}
        \item if $\lp(c) = (e \of \U{p}{nxt}{\wtv}{\tup{0, n}} )$ then
          $
            \pt(c) = \min_{\nvo}\set{e' | (e'\of\W{n}{valid}{1}) \in G.\Persisted}
          $\\
          (which is undefined if~$G.\Persisted$ contains no write to \loc[n.valid])
        \item if there is an event $(r\of\Ret{\p{false}}) \in G.\EvOfCid{c}$
              then~$ c \in \PRd $.
      \end{itemize}
    \item if $\callOf(c) = \tup{\p{delete}, \wtv}$ then
      \begin{itemize}
        \item if there is an event $e\in G.\EvOfCid{c} \inters G.\Persisted$ generated
              at~\ref{lp:delete-ok}
              then $ \pt(c) = e $
        \item if there is an event $(r\of\Ret{\p{false}}) \in G.\EvOfCid{c}$
              then~$ c \in \PRd $.
      \end{itemize}
  \end{itemize}
   By construction we have $\pt$ is injective,
   $\dom(\pt) \subs \dom(\lp)$,
   $ \PRd \inters \dom(\pt) = \emptyset $,
   and
   all calls in $ \PRd $ do not modify the abstract state.
\end{definition}

\subsubsection{Flush Before Returning}

\begin{lemma}[\ref{op:flush-ins} flushes]
\label{lm:flush-ins-before-ret}
  Assume~$G$ is an execution containing all the events
  generated by a call to
  \ref{op:flush-ins}$(n)$,
  for some~$n$ not written to by initial or recovery events.
  Then there are
  $ w, \var{fl}, e_1 \in G.E $ such that
  \[
  \smash{
    (w \of \W{n}{valid}{1})
    \nvo->
    (\var{fl} \of \FL{n})
    \nvo->
    (e_1 \of \W{n}{insFl}{1}).
  }
  \]
\end{lemma}
\begin{proof}
  We do a case analysis on the events generated by \ref{op:flush-ins}$(n)$.
  \begin{casesplit}
  \case[$
      (e_0 \of \R{n}{insFl}{0})
      \po->
      (\var{fl} \of \FL{n})
      \po->
      (e_1 \of \W{n}{insFl}{1})
    $]
    Then by \cref{lm:makevalid-irrev}
    there is a $ (w \of \W{n}{valid}{1}) $ with
    $ w \tso-> \var{fl} $, which implies $ w \nvo-> \var{fl} $.
    Notice that $\var{fl} \nvo-> e_1$.
  \case[$(e_0 \of \R{n}{insFl}{1}) $]
    Consider the \rf-edge $ (e_1 \of \W{n}{insFl}{1}) \rf-> e_0 $.
    We have that since~$n$ is not written to by initial or recovery events,
    $ e_1 $ must be generated by some \ref{op:flush-ins}$(n)$.
    By definition we thus have $(\var{fl} \of \FL{n}) \po-> e_1$
    which implies $ \var{fl} \nvo-> e_1 $.
    Since every call to \ref{op:flush-ins}$(n)$ is \po-preceded by
    a call \ref{op:makevalid}$(n)$, we can apply \cref{lm:makevalid-irrev}
    to~$\var{fl}$ and obtain the desired~$(w \of \W{n}{valid}{1})$.
  \qedhere
  \end{casesplit}
\end{proof}

\begin{lemma}[Flush before delete returns]
\label{lm:flush-del-before-ret}
  Let
    $ (\ret \of \Ret{\p{true}}) \in G.\EvOfCid{c}$ and
    $ (u \of \U{n}{nxt}{\wtv}{\wtv}) \in G.\EvOfCid{c} $
    generated at~\ref{lp:delete-ok}
  for some~$n$.
  Then there are
  $ \var{fl}, e_1 \in G.E $ such that
  \[
  \smash{
    (u \of \U{n}{nxt}{\tup{0,\wtv}}{\tup{1,\wtv}})
    \nvo->
    (\var{fl} \of \FL{n})
    \nvo->
    (e_1 \of \W{n}{delFl}{1}).
  }
  \]
\end{lemma}
\begin{proof}
We do a case analysis on the events generated by \ref{op:flush-del}$(n)$
in the call to~\ref{op:trim}$(\wtv,n)$
that \po-precedes~$\ret$.
  \begin{casesplit}
  \case[$
    \R{n}{delFl}{0} \po-> (\var{fl}\of\FL{n}) \po-> (e_1\of\W{n}{delFl}{1})
  $]
    The call to \ref{op:trim} is \po-after~$u$.
    Therefore we have $ u \nvo-> \var{fl} \nvo-> e_1 $ as desired.

  \case[$ (e_0 \of \R{n}{delFl}{1}) $]
    Consider the \rf-edge $ (e_1 \of \W{n}{delFl}{1}) \rf-> e_0 $.
    The event~$e_1$ must come from a call to~\ref{op:flush-del}$(n)$.
        Such call would also generate~$(\var{fl} \of \FL{n}) \po-> e_1$.
    \ref{op:flush-del}$(n)$ is only called by \ref{op:trim}$(n)$,
    which in turn is called in two contexts:
    \begin{casesplit}
    \case*[In \ref{op:garbage}]
      We have
      $
        w_1 \rf->
        (r \of \R{n}{nxt}{\tup{1,\wtv}})
        \po->
        \var{fl}.
      $
      By \cref{lm:wrfl-tso}, $w_1 \tso-> \var{fl}$ and so
      $w_1 \nvo-> \var{fl}$
      since $\locOf(w_1)=\locOf(\var{fl})=n$.

      Establishing that~$w_1 = u$ would prove this case.
      Assume $w_1 \ne u$.
      We have~$ w_1 \mo/-> u $,
      since $u$~reads~$\tup{0,\wtv}$.
      However~$ u \mo-> w_1 $ is also impossible since the only events
      generated by the program that write~$\tup{1,\wtv}$ to $\loc[n.nxt]$
      are updates from~$\tup{0,\wtv}$.

    \case*[In \ref{op:delete-ok}]
      We have
      $(u_1 \of \U{c}{nxt}{\tup{0,\wtv}}{\tup{1,\wtv}}) \po-> \var{fl}$.
      Similarly to the previous case we can conclude~$u_1 = u$.
      \qedhere
    \end{casesplit}
  \end{casesplit}
\end{proof}

\begin{theorem}
\label{th:linkfree:rets-persist}
  $ \cidOf(G.\Rets) \subs \dom(\pt) \union \PRd $.
\end{theorem}
\begin{proof}
  Consider an arbitrary~$\ret \in G.\Rets$ and let $c = \cidOf(\ret)$.
  If $ \valOf(\ret) = \p{false} $ then $ c \in \PRd $.
  Otherwise, we have $ \valOf(\ret) = \p{true} $ and we have to show
  $ \pt(c) \ne \bot $.
  We examine two cases:
  \begin{casesplit}
  \case[$ \callOf(c) = \tup{\p{insert}, k, v} $]
    Recall from \cref{def:linkfree:pers-pt}~$
      \pt(c) = \min_{\nvo}\set{e' | (e'\of\W{n}{valid}{1}) \in G.\Persisted}
    $;
    therefore $\E e' \in {G.\Persisted}.(e'\of\W{n}{valid}{1})$
    would imply $ \pt(c) \ne \bot $.
    Since $\ret \in G.\EvOfCid{c}$, there is an event
    $ (u \of \U{\wtv}{nxt}{\wtv}{\tup{0, n}}) \in G.\EvOfCid{c} $
    generated at \ref{lp:insert-ok}, for some freshly generated~$n$.
    By \cref{lm:flush-ins-before-ret},
    there are $w,\var{fl} \in G.E$ with
    $ (w \of \W{n}{valid}{1}) \nvo-> (\var{fl} \of \FL{n}) $.
    By \cref{def:exec},~$\var{fl} \in G.\Persisted$, and so
    $ w \in G.\Persisted $ as required.

  \case[$ \callOf(c) = \tup{\p{delete}, k} $]
    Since $\ret \in G.\EvOfCid{c}$, there is an event
    $ (u \of \U{n}{nxt}{\wtv}{\wtv}) \in G.\EvOfCid{c} $
    generated at~\ref{lp:delete-ok}
    for some~$n$.
    Therefore by \cref{lm:flush-del-before-ret}
    there is
    $ \var{fl} \in G.E $ such that
    $ u \nvo-> (\var{fl} \of \FL{n}) $.
    Since~$\var{fl} \in G.\Persisted$,
    we have $u\in G.\Persisted$ too.
    From \cref{def:linkfree:pers-pt} then
    $ \pt(c) = u $.
    \qedhere
  \end{casesplit}
\end{proof}

\subsubsection{Voidability of voided calls}

\begin{theorem}
\label{th:linkfree:voided-voidabile}
  Let $c \in \dom(\lp) \setminus (\dom(\pt) \union \PRd)$,
  and $\vec{e} \in {\enum[G.E]{\vo}}$.
  If $
    \hres{\lp}[\vec{e}] = \vec{e}' \concat \lp(c) \concat \vec{e}''
  $ then $
      \tup{\callOf(c), r(c)}
  $ is
  \pre \h-voidable,
  where $
    \h = \restr{\histOf[\lp]{r}(\vec{e}'')}{(\dom(\pt) \union \PRd)}
  $.
\end{theorem}

\begin{proof}
  We prove the statement by
  case analysis on $ \tup{\callOf(c), r(c)} $.
  \begin{casesplit}
  \case[$\tup{\p{insert}, k,v, \p{false}}$]
    The call is \pre h-voidable for every~$h$
    since it does not alter the abstract state.

  \case[$\tup{\p{delete}, k, \p{false}}$]
    The call is \pre h-voidable for every~$h$
    since it does not alter the abstract state.

  \case[$\tup{\p{insert}, k,v, \p{true}}$]
We show that if we assume the call $\tup{\callOf(c), r(c)}$
    is not \pre h-voidable,
    we would contradict the assumption that $c \notin \dom(\pt) \union \PRd$.
    By \cref{lm:kvs-voidable},
    we only need to rule out the cases where
    $\tup{\callOf(c), r(c)} \concat h$ is legal but
    $h$ contains an operation on~$k$.
    Let $ \tup{c',v'} $ be
    the first such occurrence in~$h$.
    By construction, $\tup{c',v'}$ has executed its linearization point,
    and has been persisted.
    We consider the two legal cases:
    \begin{casesplit}
    \case[$\tup{c',v'}=\tup{\p{insert}, k,\wtv, \p{false}}$]
      By legality of the sequence,
      the node~$n$ inserted by~$c$ is the node found by~$c'$.
      Since~$c' \in \PRd$, it returned, and therefore
      it executed \ref{op:flush-ins}$(n)$.
      We know~$n$ was freshly allocated by~$c$ so it is not written to by initial or recovery events.
      We can therefore apply \cref{lm:flush-ins-before-ret} and get
      events
      \[
      \smash{
        (w \of \W{n}{valid}{1})
        \nvo->
        (\var{fl} \of \FL{n})
        \nvo->
        (e_1 \of \W{n}{insFl}{1}).
      }
      \]
      Since $\var{fl} \in G.\Persisted$, $w \in G.\Persisted$.
      By definition~$ \pt(c) $ will either be~$w$ or some
      event with the same label that \nvo-precedes~$w$
      (and thus that would be in $G.\Persisted$ too),
      leading to a contradiction.

    \case[$\tup{c',v'}=\tup{\p{delete}, k, \p{true}}$]
      By legality of the sequence,
      the node~$n$ inserted by~$c$ is the node deleted by~$c'$.
      Since the persisted threshold of~$c'$ is defined,
      we have $\lp(c') = \pt(c') \in G.\Persisted$.
      The event~$\lp(c')$
      (generated at~\ref{op:delete-ok})
      is \po-preceded by a call to \ref{op:makevalid}$(n)$.
      By \cref{lm:makevalid-irrev} we get
      \[
        (e \of \W{n}{valid}{1})
        \tso->
        (\lp(c') \of \U{n}{nxt}{\wtv}{\wtv})
      \]
      which implies
      $ e \nvo-> \lp(c') = \pt(c') \in G.\Persisted $.
      Therefore, $\pt(c)$ must be either~$e$ or some
      event with the same label that \nvo-precedes~$w$
      (and thus that would be in $G.\Persisted$ too),
      leading to a contradiction.
    \end{casesplit}

  \case[$\tup{\p{delete}, k, \p{true}}$]
    We show that if we assume the call $\tup{\callOf(c), r(c)}$
    is not \pre h-voidable,
    we would contradict the assumption that $c \notin \dom(\pt) \union \PRd$.
    By \cref{lm:kvs-voidable},
    we only need to rule out the cases where
    $\tup{\callOf(c), r(c)} \concat h$ is legal but
    $h$ contains an operation on~$k$.
    Let $ \tup{c',v'} $ be
    the first such occurrence in~$h$.
    By construction, $\tup{c',v'}$ has executed its linearization point,
    and has been persisted.
    Consider the two legal cases:
    \begin{casesplit}
    \case[$\tup{c',v'}=\tup{\p{delete}, k, \p{false}}$]
      The call~$c'$ is placed in~$h$ by the hindsight resolution~$\hres{\lp}$.
      Let~$i$ be such that $ \lp(c) = \vec{e}(i) $, and
          $n = \addrOf(\lp(c))$ be the removed node.
      From \cref{th:linkfree:lp-validates} we know that
      $\reach(\mem{\upto{\vec{e}}{i}}, n)$ holds.
      Now consider the~$j_1$ and~$j_2$
      identified in \cref{lm:linkfree:hind-delete},
      \ie the indexes in $\vec{e}$ of \ref{ev:find:rp} and \ref{lp:find}
      in the call of $\ref{op:find}(k,p,c)$ in \ref{op:delete-no}.
      We have two cases: either $\reach(\mem{\upto{\vec{e}}{j'}}, n)$ is false
      for some $j' < j_2$, or $\ref{op:find}(k,p,c)$ would traverse it,
      finding it marked.
      In both cases \ref{ev:trim-w-nxt} of a call to \ref{op:trim}$(\wtv,n)$
      would have been executed.
      Since \ref{ev:trim-w-nxt} is \po-preceded by a call to
      \ref{op:flush-del}$(n)$,
      we obtain
      $
        \lp(c) \nvo-> (\var{fl} \of \FL{n}) \in G.E.
      $
            Since~$\var{fl} \in G.\Persisted$,
      $ \lp(c) \in G.\Persisted $ and so $ \pt(c) = \lp(c) $.

    \case[$\tup{c',v'}=\tup{\p{insert}, k,\wtv, \p{true}}$]
      Similar to the previous case, we can show \ref{op:trim} is called
      on the deleted node before the linearization point of $c'$ is executed,
      which implies that $\lp(c)$ is flushed.
    \end{casesplit}
  \end{casesplit}
\end{proof}

\subsubsection{Ordering of Non-Commuting Calls}

\begin{lemma}
\label{lm:linkfree:non-commut-ghb}
  Let~$c, c'\in \dom(\pt)$ be two calls such that
  $
    \tup{\callOf(c), r(c)}
      \not\comm{\AbsState}{\Delta}
    \tup{\callOf(c'), r(c')}
  $.
  Then $
      \lp(c) \ghb-> \lp(c')
      \lor
      \lp(c') \ghb-> \lp(c).
    $
\end{lemma}
\begin{proof}
  Let $\vec{e} \in \enum[G.E]{\ghb}$.
  By \cref{th:linkfree:lp-validates} we know the sequence of linearization points
  in $\vec{e}$ is legal.
  We know from the definition of $\KVS'$ that $c,c'\in \dom(\pt)$
  do not commute if they are (successful) calls on the same key.
  By transitivity of~$\ghb$ we only need to prove the claim for
  contiguous calls on the same key~$k$,
  \ie calls $c,c' \in \dom(\pt)$ such that
  $ \vec{e}(i) = \lp(c) $ and $ \vec{e}(j) = \lp(c') $ implies that
  $ i<j $ and
  if $ \vec{e}(i') = \lp(c'') $ for some $i<i'<j$ and $c''\in \dom(\pt)$, then $c''$ is not a call on~$k$.
  This leaves us with two cases:
  \begin{casesplit}
  \case[$\tup{c ,v }=\tup{\p{insert}, k, \p{true}}$,
        $\tup{c',v'}=\tup{\p{delete}, k, \p{true}}$]
    Since the two calls are contiguous, we know the node~$n$ added by~$\lp(c)$
    is the one marked for deletion by~$\lp(c')$.
    Since~$n$ is initialized at~$ \lp(c') $,
    we get $\lp(c) \ghb-> \lp(c')$.
  \case[$\tup{c ,v }=\tup{\p{delete}, k, \p{true}}$,
        $\tup{c',v'}=\tup{\p{insert}, k, \p{true}}$]
    Let $n$ be the node marked by $\lp(c)$.
    As argued in the proof of \cref{th:linkfree:voided-voidabile},
    \ref{op:trim}$(\wtv, n)$ is called before~$\lp(c')$ which implies
    the update to the \p{nxt} field of the node deleted by~$c$
    is read \ghb-before~$\lp(c')$.
    \qedhere
  \end{casesplit}
\end{proof}

\begin{theorem}
\label{th:linkfree:commuting-calls}
  For any~$c, c'\in \dom(\pt)$, either:
  \begin{itemize}
  \item $
      \tup{\callOf(c), r(c)}
        \comm{\AbsState}{\Delta}
      \tup{\callOf(c'), r(c')}
    $, or
  \item $
      \pt(c) \nvo-> \pt(c')
      \implies
      \lp(c) \ghb-> \lp(c').
    $
  \end{itemize}
\end{theorem}
\begin{proof}
  We prove that for all non-commuting~$c,c' \in \dom(\pt)$,
  $
    \lp(c) \ghb-> \lp(c')
    \implies
    \pt(c) \nvo-> \pt(c')
  $.
  This, by totality of $\nvo$ and \cref{lm:linkfree:non-commut-ghb} implies
  our claim.
  By transitivity and legality of the volatile linearization,
  we only need to consider legal non-commuting contiguous pairs of calls.
  \begin{casesplit}
  \case[$\tup{c ,v }=\tup{\p{insert}, k, \p{true}}$,
        $\tup{c',v'}=\tup{\p{delete}, k, \p{true}}$]
    Since the two calls are contiguous, we know the node~$n$ added by~$\lp(c)$
    is the one marked for deletion by~$\lp(c')$.
    Since a call to \ref{op:makevalid}$(n)$ \po-precedes $\lp(c')$,
    Therefore, by \cref{lm:makevalid-irrev}, there is some event~$w$
    with
    $ (w:\W{n}{valid}{1}) \tso-> \lp(c') $.
    This implies $ \pt(c) \nvo?-> w \nvo-> \lp(c') = \pt(c') $.
  \case[$\tup{c ,v }=\tup{\p{delete}, k, \p{true}}$,
        $\tup{c',v'}=\tup{\p{insert}, k, \p{true}}$]
    Let $n$ be the node marked by $\lp(c)$.
    As argued in the proof of \cref{th:linkfree:voided-voidabile},
    \ref{op:trim}$(\wtv, n)$ is called before~$\lp(c') \tso-> \pt(c')$
    (by \cref{lm:makevalid-irrev}) which implies
    $\Updates_n \ni \pt(c) = \lp(c) \nvo-> (\var{fl} \of \FL{n}) \nvo-> \pt(c')$.
    \qedhere
  \end{casesplit}
\end{proof}

\subsubsection{Correctness of persisted state}

\begin{lemma}
\label{lm:newnode-nvo-valid}
  If~$(w \of \W{n}{valid}{1}) \in G.E \setminus E_0$,
  then a call to $\ref{op:newnode}(n)$ has been executed
  and all its events~$e$ are such that
  $ e \nvo-> w $.
\end{lemma}
\begin{proof}
  The event~$w$ must have been generated by a call to \ref{op:makevalid}$(n)$.
  If $w$ is \po-after \ref{op:newnode}$(n,\wtv[3])$
  then, since all the events involved are writes on the same address~$n$,
  $w$ is also \nvo-after the \ref{op:newnode}$(n,\wtv[3])$ events.
  In all the remaining cases,
  the \ref{op:makevalid}$(n)$ generating~$w$ is called
  \po-after a read $(r \of \R{n}{nxt}{\wtv})$.
  Since~$n$ must be initialised at~$r$, there is a~$\rfe$ edge
  from the write~$w_0$ to~$\loc[n.nxt]$ in \ref{op:newnode}$(n,\wtv[3])$
  and $w$.
  This means that $w_0 \tso-> w$, which implies $w_0 \nvo-> w$.
\end{proof}

\begin{lemma}
\label{lm:mark-nvo-delfl}
  If $ (w \of \W{n}{delFl}{1}) \in G.E \setminus E_0$,
  then there is a $ (u \of \U{n}{nxt}{\tup{0,\wtv}}{\tup{1,\wtv}}) \in G.E $
  with $ u \nvo-> w $,
  or a $w_0 \in G.E_0 \inters \UWrites$ with $\wvalOf(w_0) = \tup{1,\wtv}$
  and $ w_0 \nvo-> w $.
\end{lemma}
\begin{proof}
  The event~$w$ must have been generated by a call to \ref{op:flush-del}$(n)$,
  which is only called by \ref{op:trim}$(\wtv,n)$.
  Each such call is \po-after either some read $(r \of \R{n}{nxt}{\tup{1,\wtv}})$,
  or an update $ (u \of \U{n}{nxt}{\tup{0,\wtv}}{\tup{1,\wtv}}) $.
  In either case, there is either an update like~$u$ or a write~$w_0 \in E_0$
  with $\wvalOf(w_0) = \tup{1,\wtv}$, that are \tso-before $w$.
  Since these are all writes to the same locations,
  this implies $u \nvo-> w$ or $w_0 \nvo-> w$ respectively.
\end{proof}

\begin{lemma}
  $\tup{\pt, r, \restr{\nvo}{G.\Persisted}, \durable}$
  is a linearization strategy.
\end{lemma}
\begin{proof}
  Straightforward by inspecting \cref{def:linkfree:pers-pt}.
\end{proof}

\begin{lemma}
\label{lm:mark-irrev-nvo}
  Let $ \vec{e} = \enum{\nvo} $.
  If $ \mem{\upto{\vec{e}}{i}}(\loc[x.nxt]) = \tup{1,y} $ then
     $ \mem{\upto{\vec{e}}{j}}(\loc[x.nxt]) = \tup{1,y} $.
\end{lemma}
\begin{proof}
  Immediate by observing that all the events affecting \loc[x.nxt] are updates
  that either keep a node unmarked or mark an unmarked node,
  and that the $\rf$ edges between them are preserved by~$\nvo$.
\end{proof}

\begin{theorem}
\label{th:linkfree:pers-state}
  $
    \histOf[\pt]{r}(\enum[G.\Persisted]{\nvo})
    \in \LegalFrom[\KVS',\Delta]{\kvs_0}
      \implies
        \tup{\pt, r, \restr{\nvo}{G.\Persisted}, \durable}
  $ \pre\initOf-validates~$G$.
\end{theorem}
\begin{proof}
Let
  $ \vec{e} = \enum[G.\Persisted]{\nvo} $, and
  $M_i \is \mem{\upto{\vec{e}}{i}}$.
Also let
  $ M_{\initOf(\vec{e})} \in \durable(\kvs_0) $,
  $ \initOf(\vec{e}) < i < \len{\vec{e}} $, and
  $ \kvs \in \KVS' $.
Assume:
\begin{align*}
  \histOf[\pt]{r}(\vec{e})
    &\in \LegalFrom[\KVS',\Delta]{\kvs_0}
  &
  \histOf[\pt]{r}(\upto{\vec{e}}{i-1})
    &\in \LegalPath[\KVS',\Delta]{\kvs_0}{\kvs}
  &
  M_{i-1}
    &\in \durable(\kvs)
\end{align*}
From the last assumption we let
$
  M_{i-1} =
    \durM_{\lbl{s}} \dunion
    \durM_{\lbl{d}} \dunion
    \durM_{\lbl{g}}
$
as per \cref{def:linkfree:recoverable}.
We then check \cref{cond:stutter,cond:linpt-trans}
by a case analysis on~$\vec{e}(i)$.
\begin{defenum}
  \item If $ \vec{e}(i) \ne \pt(c) $ for all~$c$, we have to prove
    $ M_{i} \in \durable(\kvs) $.
    The claim is trivial for recovery events
    by the assumption of soundness of the recovery.
    We consider each possible write or update.
    \begin{casesplit}
    \case[$ (\vec{e}(i) \of \Alloc{x} $]
      Then all the fields are zero and the node can be added to $\durM_{\lbl{g}}$ yielding $M_i \in \durable(\kvs)$.
    \case[$ \vec{e}(i) \in \Writes $ generated by $\ref{op:newnode}(x)$]
      Then, by \cref{lm:newnode-nvo-valid} we have
      $M_{i-1}(\loc[x.valid])=0$, which means $x$ belongs to $\durM_{\lbl{g}}$, where the fields \p{key}, \p{val} and \p{nxt}
      are not constrained.
      The node~$x$ can then belong to the garbage nodes of $M_{i}$ too.
    \case[$ (\vec{e}(i) \of \W{n}{valid}{1}) $]
      Then $ \vec{e}(i) $ cannot be the first such write in $\mo$ order
      or it would be the persistency point of some call.
      Then since there is some other $ (w \of \W{n}{valid}{1}) \mo-> \vec{e}(i) $
      we have $ w \nvo-> \vec{e}(i) $ and so by \cref{lm:makevalid-irrev},
      $M_{i-1}(\loc[n.valid]) = 1$ and so $ M_i = M_{i-1} $.
    \case[$ (\vec{e}(i) \of \W{n}{delFl}{1}) $]
      By \cref{lm:mark-nvo-delfl} and \cref{lm:mark-irrev-nvo}
      we have $M_{i-1}(\loc[n.nxt])=\tup{1,\wtv}$.
      This means $n$ belongs to~$ \durM_{\lbl{d}} $,
      which does not constrain the \p{delFl} field.
      We can thus keep the node in the deleted nodes in $M_{i-1}$.
    \case[$ (\vec{e}(i) \of \W{n}{insFl}{1}) $]
      Irrelevant as $\durable$ does not constrain the \p{insFL} field.
    \case[$ (\vec{e}(i) \of \U{p}{nxt}{\tup{b,n}}{\tup{b',n'}}) $]
      For this not to be a persistency point we have $b=b'=0$.
      The update is then irrelevant as $\durable$ does not constrain
      the marking bit of the \p{nxt} field.
    \end{casesplit}
\item If $ \vec{e}(i) = \pt(c) $ for some~$c$, we have to prove
    $ \mem{\upto{\vec{e}}{i}} \in \durable(\kvs') $
    for some~$\kvs'$ such that
    $ (\kvs,\kvs') \in \Delta(\callOf(c), r(c)) $.\begin{casesplit}
    \case[$ \pt(c) = (\vec{e}(i) \of \W{n}{valid}{1}) $]
      Then the call is a successful insert of some~$k$ and~$v$.
      From legality of the whole sequence we must have $k \notin \dom(\kvs) $.
      Since \ref{op:insert-ok} calls \ref{op:newnode} first,
      by \ref{lm:newnode-nvo-valid} we have
      $ M_{i-1}(\loc[n.key]) = k $, and
      $ M_{i-1}(\loc[n.val]) = v $.
      Since $\pt(c)$ is defined as the \mo-first event
      writing 1 to $\loc[n.valid]$,
      we have $ M_{i-1}(\loc[n.valid]) = 0 $.
      We therefore have that $n$ belongs to $\durM_{\lbl{g}}$,
      making $n$ fresh in $\kvs$.
      Note that~$n$ can only be marked at
      \ref{lp:delete-ok} in \ref{op:delete-ok},
      which is \po-preceded by a call to $\ref{op:makevalid}(n)$.
      Therefore, by \cref{lm:makevalid-irrev}
      we also have $ M_{i-1}(\loc[n.nxt]) = \tup{0,\wtv} $.
      From this also follows that, by \cref{lm:mark-nvo-delfl},
      $ M_{i-1}(\loc[n.delFl]) = 0 $.
      The node~$n$ can therefore be transferred from $\durM_{\lbl{g}}$
      to $\durM_{\lbl{s}}$ in $M_i$,
      inducing the desired state change
      from $\kvs$ to $\kvs \dunion \tup{k, n, v}$.

    \case[$ \pt(c) = (\vec{e}(i) \of \U{n}{nxt}{\tup{0,\wtv}}{\tup{1,\wtv}}) $]
      Then the call is a successful delete of some~$k$.
      The persistency point is \po-preceded by some
      $(r \of \R{n}{key}{k})$
      and a call to \ref{op:makevalid}$(n)$.
      So by~\cref{lm:newnode-nvo-valid,lm:init-key-val} we have
      $M_{i-1}(\loc[n.key]) = k$.
      Moreover,
      by noting $ \rf\seq\ev{\Updates} \subs \nvo $ and $ \mo \subs \nvo $,
      we have $ M_{i-1}(\loc[n.nxt]) = \rvalOf(\vec{e}(i)) = \tup{0,\wtv} $.
      By legality of the whole call sequence we have $ k \in \kvs $,
      and therefore, since it is valid,
      $n$ belongs to $\durM_{\lbl{s}}$ and $ \kvs(k) = \tup{n,v} $
      for some~$v$.
      By marking~$n$, $\vec{e}(i)$ has the effect of moving
      $n$ to $\durM_{\lbl{d}}$,
      inducing the desired change in the state from
      $\kvs$ to $\kvs \setminus \set{\tup{k,n,v}}$.
    \end{casesplit}
\end{defenum}
\end{proof}

\subsection{Soundness of Recovery}
\label{sec:recovery-sound}

The verification of the recovery (in \cref{fig:lf-set-code}) is routine:
since it is sequential, no weak behaviour is observable,
and the verification can be equivalently done under SC.
It is easy to check that,
from a memory~$M \in \durable(\kvs)$ the recovery procedure will
produce a memory
$M' \in \durable(\kvs) \inters \volatile(\kvs) = \recovered(\kvs)$:
it only affects links, so it preserves $\durable$ and it produces
a sorted list as required by $\volatile$.

\begin{lemma}
\label{lm:linkfree:recovery-sound}
  The recovery implementation $\LinkFreeImpl[rec]$
  is \pre\tup{\durable,\recovered}-sound.
\end{lemma}

\subsection{Putting It All Together}

\begin{theorem}[Durable linearizability of \LinkFree]
  The link-free store implementation
  $ \tup{\LinkFreeImpl[op],\LinkFreeImpl[rec]} $
  is durably linearizable.
\end{theorem}
\begin{proof}
  By \cref{lm:linkfree:recovery-sound} the recovery is \pre\tup{\durable,\recovered}-sound.
  By \cref{th:rec-decoupling}
  we are left to prove
  $\LibImpl[op]$ is \pre\tup{\durable,\recovered}-linearizable.
  We do so by showing all the conditions of the \masterthm\ are satisfied.
  \begin{defenum}[itemsep=.4\baselineskip]
    \item $ \dom(\pt) \subs \dom(\lp). $
      \\ By \cref{def:linkfree:pers-pt} and \cref{def:linkfree:linpt}.
\item $
        \A \vec{e} \in {\enum[G.E]{\ghb}}.
          \dom(\hres{\lp}(\vec{e})) =
            \PRd \setminus \dom(\lp).
        $
      \\ By \cref{def:linkfree:linpt} and \cref{def:linkfree:hres}.
\item $
          \A c \in \PRd.\;
            \Delta(\callOf(c),r(c)) \subs \idOn{\AbsState}.
        $
\\ By \cref{def:linkfree:pers-pt}.
\item $
        \ev{\UWrites} \seq \hb \seq \ev{\UWrites}
          \subs \ghb
      $
      \\ By \cref{lm:hb-implies-ghb}.
\item $
        \A c \in \dom(\lp).
          \E e_1,e_2 \in {G.\EvOfCid{c}}.
            \smash{e_1 \hb?-> \lp(c) \hb?-> e_2}.
       $
      \\ By \cref{def:linkfree:linpt} we have $\cidOf(\lp(c))=c$ so the condition
      is proven using $e_1=e_2=\lp(c)$.
\item For any~$\vec{e} \in {\enum[G.E]{\ghb}}$ and
          $c\in\dom(\hres{\lp}(\vec{e}))$:
      \begin{itemize}
      \item $
          \E i_1,i_2.
            \vec{e}(i_1), \vec{e}(i_2) \in \UWrites
            \land
            c=\cidOf(\vec{e}(i_1))=\cidOf(\vec{e}(i_2))
            \land
            i_1 \leq \hres{\lp}(\vec{e})(c) \leq i_2.
      $
      \end{itemize}
      By \cref{def:linkfree:hres,lm:linkfree:hind-delete}.
\item $\tup{\lp, \hres{\lp}, r, \ghb, \volatile}$
      \pre\recEndOf-validates~$G$.
      \\ By \cref{th:linkfree:hres-validates}.
\item $
        \cidOf(G.\Rets) \subs \dom(\pt) \union \PRd.
      $
      \\ By \cref{th:linkfree:rets-persist}.
\item For any~$c \in \dom(\lp) \setminus (\dom(\pt) \union \PRd)$,
          and all $\vec{e} \in {\enum[G.E]{\ghb}}$:
      \begin{itemize}
      \item If $
          \hres{\lp}[\vec{e}] = \vec{e}' \concat \lp(c) \concat \vec{e}''
        $ then $
            \tup{\callOf(c), r(c)}
        $ is
        \pre \h-voidable,\\
        where $
          \h = \restr{\histOf[\lp]{r}(\vec{e}'')}{(\dom(\pt) \union \PRd)}
        $.
      \end{itemize}
      By \cref{th:linkfree:voided-voidabile}.
\item For any~$c, c'\in \dom(\pt)$, either:
      \begin{itemize}
      \item $
          \tup{\callOf(c), r(c)}
            \comm{\AbsState}{\Delta}
          \tup{\callOf(c'), r(c')}
        $, or
      \item $
          \smash{\pt(c) \nvo-> \pt(c')}
          \implies
          \smash{\lp(c) \ghb-> \lp(c')}.
        $
      \end{itemize}
      By \cref{th:linkfree:commuting-calls}.
\item $
      \histOf[\pt]{r}(\enum[G.\Persisted]{\nvo})
      \in \LegalFrom[\KVS',\Delta]{\kvs}
        \implies
          \tup{\pt, r, \restr{\nvo}{G.\Persisted}, \durable}
      $ \pre\initOf-validates~$G$.
      \\ By \cref{th:linkfree:pers-state}.
\qedhere
  \end{defenum}
\end{proof}

\subsection{The \p{contains} operation}
\label{sec:contains}

\begin{wrapfigure}[15]{R}{25ex}\begin{sourcecode}[gobble=2,lineskip=-2pt]
  def contains(h, k):
    <|_,c|> = h.nxt
    while(c.key < k):
      <|_,c|> = c.nxt
    if c.key != k:
      return false @      \label{line:contains-ret-f1}@
    if c.nxt == <|1,_|>:
      flushDel(c)
      return false
    makeValid(c)
    flushIns(c)
    return true
\end{sourcecode}
   \vspace*{-.5em}
  \caption{
    A wait-free membership check operation.
  }
  \label{fig:linkfree-contains}
\end{wrapfigure}

The full algorithm includes a wait-free \p{contains} operation,
shown in \cref{fig:linkfree-contains},
traverses the list ignoring whether the nodes are marked or not.
If no node holding~$k$ was found, the operation returns \p{false}.
Otherwise, if the node holding~$k$ is marked, the node is flushed
and the operation again returns \p{false}.
If it was not marked, then the operation helps persisting the insert
and returns \p{true}.

Both the optimization of \p{find} and the \p{contains} operation
introduce the need for hindsight proofs.
Consider a \p{contains($k$)} call returning \p{false}
at line~\ref{line:contains-ret-f1}, for example.
Nowhere in the whole call, the code checks the marking bits of the
nodes it traverses.
This means that it might be traversing nodes that have been logically
removed and be reading the key of a node that is not reachable from
the head any more.
The reason this is correct is that since we start the traversal from the head,
there must have been some point after the traversal started,
where the nodes were reachable.
This kind of ``after the fact'' argument is called
hindsight linearization.

The techniques we used to handle failed \p{delete} operations
can be used to handle $\p{contains}(k)$:
we can prove that, for every $\vec{e} \in \enum[G]{\ghb}$,
between the first read $\vec{e}(i_1)$ of \p{h.nxt}
and the last read $\vec{e}(i_2)$ of a field of \p{c},
there is an index $ i_1 \leq i \leq i_2 $ such that
$ \mem{\upto{\vec{e}}{i}} \in \volatile(\kvs)$
for some $S$ with $ k \notin \dom(\kvs) $.
This provides both the definition of $\hres{\lp}(\vec{e})$ in terms of such~$i$,
and the proof of conditions~\ref{cond:hind-linpt} and \ref{cond:hind-hb}.

\end{document}